\documentclass[10pt]{article}
\usepackage{plain}
\usepackage{amsthm}
\usepackage{amsfonts}
\usepackage{amsmath}
\usepackage{amssymb}
\usepackage{physics}
\usepackage{mathtools}
\usepackage{comment}
\usepackage{url}
\usepackage{tikz}
\usepackage{eurosym}
\usetikzlibrary{arrows,decorations.pathmorphing,backgrounds,positioning,fit,petri,decorations}
\usetikzlibrary{calc,intersections,through,backgrounds,mindmap,patterns,fadings}
\usetikzlibrary{decorations.text}
\usetikzlibrary{decorations.fractals}
\usetikzlibrary{fadings}
\usetikzlibrary{shadings}
\usetikzlibrary{shadows}
\usetikzlibrary{shapes.geometric}
\usetikzlibrary{shapes.callouts}
\usetikzlibrary{shapes.misc}
\usetikzlibrary{spy}
\usetikzlibrary{topaths}

\newcommand{\beq}{\begin{equation}}
\newcommand{\eeq}{\end{equation}}
\newcommand{\beqnl}{\begin{align}}
\newcommand{\eeqnl}{\end{align}}

\newcommand{\be}{\begin{equation}}
\newcommand{\ee}{\end{equation}}
\newcommand{\bes}{\begin{equation*}}
\newcommand{\ees}{\end{equation*}}
\newcommand\veclungo[1]{\overrightarrow{#1}}
\newcommand\cev[1]{\overleftarrow{#1}}

\newtheorem{Theorem}{Theorem}
\newtheorem{Definition}{Definition}
\newtheorem{Lemma}{Lemma}
\newtheorem{Example}{Example}
\begin{document}
\begin{titlepage}
\title{\textbf{Co-homology of Differential Forms and Feynman diagrams}}
\maketitle
\begin{center}
\author{Sergio L. Cacciatori $^{1,2}$ \and Maria Conti $^{1,2}$ \and Simone Trevisan $^{1,2}$}
\linebreak
\linebreak
\textsl{$^{1}$ \quad Dipartimento di Scienza ed Alta Tecnologia, Universit\`a degli Studi dell'Insubria;\\
$^{2}$ \quad INFN, Sezione di Milano, Via Celoria 16, 20133 Milano, Italy}
\abstract{In the present review we provide an extensive analysis of the intertwinement between Feynman integrals and cohomology theories in the light of the recent developments. Feynman integrals enter in several perturbative methods for solving
non linear PDE, starting from Quantum Field Theories and including General Relativity and Condensed Matter Physics. Precision calculations involve several loop integrals, onec strategy to address which is to bring them back in terms of linear 
combinations of a complete set of integrals (the Master Integrals). In this sense Feynman integrals can be thought as defining a sort of vector space to be decomposed in term of a basis. Such a task may be simpler if the vector space is endowed with a 
scalar product.
Recently, it has been discovered that, interpreting these spaces in terms of twisted cohomology, the role of a scalar product is played by intersection products. The present review is meant to provide the mathematical tools, usually familiar to mathematicians
but often not in the standard baggage of physicists, like singular, simplicial and intersection (co)homologies, hodge structures etc., apt to restate this strategy on precise mathematical grounds. It is thought to be both an introduction for beginners interested to 
the topic, as well as a general reference providing helpful tools for tackling the several still open problems. }
\end{center}
\end{titlepage}

\section{Introduction}

Feynman diagrams are introduced in the context of quantum interacting field theory, as a graphical representation of the solution of a system of first order differential equations, admitting a path-ordered exponential expression. Usually, the matrix of the system is composed of two terms: 
one identifying the solution in absence of interactions, {\it i.e.} the free solution; and a second one, carrying information on the interaction, treated as a perturbation to the free evolution, and characterized by the strength of the interaction, {\it i.e.} the coupling constant, considered as a small quantity. 
The perturbative expansion of the path-ordered exponential, obtained by a series expansion in the coupling constant, gives rise to the Dyson series, containing an infinite sequence of iterated integrals,
whose iteration number increases with the perturbative order: at any given order, hence for any given power of the coupling constant the integrands are formed by an ordered product of functions 
- to better say, distributions -, representing the free evolution (the propagators), and the insertion of interaction terms (the vertices). 

Dyson series can be used to describe the evolution of physical systems whose dynamics follows Volterra-type model, within quantum as well as classical physics. Therefore, the predictive power of a theoretical model, aiming at describing the dynamics of physical systems, on a wide spectrum of physical scales, being them either so microscopic as colliding elementary particles, or as macroscopic as coalescing astrophysical binary systems, may
depend on our ability of evaluating Feynman integrals, alias solving systems of differential equations. 

Hamiltonian and Lagrangian carry information about the free and the interactive dynamics, and the basic rules to build Feynman graphs can systematically be derived from them. In particular, the interaction between two elementary entities, experiencing the {\it presence} of each other through the mediation of a third entity, can be described like a scattering event. Therefore, quantities like the impact parameter, the cross section, 
the scattering angle, or the interaction potential turn out to be related to the scattering amplitude, which ultimately admits a representation in terms of Feynman graphs. 

"Perturbation theory means Feynman diagrams" \cite{Veltman:1994wz}, yet the diagrammatic approach is not limited to the perturbative regime:
"Perturbation theory is a very useful device to discover very useful equations and properties that may hold true even if the perturbation expansion fails" \cite{tHooft:1973wag}. 

More modern approaches based on analyticity and unitarity, so called on-shell and unitarity-based methods \cite{Bern:1994cg,Bern:1994zx,Britto:2004nc,Britto:2005ha}, make use of the factorization 
properties of scattering amplitudes (exposed by using complex variables, to build suitable combinations of energy and momenta of the interacting objects) in order to group more efficiently the contributing Feynman diagrams, and exploit recursive patterns, hard to identify within the pure diagrammatic approach. In this case, the symmetries, which do not necessarily hold for the individual diagrams, and which are inherited from the lower-order amplitudes, yield novel representations of the scattering amplitudes (see f.i. \cite{Cachazo:2013gna,Bern:2019prr}).

Therefore, scattering amplitudes, at any given order in perturbation theory, can be canonically built out of linear combination of Feynman graphs, and equivalently out of products (convolutions) of lower order amplitudes. 
Independently of the strategy adopted for their generation, the evaluation of scattering amplitudes beyond the tree-level approximation requires the evaluation of multivariate Feynman integrals.

Dimensional regularization played a crucial role in the formal mathematical developments of gauge theories and of Feynman integrals. 
Exploiting the analytic continuation in the space-time dimensions $d$ of the interacting fields, it is possible to modify the number of integration variables in order to {\it stabilize} otherwise ill-defined (mathematically non existing) 
integrals emerging in the evaluation of quantities which ultimately have to be compared with numbers coming from (physically existing) experiments.

Within the dimensional regularization scheme, Feynman integrals, are not independent functions.
They obey relations that can be established at the integrand level, namely among the integrands related to different graphs, systematized in the so called integrand decomposition method for scattering amplitudes~\cite{Ossola:2006us,Ellis:2007br,Ellis:2008ir,Mastrolia:2012bu,Zhang:2012ce,Mastrolia:2012an,Mastrolia:2011pr,Badger:2013gxa}, as well as relation that hold, instead, just upon integration. The latter are 
contiguity relations known as integration-by-parts (IBP) identities \cite{Chetyrkin:1981qh}, 
which play a crucial role in the evaluation of scattering amplitudes beyond the tree-level approximation. 
Process by process, IBP identities yield the identification of an elementary set of integrals, the so-called {\it master integrals} (MIs), which can be used as a basis for the decomposition of multi-loop amplitudes \cite{Laporta:1996mq}. 
MIs are special integrals, as elementary Feynman integrals that admit a graphical representation (in terms of products of scalar propagators and scalar interaction vertices).
At the same time, IBP relations can be used to derive differential equations \cite{Barucchi:1973zm,KOTIKOV1991158,Bern:1993kr,Remiddi:1997ny,Gehrmann:1999as,Henn:2013pwa,Argeri:2014qva,Adams:2017tga}, 
finite difference equations \cite{Laporta:2001dd,Laporta:2003jz}, 
and dimensional recurrence relations \cite{Tarasov:1996br,Lee:2009dh} 
obeyed by MIs. The solutions of those equations are valuable methods for the evaluation of MIs, for those cases where their direct integration might turn out to be prohibitive
(see f.i. \cite{Argeri:2007up,Henn:2014qga,Kalmykov:2020cqz}).

The study of Feynman integrals, the systems of differential equations they obey, the iterated integral representation of their solution
~\cite{Chen:1977oja,Goncharov:1998kja,Remiddi:1999ew}
has been stimulating a vivid interplay and renovated interests between field theoretical concepts and formal mathematical ideas in Combinatorics, Number Theory, Differential and Algebraic Geometry, and Topology 
(see, f.i. \cite{Broadhurst:1996kc,Broadhurst:2000em,Bloch:2005bh,Bogner:2007mn,Brown:2009ta,Marcolli:2009zy,Arkani-Hamed:2016byb,Arkani-Hamed:2017tmz,Mizera:2017cqs,Mizera:2017rqa,Broedel:2018qkq,Henn:2020lye,Sturmfels:2020mpv}).

The geometric origin of the analytic properties of Feynman integrals finds its roots in the application of topology to the S-matrix theory
\cite{hwa1966homology,Pham:1965zz,Lefschetz:1975ta}. 
In more recent studies, co-homology played an important role for identifying relations among Feynman integrals and to expose deeper properties of scattering amplitudes 
\cite{Bloch:2005bh,Brown:2009ta,Lee:2013hzt,Broadhurst:2016hbq,Broadhurst:2016myo,Broadhurst:2018tey,Mizera:2017cqs,Mizera:2017rqa,Abreu:2019wzk,Abreu:2019xep,Mastrolia:2018uzb,Mizera:2019gea,Frellesvig:2019kgj,Mizera:2019vvs,Frellesvig:2019uqt,Mizera:2020wdt,Frellesvig:2020qot,Weinzierl:2020xyy,Kaderli:2019dny,Kalyanapuram:2020vil,Weinzierl:2020nhw,fresn2020quadratic,fresn2020quadratic2,Chen:2020uyk,Britto:2021prf}

In this editorial, we elaborate on the recently understood vector space structure of Feynman integrals 
\cite{Mastrolia:2018uzb,Frellesvig:2019kgj,Mizera:2019gea,Mizera:2019vvs,Frellesvig:2019uqt,Mizera:2020wdt,Frellesvig:2020qot,Weinzierl:2020xyy} 
and the role played by the intersection theory for twisted de Rham (co)-homology to access it. 

As observed in ref. \cite{Mastrolia:2018uzb}, after looking at Feynman integrals as generating a vector space, one can se {\it intersection numbers} of differential forms \cite{cho1995} as a sort of scalar product over it. From this viewpoint, the intersection
products with a basis of MIs mimic the projection of a vector into a basis. For example, using intersection projections for 1-forms, applied to integral representations of the Lauricella $F_{D}$ functions allowed to easily re-derive continuity relations for 
such functions and the decomposition in terms of MIs for those Feynman integrals on maximal cuts that admit a representation as a one-fold integral \cite{Mastrolia:2018uzb,Frellesvig:2019kgj}.
For more general cases, when one has to deal with multifold integral representations \cite{Frellesvig:2019kgj,Frellesvig:2019uqt,Frellesvig:2020qot}, it have been introduced the {\it multivariate} intersection 
numbers  \cite{matsumoto1994,matsumoto1998,OST2003,doi:10.1142/S0129167X13500948,goto2015,goto2015b,Yoshiaki-GOTO2015203,Mizera:2017rqa,matsubaraheo2019algorithm}.
For the case of meromorphic $n$-forms, an iterative method for the determination of intersection numbers was proposed in \cite{Mizera:2019gea} and 
successively refined in \cite{Frellesvig:2019uqt,Frellesvig:2020qot,Weinzierl:2020xyy}. The only simple case is for logarithmic (dlog) differential forms, which bring simple poles only, whose intersection numbers can be computed by employing the global 
residue theorem \cite{Weinzierl:2020xyy}. 

Within this approach, the number of MIs, proven to be finite \cite{Smirnov:2010hn}, is the dimension of the vector space of Feynman integrals \cite{Frellesvig:2019uqt}, and corresponds to the dimension of the homology groups \cite{Lee:2013hzt}, or equivalently of the cohomology group~\cite{Mastrolia:2018uzb,Frellesvig:2019kgj,Frellesvig:2019uqt,Frellesvig:2020qot}, 
and can be related to topological quantities such as the number of critical points \cite{Lee:2013hzt}, Euler characteristics \cite{Aluffi2009FeynmanMA,Aluffi:2011,Bitoun:2017nre,Bitoun:2018afx}, as well as to the dimension of quotient rings of polynomials, for zero dimensional ideals, in the context of computational algebraic geometry~\cite{Frellesvig:2020qot}.

Another interesting consequence of intersection number is about their underlying geometrical meaning, which leads to determining linear relations, equivalent to IBP relations, and quadratic relations for Feynman integrals 
\cite{Mastrolia:2018uzb,Frellesvig:2019kgj,Frellesvig:2019uqt,Frellesvig:2020qot}, colled {\it twisted Riemann periods relations} (TRPR) \cite{cho1995} since represent a twisted version of the well-known bilinear Riemann relations. 
Some of such quadratic relations were already noticed by using number-theoretic methods to Feynman calculus, and given rise to conjectures \cite{Broadhurst:2016hbq, Broadhurst:2016myo,Broadhurst:2018tey,Zhou:2017vhw,Zhou:2017jnm}, 
whose proof has been given only recently \cite{fresn2020quadratic,fresn2020quadratic2}, while other bilinear relations, proposed in \cite{Lee:2018jsw}, have yet to be understood in the light of the TRPR. 

As stated in advance, in this work we first concentrate on the basic aspects leading to the definition of a vector space of Feynman integrals; the we move onto the mathematical description ( starting from an elementary point of view) of the instruments needed in order to tackle the geometrization program just described. Finally, we devote some time to address the problem of the actual computation of intersection numbers. More precisely, in Section \ref{secbaikov} we mostly describe the Baikov representation of Feynman integrals and its role in uncovering the underlying cohomological structure, while in Section \ref{secnumberMI} we consider the vector structure of Feynman integrals - including bilinear identities - and also precisely describe how the number of MIs can be computed. In Section \ref{seccohomology} we provide an elementary illustration of cohomologies, while in Section \ref{HardMaths} we highlight the link between cohomology theory and integration theory. In Section \ref{veryhardmaths} we give an extensive lookout at the advanced mathematical constructions behind. Finally, in Section \ref{secinternumb} we make some explicit examples of practical techniques adopted to compute intersection numbers. The three appendices include some technical details.

\section{Feynman integral representation}\label{secbaikov}

We aim at describing the properties of (scalar) Feynman integrals, representing the most general type of integrals appearing 
in the evaluation of Scattering Amplitudes, left over after carrying out the spinor and the Lorentz algebra (spinor-helicity decomposition, Dirac-Clifford gamma algebra, form factor decomposition), generically indicated as,

\begin{equation}
I_{\nu_1,\cdots,\nu_N}=\int\prod_{i=1}^L \frac{d^Dq_i}{\pi^{D/2}}\prod_{a=1}^N \frac{1}{D_a^{\nu_a}}\label{feynman}\, .
\end{equation}

In the classical literature, the evaluation of Feynman integrals is carried out by direct integration, in position and/or momentum-space representation, making use of Feynman, or equivalently Schwinger, parameters. More advanced methods make use of differential equations, as an alternative computational strategy, which turns out to be very useful whenever the direct integration becomes prohibitive, for instance, due to the number of the physical scales in the scattering reaction (number of particles and/or masses of particles). 
In this work, we would like to approach the multi-loop Feynman calculus in a different fashion from the direct integration,
making use of novel properties that emerge when Feynman integrals are cast in suitable parametric representation, such as 
the so called \textsl{Baikov representation} \cite{Baikov:1996iu} (see also \cite{Grozin:2011mt} for review).
First, we observe that the integration variables involved in the integral \eqref{feynman} are the usual $L$ loop momenta $q_i$, which are not Lorentz invariants. Baikov representation consists in a change of variables in which the new integration variables are actually Lorentz invariants: namely, the independent scalar products one can build using the $L$ loop momenta $q_i$ and $E$ independent external momenta $p_j$. Using these ideas, one can put the Feynman Integrals in the following form, called Baikov representation:
\begin{equation}
I_{\nu_1,\cdots,\nu_M}=K\int_\Gamma B^\gamma \prod_{a=1}^M
\frac{dz_a}{z_a^{\nu_a}}\,.\label{baikov}
\end{equation}
For a proof, see Appendix \ref{AppBaikov}.

\

Before getting more into depth of the meaning behind Eq.~\eqref{baikov}, an observation is necessary. By comparing the original integral with Eq.~(\ref{baikov}), one observes that the number of integration variables changes from $LD$ to $M$. When we perform the projection \eqref{decomp} of each 4-momentum onto the space generated by the vectors coming next, it is actually clear that this process cannot continue indefinitely, as all the vectors are certainly not independent if they lie in the physical 4D space. The decomposition we describe in Eq.~\eqref{decomp} is to be thought in an abstract (sufficiently large) dimension $D$. Since the final expression is an analytic function of $D$, we get the physically meaningful result via an analytic continuation down to $D=4$. This discussion can be also summarised by saying that we are implicitly using dimensional regularization to make sense of the expression (\ref{feynman}), which is obviously divergent in $D=4$.

\

The representation in Eq.~\eqref{baikov} highlights new properties of the original integral \eqref{feynman}, and allows to study its topological structure as Aomoto-Gel'fand integral \cite{Mastrolia:2018uzb}. In fact, extending the integral in  Eq.~\eqref{baikov} to the complex space, it takes the form
\begin{equation}
I=K\int_C u(\vec{z})\phi(\vec{z})\,,\label{I}
\end{equation}
where $K$ is constant prefactor, $u=B^\gamma$ is a multivalued function such that $u(\partial C)=0$ and $\phi$ is an $M$-form
\begin{equation}
\phi\equiv\hat{\phi}d^Mz=\frac{dz_1\wedge\cdots\wedge dz_M}{z_1^{\nu_1}\cdots z_M^{\nu_M}}\,.
\end{equation}
Because of the Stokes theorem, given a certain $(M-1)$-form $\xi$ the following identity holds:
\begin{equation}
\int_Cd(u\xi)=\int_{\partial C}u\xi=0\,,\label{stokes}
\end{equation}
as $u\xi$ is integrated along $\partial C$ where $u$ vanishes. Eq.~\eqref{stokes} can also be rewritten as
\begin{equation}
\int_Cd(u\xi)=\int_C(du\wedge \xi+ud\xi)=\int_C u(\underbrace{d\log u}_{\omega}\wedge+d)\xi\equiv\int_C u\nabla_\omega\xi=0\label{stokes2}.
\end{equation}
Eq.~\eqref{stokes2} states that, because of the introduction of the connection $\nabla_\omega=d+\omega\wedge$ where $\omega\equiv d\log u$, it holds
\begin{equation}
\int_C u\phi=\int_Cu(\phi+\nabla_\omega\xi)\,,\label{equiv}
\end{equation}
as the second term in the right side gives a null contribution. Eq.~\eqref{equiv} identifies an equivalence class, addressed as
\begin{equation}
{}_\omega{\bra{\phi}}:=\,\,\,\phi\sim\phi+\nabla_\omega\xi\,.
\end{equation}

This equivalence class defines \textsl{twisted cohomologies} (twisted because the derivative involved is not simply $d$ as in the de Rham cohomology but it is the covariant derivative $\nabla_\omega$). Representatives of a class are called 
\textsl{twisted cocycles}. 

In this fashion, the integral $I$,
\begin{equation}
I=\int_C u\phi 
\label{eq:def:integralpairinggeneric}
\end{equation}
arises as a {\it pairing} between the twisted cycle $|C]$ and the twisted cocycle $\bra{\phi}$.
For notation ease, the subscript $\omega$ is understood, and restored when needed.

\
Aomoto-Gel'fand integrals admit linear and quadratic relations which can be used to simplify the evaluation of scattering amplitudes. In particular, linear relations can be exploited to express any integral as linear combination of an independent set of functions, called {\it master integrals} (MIs). The decomposition of Feynman integrals in terms of MIs was proposed in \cite{Chetyrkin:1981qh} and later systematised in \cite{Laporta:1996mq}, and represents the most common computational technique for addressing multi-loop calculus nowadays. The novel insights we elaborate on in this work allow to explore the underpinning vector space structure obeyed by Aomoto-Gel'fand integrals, in order to investigate the properties of Feynman integrals making use of co-homological techniques.
Accordingly, the decomposition of any generic integral $I=\int_C u\phi \equiv \bra{\phi}C]$ in terms of a basis of MIs, 
say $J_i$, can be achieved in a twofold approach:
\begin{eqnarray}
I
 = \sum_{i = 1}^\nu c_i \int_C u \, e_i = \sum_{i = 1}^\nu c_i \, J_i
\end{eqnarray}
or 
\begin{eqnarray}
I
 = \sum_{i = 1}^\nu c_i \int_{C_i} u \, \phi = \sum_{i = 1}^\nu c_i \, J_i\,.
\end{eqnarray}
The former decomposition involves a basis formed by independent equivalence classes $\{e_i\}$ of the underlying twisted cohomology, while the latter, a basis formed by independent equivalence classes $\{C_i\}$ of the twisted homology. Remarkably the dimension of the twisted homology and co-homology spaces is the same.

%
%
%
%
%
%

\

Let us finally remark that although Baikov representation turned out to be useful to uncover the cohomological structure of Feynman integrals \cite{Mastrolia:2018uzb}, there is no commitment in using it necessarily. 
Other parametric representations, such as Feynman-Schwinger \cite{Bogner:2010kv}, 
$n$-dimensional polar coordinates \cite{Mastrolia:2016dhn},
and Lee-Pomeransky representation \cite{Lee:2013hzt}, 
to name a few,
can be equivalently used \cite{Mizera:2020wdt}. 
In fact, the integral in \eqref{feynman} can be cast in the form \cite{Lee:2013hzt}
\begin{equation}
I_{\nu_1,\cdots,\nu_N}=\frac{\Gamma(D/2)}{\Gamma((L+1)D/2-|n|)\prod_a \Gamma(\nu_a)}\int_0^\infty\cdots\int_0^\infty\frac{dz_a}{z_a^{1-\nu_a}}\mathcal{G}^{-D/2}\,,
\label{parametric}
\end{equation} 
with $|n|=\sum_a \nu_a$ and $\mathcal{G}=\mathcal{F}+\mathcal{U}$, where $\mathcal{F}$ and $\mathcal{U}$ are the Symanzik polynomials. The latter are defined as 
\begin{equation}
    \mathcal{F}=\det{A}C-(A^{\text{adj}})_{ij}B_iB_j\,, \quad \mathcal{U}=\det{A}\,,
\end{equation}
with $A$, $B$ and $C$ being the matrices that appear in the decomposition of the denominators, as 
\begin{equation}
    D_a=A_{a,ij}q_iq_j+2B_{a,ij}q_ip_j+C_a\,,
\end{equation}
where
\begin{equation}
    A_{ij}=\sum_a z_a A_{a,ij}\,,\quad B_i=\sum_a z_a B_{a,ij}p_j\,, \quad C=\sum_a z_a C_a\,.
\end{equation}

We observe that, although the integral in Eq.~(\ref{parametric}) has the structure of Aomoto-Gel'fand integrals Eq.~(\ref{eq:def:integralpairinggeneric}),
the polynomial $\mathcal{G}$ does not vanish on the boundary of the integration domain, therefore surface terms can emerge in in the (co)-homology decomposition. It turns out that these extra terms can be related to integrals belonging to simpler sectors, i.e. with fewer denominators than the ones in the original diagram.

\section{The twisted cohomology vector space}\label{secnumberMI}

\subsection{Vector-space structure}

To study the co-homology of dimensionally regulated Feynman integrals, we consider integrals of the form
\begin{equation}
I = \int_{\mathcal{C}_R}u(\mathbf{z}) \, \varphi_L (\mathbf{z}) = \langle \varphi_L| \mathcal{C}_R],
\label{eq:integral_definition}
\end{equation}
regarded as a pairing between $\langle \varphi_L|$ and the function $u(\mathbf{z})$, integrated over the contour $| \mathcal{C}_R]$.
In particular $u(\mathbf{z})$ is a multivalued function, $u(\mathbf{z})=\mathcal{B}(\mathbf{z})^{\gamma}$ (or $u(\mathbf{z})=\prod_{i}, \mathcal{B}_{i}(\mathbf{z})^{\gamma_i}$), with 
\begin{equation}
\mathcal{B}(\partial \mathcal{C}_R)=0 \ .
\label{eq:B_vanishing_on_parital_C}
\end{equation}
The pairing so defined is not strictly speaking a pairing of forms but of equivalence classes of $n$-forms, so that two differential forms in the same class differ by covariant derivative-terms whose contribution under integration over a contour vanishes.
Let us see how this works.

\subsection{Dual cohomology groups}

Let $\xi_L$ be an $(n{-}1)$-differential form and $\mathcal C_R$ an integration contour such that (\ref{eq:B_vanishing_on_parital_C}) holds true. Thus, we can use Stokes theorem to write
\begin{equation}
0=\int_{\partial \mathcal{C}_R}u \, \xi_L =\int_{\mathcal{C}_R}d(u \, \xi_L)= \int_{\mathcal{C}_R}u \left(\frac{du}{u} \wedge+ d \right) \xi_L= \int_{\mathcal{C}_R}u \, \nabla_{\omega} \, \xi_L
= \langle{ \nabla_\omega \xi_L}| {\cal C}_R] \ ,
\end{equation}
where
\begin{equation}
\nabla_{\omega}= d + \omega \wedge, \qquad \omega= d \log u.
\label{eq:nablasmallomegaplus}
\end{equation}
As a consequence we immediately get
\begin{equation}
\langle{\phi_L}|{\cal C}_R]
= \langle{\phi_L } |{\cal C}_R] + \langle{\nabla_{\omega} \xi_L } |{\cal C}_R]
= \langle{\phi_L +\nabla_{\omega} \xi_L } |{\cal C}_R] \ ,
\end{equation}
which meaans that the forms $\varphi_L$ and $\varphi_L+\nabla_{\omega} \xi_L$ give the same result upon integration, as stated above, and we can consider them as equivalent for the purpose of computing intersections. This is an equivalence relation
and we can say that the two forms are in the same equivalence class, writing
\begin{equation}
\phi_L \sim \phi_L+\nabla_{\omega} \xi_L \ .
\label{eq:equivalence_relation}
\end{equation}
Of course the equivalence classes of $n$-forms defined by the equivalence relation ($\ref{eq:equivalence_relation}$), that we will denote with $\langle \varphi_L|$, define a vector space, called the \emph{twisted cohomology group} $H_{\omega}^{n}$.

\

Analogously, every time two $n$-dimensional contours differ by a boundary, they have the same boundary so that if (\ref{eq:B_vanishing_on_parital_C}) holds for one of them it holds for the other one. Using again the Stokes theorem, it follows that when used as
integration contours for a closed $n$-form $u\phi_L$ they give the same result and, again, can be considered equivalent. We will indicate the equivalence classes of integration contours so obtained with $| \mathcal{C}_R]$. They define the 
\emph{twisted homology group} $H_{n}^{\omega}$, which is indeed an abelian group identifiable with a vector space after tensorization with $\mathbb R$.

\


Let us consider, now, what happens when using $u^{-1}$ in the integral definition, instead of $u$.
Let us define a \emph{dual} integral
\begin{equation}
\Tilde{I}=\int_{\mathcal{C}_L} u(\mathbf{z})^{-1} \, \phi_R(\mathbf{z}) 
= [ {\cal C}_L| \phi_R\rangle \ ,
\label{eq:dual_integral_definition}   \end{equation}
as a pairing between the dual twisted cycle $[{\cal C}_L|$ and the dual twisted cocycle $|\phi_R\rangle$.
It is clear by construction that what we did before can be repeated in the same way just replacing the covariant derivative with
\begin{equation}
\nabla_{-\omega}=d-\omega \wedge, \qquad \omega=d \log u.
\label{eq:nablasmallomegaminus}
\end{equation}
Then, we immediately get an equivalence relation for dual twisted cocycles, analogous to \eqref{eq:equivalence_relation},
\begin{equation}
\varphi_R \;\sim\; \varphi_R+ \nabla_{- \omega} \xi_R
\label{eq:dual_equivalence_relation}
\end{equation}
which defines equivalence classes denoted by $| \varphi_R \rangle$. These equivalence classes define the \emph{dual} vector space $(H^{n}_{\omega})^{\ast}=H^n_{- \omega}$. 

As mentioned earlier, an equivalence relation among dual integration contours may be also considered, yielding to 
the identification of the dual homology vector space $(H_{n}^{\omega})^{\ast}=H_{n}^{- \omega}$, referred to as the \emph{dual twisted homology group}, whose elements are denoted by $[ \mathcal{C}_{L}|$.\\ 

\subsection{Number of Master Integrals}

We have shown that we can recognize a vector space structure in Feynman integrals, so that a given integral may be expressed in some basis of the twisted cohomology: if one is able to compute the coefficients of the decomposition, which we discuss in Section \ref{secinternumb}, the computation will be reduced to the evaluation of some fixed Master Integrals. All this reasoning would be useless if the dimension of such vector space turned out to be infinite. Luckily, the number of the independent MIs, which we refer to as $\nu$, is known to be actually finite \cite{Smirnov:2010hn}. Moreover, it is known how to compute $\nu$ in practice, as we show in this section.

\

There are actually many interpretations of $\nu$: here we focus mainly on the interpretation given by Lee and Pomeransky \cite{Lee:2013hzt}. Firstly, $\nu$ is the number of the many independent equivalence classes in the associated cohomology group. Due to the Poincar\'e duality between cohomology and homology classes, it is equivalent to study the dimension of the homology space, which consists in counting how many non homotopically contractible closed paths exist in the space of integration contours. 

\

We consider a particularly simple case in order to have a clear understanding of the usual reasonings followed when counting independent cycles. Once having derived the associated Baikov representation \eqref{baikov}, it is appropriate to perform a maximal cut (for a more detailed discussion on cuts performed on Feynman integrals in the context of Baikov representation, we recommend to take a look at \cite{Lee:2012te} and \cite{Harley:2017qut}): out of the $M$ total denominators, we put to 0 the $N$ original ones corresponding to physical propagators. Of course, putting a propagator on mass shell is the same as performing a residue, which reduces the level of complexity of the integral. We consider the case where the integral along the maximal cut reduces to one with a single variable:
\begin{equation}
I=\int \frac{dz_1}{z_1^{\nu_1}} B(z_1)^\gamma.\label{inteasy}
\end{equation}
We stress that, by taking the maximal cut, all the original physical propagators were eliminated, so that in Eq.~\eqref{inteasy} the variable $z_1$ is related to one of the fake propagators introduced in Section~\ref{secbaikov}: hence we can consider $\nu_1<0$ so it does not introduce any additional singularities. When looking at how many types of non equivalent integration contours one can build, it is clear that the topology of the space must be taken into account. Suppose that $B(z_1)$ has $m$ distinct zeros: the power $\gamma$ introduces in the $z_1$ complex space $m$ cuts starting from each zero point and going to infinity. Supposing that the integrand is well-behaved at infinity (if this were not true, then the whole integral $I$ would not converge), we can connect at infinity the $m$ paths one can draw around the $m$ cuts. The resulting closed path is actually contractible in a single point, hence only $m-1$ paths are independent (Figure \ref{fig:closedpath}).\\
\noindent
Qualitatively, notice that if $m$ is the order of the polynomial $B(z_1)$, then $m-1$ is the order of the polynomial $\partial_{z_1} B$, hence it is related to the number of zeros of $\partial_{z_1} B$. Getting back to the notation
\begin{equation}
I=\int_C u\phi\,,\label{Iuphi}
\end{equation}
where $u=B^\gamma$, it is equivalent to the number of solutions of
\begin{equation}
\omega=d\log u=\gamma (\partial_{z_1}B/B) dz_1=0\,,\label{properzeros}
\end{equation}
called the number of \textsl{proper zeros}. Eq.~\eqref{properzeros} suggests a deep connection between the number $\nu$ of MIs and the number of critical points of $B$.

\begin{figure}[!htbp]
\begin{center}
\begin{tikzpicture}[>=latex,decoration={zigzag,amplitude=.5pt,segment length=2pt}]
\draw [ultra thick,->](-4,0) -- (4,0);
\draw [ultra thick,->](0,-4) -- (0,4);
\draw [ultra thick, dashed, green!40!black] (0,0) circle (3.5);
\filldraw [white] (-2.47487,2.47487) circle (0.15);
\filldraw [white] (1.4667,3.17786) circle (0.15);
\filldraw [white] (3.03885,1.73649) circle (0.15);
\filldraw [white] (1.26856,-3.26202) circle (0.15);
\filldraw [white] (-3.26202,-1.26856) circle (0.15);
\draw [thick, white!50!blue, decorate] (-1,1) -- (-2.47487,2.47487);
\draw [thick, white!50!blue, decorate] (0.6,1.3) -- (1.4667,3.17786);
\draw [thick, white!50!blue, decorate] (1.4,0.8) -- (3.03885,1.73649);
\draw [thick, white!50!blue, decorate] (0.7,-1.8) -- (1.26856,-3.26202);
\draw [thick, white!50!blue, decorate] (-1.8,-0.7) -- (-3.26202,-1.26856);
\draw [rotate=45,green!40!black,ultra thick] (-0.15,3.52) -- (-0.15,1.414) arc (180:360:0.15) -- (0.15,3.52);
\draw [rotate=-24.7751,green!40!black,ultra thick] (-0.15,3.52) -- (-0.15,1.43178) arc (180:360:0.15) -- (0.15,3.52);
\draw [rotate=-60.2551,green!40!black,ultra thick] (-0.15,3.52) -- (-0.15,1.61245) arc (180:360:0.15) -- (0.15,3.52);
\draw [rotate=21.2505-180,green!40!black,ultra thick] (-0.15,3.52) -- (-0.15,1.93132) arc (180:360:0.15) -- (0.15,3.52);
\draw [rotate=111.2505,green!40!black,ultra thick] (-0.15,3.52) -- (-0.15,1.93132) arc (180:360:0.15) -- (0.15,3.52);
\filldraw [red] (-1,1) circle (2pt);
\filldraw [red] (0.6,1.3) circle (2pt);
\filldraw [red] (1.4,0.8) circle (2pt);
\filldraw [red] (0.7,-1.8) circle (2pt);
\filldraw [red] (-1.8,-0.7) circle (2pt);
\node at (-0.5, 3.8) {\pmb {Im$(z)$}};
\node at (4, -0.3) {\pmb {Re$(z)$}};
\end{tikzpicture}
\caption{Complex plane with $m=5$ cuts (undulate blue curves). Each cut is encircled by a path going to infinity while never crossing any cut. Dashed green lines connect at infinity the full green lines and overall create a closed path which is clearly 
contractible in 0.}
\label{fig:closedpath}
\end{center}
\end{figure}
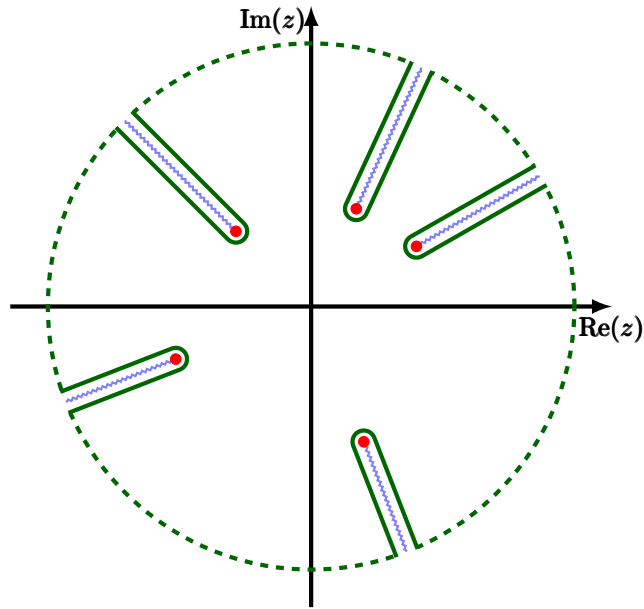

As shown more extensively in \cite{Lee:2013hzt}, this connection is actually much more general: given an integral of the form \eqref{Iuphi}, in which $\phi$ is a holomorphic $M$-form and $u$ is a multivalued function such that $u(\partial C)=0$ , then the number of Master Integrals is
\begin{equation}
\nu=\,\,\,\text{number of solutions of the system}\,\,\,
\begin{cases}
\omega_1=0\\
\vdots\\
\omega_n=0
\end{cases},
\end{equation}
where
\begin{equation}
\omega=d\log u(\vec{z})=\sum_{i=1}^n \partial_{z_i}\log u(\vec{z}) dz_i= \sum_{i=1}^n \omega_i dz_i.
\end{equation}
Summing up, the number $\nu$ of MIs, which is the dimension of both the cohomology and homology groups thanks to the Poincar\'e duality, is equivalent to the number of proper critical points of $B$, which solve $\omega=0$. We mention that $\nu$ is also related to another geometrical object: the Euler characteristic $\chi$ \cite{Lee:2013hzt}\cite{Bitoun:2018afx}. It is found that is linked to $\chi(P_\omega)$, where $P_\omega$ is a projective variety defined as the set of poles of $\omega$, through the relation \cite{Frellesvig:2019uqt}
\begin{equation}
\nu=\dim H^{n}_{\pm \omega}=(-1)^n\left(n+1-\chi(P_\omega)\right).
\end{equation}
While we do not delve into the details of this particular result, we highlight how, once again, $\nu$ relates the physical problem of solving a Feynman integral into a geometrical one.


\subsection{Linear and bilinear identities}

The bases of dual twisted cocycles $\{e_i \} \in H^n_\omega$ and $\{h_i\} \in H^n_{-\omega}$, as well as the bases of  
the dual twisted cycles $|{\cal C}_{R,i}] \in H_n^\omega$ and $[{\cal C}_{L,i}| \in H_n^{-\omega}$ can be used 
to express the identity operator in the respective vector spaces.
In particular, in the cohomology space, the identity resolution reads as,
\begin{eqnarray}
\sum_{i,j=1}^{\nu} | h_i \rangle \left( \mathbf{C}^{-1} \right)_{ij} \langle e_j | = \mathbb{I}_c 
\label{eq:identityexp}
\end{eqnarray}
where we defined the \emph{metric matrix}
\begin{equation}
\label{eq:metric_matrix}
\mathbf{C}_{ij}= \langle e_i | h_j \rangle \ ,
\end{equation}
whose elements are {\it intersection numbers} of the twisted basic forms.
We can do the same in the homology space, where the resolution of the identity takes the form
\begin{eqnarray}
\sum_{i,j=1}^{\nu} | \mathcal{C}_{R,i}]  \left( \mathbf{H}^{-1} \right)_{ij} [ \mathcal{C}_{L,j} | = \mathbb{I}_h \, ,
\end{eqnarray}
with the metric matrix now given by
\begin{equation}
\label{eq:metric_matrix:cycles}
\mathbf{H}_{ij} = [\mathcal{C}_{L,i} | \mathcal{C}_{R,j}] \ ,
\end{equation}
in terms of intersection numbers of the basic twisted cycles.

Linear and bilinear relations for Aomoto-Gel'fand-Feynman integrals, as well as the differential equations and the finite difference equation they obey are a consequence of the purely algebraic application of the identity operators defined above:
this is the simple observation made in \cite{Mastrolia:2018uzb}, yielding what we consider as the profound developments reached in the subsequent studies \cite{}.

In fact, the decomposition of any generic integral $I=\langle \varphi_L| \mathcal{C}_R]$ in terms of a set of $\nu$ MIs $J_i=\langle e_i | \mathcal{C}_R]$ reads as
\begin{equation}
I= \sum_{i=1}^{\nu} c_i \, J_i
\label{eq:decomposition_MIs}
\end{equation}
and can be understood as coming from the underpinning decomposition of differential forms, as
\begin{equation}
\langle \phi_L|= \sum_{i=1}^{\nu} c_i \, \langle e_i|\,,
\label{eq:decomposition_master_forms}
\end{equation}
(because the integration cycle is the same for all the integrals of eq.~\eqref{eq:decomposition_MIs}).
Likewise, the decomposition of a \emph{dual} integral $\Tilde{I}=[ \mathcal{C}_L| \varphi_R \rangle$ in terms of a set of $\nu$ \emph{dual} MIs $\Tilde{J}_i=[ \mathcal{C}_L| h_i \rangle$
\begin{equation}
\Tilde{I}= \sum_{i=1}^{\nu} \Tilde{c}_i \, \Tilde{J}_i
\end{equation}
becomes
\begin{equation}
| \phi_R \rangle = \sum_{i=1}^{\nu} \Tilde{c}_i \, | h_i \rangle.
\label{eq:decomposition_dual_master_forms}
\end{equation}
By applying the identity operator $\mathbb{I}_c $ to the forms $\langle \phi_L|$ and $| \phi_R \rangle$, 
\begin{eqnarray}
\langle \phi_L| &=& \langle \phi_L| \mathbb{I}_c 
= \sum_{i,j=1}^{\nu}  \langle \phi_L| h_j \rangle \left( \mathbf{C}^{-1} \right)_{ji} \langle e_i | \ , \\
| \phi_R \rangle &=& \mathbb{I}_c | \phi_R \rangle 
= 
\sum_{i,j=1}^{\nu} | h_i \rangle \left( \mathbf{C}^{-1} \right)_{ij} \langle e_j | \phi_R \rangle \ , 
\end{eqnarray}
the coefficients $c_i$, and $\Tilde{c}_i$ of the linear relations in Eqs.~(\ref{eq:decomposition_master_forms}),~(\ref{eq:decomposition_dual_master_forms}) read as
\begin{align}
c_i &= \sum_{j=1}^{\nu} \, \langle \phi_L | h_j \rangle \, \left( \mathbf{C}^{-1} \right)_{ji} \,, \label{eq:masterdeco} \\
\Tilde{c}_i & =\sum_{j=1}^{\nu} \, \left( \mathbf{C}^{-1} \right)_{ij} \, \langle e_j | \phi_R \rangle \, . \label{eq:masterdecodual}
\end{align}
The latter two formulas, dubbed \emph{master decomposition formulas} for twisted cycles \cite{Mastrolia:2018uzb,Frellesvig:2019kgj}, imply that the decomposition of any (dual) Aomoto-Gel'fand-Feynman 
integral can be expressed as linear combination of (dual) master integrals is an algebraic operation, similar to the decomposition/projection of any vector within a vector space, that can be executed by computing intersection numbers of twisted de Rham differential forms (cocycles).

Alternatively, by using the identity operator $\mathbb{I}_h$ in the homology space, one obtains the decomposition of (dual) twisted cycles $|{\cal C}_R|$ and $|{\cal C}_L|$ in terms of (dual) bases,
\begin{eqnarray}
            |{\cal C}_R]   =  \sum_i a_i \, | \mathcal{C}_{R,i}] \ , \qquad {\rm and} \qquad
\left[ {\cal C}_L| \right. =  \sum_i {\tilde a}_i \, [\mathcal{C}_{L,i}| \ ,
\end{eqnarray}
where
\begin{eqnarray}
a_i = \sum_{j=1}^{\nu} \left( \mathbf{H}^{-1} \right)_{ij} [ \mathcal{C}_{L,j} |{\cal C}_R]  
\ , \qquad {\rm and} \qquad
{\tilde a}_i = 
\sum_{i=1}^{\nu} [ {\cal C}_L | \mathcal{C}_{R,j}]  \left( \mathbf{H}^{-1} \right)_{ji} \ .
\label{eq:masterdeco:homology}
\end{eqnarray} 
They imply the decomposition of the (dual) integrals,
$I=\langle \varphi_L| \mathcal{C}_R]$ and $\Tilde{I}=[ \mathcal{C}_L| \varphi_R \rangle$, in terms of MIs, $J_i'$ 
and ${\tilde J}_i'$,
\begin{eqnarray}\label{Eq_decomposition}
I=\langle \varphi_L| \mathcal{C}_R] = \sum_{i=1}^\nu a_i \, J_i \ , 
\qquad {\rm and} \qquad
\Tilde{I}=[ \mathcal{C}_L| \varphi_R \rangle
= \sum_{i=1}^\nu {\tilde a}_i \, {\tilde J}_i \ ,
\end{eqnarray}
with
\begin{eqnarray}
         J_i' = \langle \varphi_L | \mathcal{C}_{R,i}] 
         \ , \qquad {\rm and} \qquad
{\tilde J}_i' = [ \mathcal{C}_{L,i} | \varphi_R \rangle \ ,
\end{eqnarray}
where the MIs are characterised by the independent elements of the homology bases.

In the above formulas, $\mathbf{C}$ and $\mathbf{H}$ are $(\nu \times \nu)$-matrices of intersection numbers, 
which, in general, differ from the identity matrix. 
For orthonormal intersections of  
basic cocycles, $\langle e_i|$ and dual-forms $|h_i\rangle$, 
and
basic cycles, $| \mathcal{C}_{R,i}]$ and their dual $[ \mathcal{C}_{L,i}|$, 
they turn into the unit matrix, hence simplifying the decomposition formulas in Eqs. (\ref{eq:masterdeco},\ref{eq:masterdecodual},\ref{eq:masterdeco:homology}).
As usual, one can use the orthonormalization Gram-Schmidt procedure in order to get an orthonormal basis with respect to the intersection product.
For 1-forms it is possible to construct orthonormal bases in a direct way starting from the expression of $\omega$ \cite{Frellesvig:2019kgj}.

We can now get the quadratic identities simply by inserting the identity operators $\mathbb{I}_c$ and $\mathbb{I}_h$ in the pairing between the twisted cocyles or cycles,
\begin{eqnarray}
\langle \phi_L | \phi_R \rangle &=& \sum_{i,j=1}^{\nu} \langle \phi_L | \mathcal{C}_{R,i}]  
\left( \mathbf{H}^{-1} \right)_{ij} [ \mathcal{C}_{L,j} | \phi_R \rangle \\
\left[ \mathcal{C}_{L} | \mathcal{C}_{R} \right] &=& 
\sum_{i,j=1}^{\nu} [ \mathcal{C}_{L} | h_i \rangle \left( \mathbf{C}^{-1} \right)_{ij} \langle e_j | \mathcal{C}_{R}] \, .
\end{eqnarray}
These are the {\it{Twisted Riemann's Period Relations}} (TRPR) \cite{cho1995}, see also Eq.~(\ref{TRPR}).

For applications of twisted intersection numbers to bilinear relations and to Gel'fand-Kapranov-Zelevinski systems see \cite{matsubaraheo2019euler,matsubaraheo2019algorithm,goto2020homology,Broadhurst:2016hbq, Broadhurst:2016myo,Broadhurst:2018tey,Zhou:2017vhw,Zhou:2017jnm,fresn2020quadratic,fresn2020quadratic2,Acres:2021sss}.

\section{Pictorial (co)homology}\label{seccohomology}
Before providing a geometric interpretation of Feynman-type integrals, we want to recall some fundamental facts of topology, homology and cohomology, only at an intuitive level, for those who need to refresh these concepts.
\subsection{The Euler characteristic}
One of the most known topological facts is probably the Euler theorem, \cite{SPHM_1982___8_A1_0}, relating the numbers of faces, edges, and vertices in a tessellation of a compact simply connected region of the plane, like an electricity grid or a Feynman diagram without external legs: 
the number $v$ of vertices minus the number $j$ of edges plus the number $f$ of faces is always equal to 1,
\begin{align}
 v-j+f=1,
\end{align}
independently on the details of the tessellation. This can be quite easily understood as follows. The rules of the game are: any face is simply connected and has a 1-dimensional boundary that is the union of edges and vertices (no lines or vertices are 
internal to a face); any edge is a simple line ending at two vertices (that may coincide). No vertices belong to an internal point of an edge. \\
With this rules, given a simply connected region, the simple tessellation we can do is with just 1 face (covering the whole region), 1 edge (the whole boundary cut in a point), 1 vertex (the cut point). Therefore, in this case, $v-j+f=1$. Now, suppose it has given
a tessellation. We can modify it by adding an edge. This can be done in four different ways:
\begin{itemize}
 \item the new edge starts and ends in already existing (and possibly coincident) vertices (Fig. \ref{fig:vertex-vertex}). In this case the new edge cuts a face in two and the new tessellation has $v'=v$, $j'=j+1$, $f'=f+1$ so that $v'-j'+f'=v-j+f$;
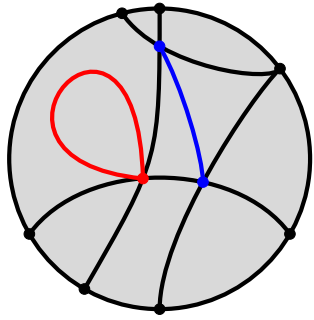
\begin{figure}[!h]
\begin{center}
\begin{tikzpicture}[>=latex,decoration={zigzag,amplitude=.5pt,segment length=2pt}]
\draw [ultra thick, fill=white!70!gray] (0,0) circle (2);
\draw [ultra thick] (0,2) .. controls (0,0) and (0,0) .. (-1,-1.732);
\draw [ultra thick] (-0.5,1.93649) .. controls (0,1.2) and (1.5,1) .. (1.6,1.2);
\draw [ultra thick] (1.6,1.2) .. controls (1.2,0.8) and (0,-1) .. (0,-2);
\draw [ultra thick] (-1.732,-1) .. controls (-1,0) and (1,0) .. (1.732,-1);
\draw [ultra thick,blue] (-0.001,1.497) .. controls (0.3,0.98) and (0.576,0) .. (0.576,-0.309);
\draw [ultra thick,red] (-0.222,-0.262) .. controls (-2.9,0) and (-0.3,2.8) .. (-0.222,-0.262);
\filldraw (0,2) circle (2pt);
\filldraw (-1,-1.732) circle (2pt);
\filldraw (-0.5,1.93649) circle (2pt);
\filldraw (1.6,1.2) circle (2pt);
\filldraw (0,-2) circle (2pt);
\filldraw (-1.732,-1) circle (2pt);
\filldraw (1.732,-1) circle (2pt);
\filldraw [blue](-0.001,1.497) circle (2pt);
\filldraw [red](-0.222,-0.262) circle (2pt);
\filldraw [blue](0.576,-0.309) circle (2pt);
\end{tikzpicture}
\caption{The red edge starts and ends at the same vertex, while the blue one belongs between two vertices.}
\label{fig:vertex-vertex}
\end{center}
\end{figure}
\item the new edge extends from an already existing vertex to a new vertex attached to an existing edge (Fig. \ref{fig:vertex-edge}). The new edge separates a face into two faces and the new vertex separates the old edge into two edges. Therefore, the new tessellation has
$v'=v+1$, $j'=j+2$, $f'=f+1$, so that $v'-j'+f'=v-j+f$;
\begin{figure}[!h]
\begin{center}
\begin{tikzpicture}[>=latex,decoration={zigzag,amplitude=.5pt,segment length=2pt}]
\draw [ultra thick, fill=white!70!gray] (0,0) circle (2);
\draw [ultra thick] (0,2) .. controls (0,0) and (0,0) .. (-1,-1.732);
\draw [ultra thick] (-0.5,1.93649) .. controls (0,1.2) and (1.5,1) .. (1.6,1.2);
\draw [ultra thick] (1.6,1.2) .. controls (1.2,0.8) and (0,-1) .. (0,-2);
\draw [ultra thick] (-1.732,-1) .. controls (-1,0) and (1,0) .. (1.732,-1);
\draw [ultra thick,red] (0.846,1.18) .. controls (0.7,0.98) and (0.576,0) .. (0.576,-0.309);
\draw [ultra thick,blue] (-0.031,0.559) .. controls (-1.9,-1) and (-2,2) .. (-0.031,0.559);
\filldraw (0,2) circle (2pt);
\filldraw (-1,-1.732) circle (2pt);
\filldraw (-0.5,1.93649) circle (2pt);
\filldraw (1.6,1.2) circle (2pt);
\filldraw (0,-2) circle (2pt);
\filldraw (-1.732,-1) circle (2pt);
\filldraw (1.732,-1) circle (2pt);
\filldraw (-0.001,1.497) circle (2pt);
\filldraw (-0.222,-0.262) circle (2pt);
\filldraw [red](0.576,-0.309) circle (2pt);
\filldraw [red](0.846,1.18) circle (2pt);
\filldraw [blue](-0.031,0.559) circle (2pt);
\end{tikzpicture}
\caption{The blue edge starts and ends at the same new vertex, while the red one starts from an old vertex and ends in a new one.}
\label{fig:vertex-edge}
\end{center}
\end{figure}
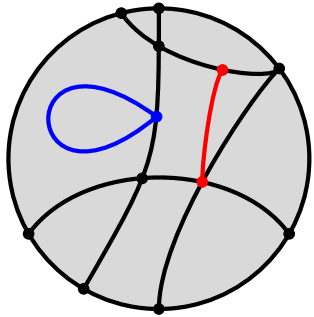
\item the new edge starts from and ends to a unique new vertex, inserted in an old edge (Fig. \ref{fig:vertex-edge}). The new edge cuts a face into two and the new vertex cuts the old edge into two. Therefore, the new tessellation has
$v'=v+1$, $j'=j+2$, $f'=f+1$, so that $v'-j'+f'=v-j+f$;
\item the new edge extends from a new vertex to a different new vertex attached to two different or the same existing edge (Fig. \ref{fig:edge-edge}). The new edge separates a face into two faces. The new vertices separate two edges into two parts each or a unique edge in three. 
Therefore, the new tessellation has $v'=v+2$, $j'=j+3$, $f'=f+1$, so that $v'-j'+f'=v-j+f$.
\begin{figure}[!h]
\begin{center}
\begin{tikzpicture}[>=latex,decoration={zigzag,amplitude=.5pt,segment length=2pt}]
\draw [ultra thick, fill=white!70!gray] (0,0) circle (2);
\draw [ultra thick] (0,2) .. controls (0,0) and (0,0) .. (-1,-1.732);
\draw [ultra thick] (-0.5,1.93649) .. controls (0,1.2) and (1.5,1) .. (1.6,1.2);
\draw [ultra thick] (1.6,1.2) .. controls (1.2,0.8) and (0,-1) .. (0,-2);
\draw [ultra thick] (-1.732,-1) .. controls (-1,0) and (1,0) .. (1.732,-1);
\draw [ultra thick,blue] (0.846,1.18) .. controls (0.7,0.98) and (0.176,0) .. ((0.178,-0.259);
\draw [ultra thick,red] (-0.1,0.1485) .. controls (-1.9,-1) and (-2,2) .. (-0.008,1.028);
\filldraw (0,2) circle (2pt);
\filldraw (-1,-1.732) circle (2pt);
\filldraw (-0.5,1.93649) circle (2pt);
\filldraw (1.6,1.2) circle (2pt);
\filldraw (0,-2) circle (2pt);
\filldraw (-1.732,-1) circle (2pt);
\filldraw (1.732,-1) circle (2pt);
\filldraw (-0.001,1.497) circle (2pt);
\filldraw (-0.222,-0.262) circle (2pt);
\filldraw (0.576,-0.309) circle (2pt);
\filldraw [blue](0.846,1.18) circle (2pt);
\filldraw [red](-0.1,0.1485) circle (2pt);
\filldraw [blue](0.178,-0.259) circle (2pt);
\filldraw [red](-0.008,1.028) circle (2pt);
\end{tikzpicture}
\caption{The red edge lies between two different new vertices on the same edge, while the blue one lies between two new vertices belonging on two different edges.}
\label{fig:edge-edge}
\end{center}
\end{figure}
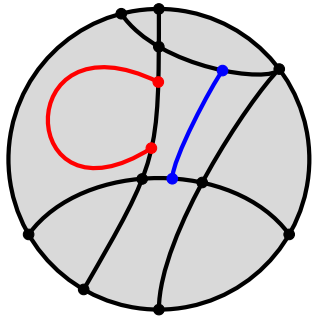
\end{itemize}
Therefore, we see that $\chi:=v-j+f$ is invariant, and since any tessellation can be obtained from the most elementary one by acting with these operations, we see that $\chi=1$. What happens if in place of a simply connected region we consider 
a region with one hole? We can again consider a tessellation for it and prove that $\chi=v-j+f$ is invariant. This follows from the fact that if we close the hole by adding the face corresponding to it, we get a simply connected region with
$f'=f+1$, $v'=v$, $j'=j$, so that 
\begin{align}
 1=v'-j'+f'=v-j+f+1, \quad \Rightarrow \quad \chi=v-j+f=0.
\end{align}
Changing the topology, the number $\chi$ is changed. From the same reasoning, we see that if we consider $k$ holes, then 
\begin{align}
 \chi=1-k.
\end{align}
Having understood the mechanism, we can go beyond and see that $\chi$ does not change if we deform the given region a bit, without changing its topology. For example, we can deform a simply connected region $S$ to cover part of a sphere. But if we close
this surface to cover the sphere, the topology changes and what we do to $(v,j,k)$ of the original piece of surface is to add a face so that for the sphere (Fig. \ref{fig:gluing-sphere})
\begin{align}
 \chi_{S^2}=\chi_S+1=2.
\end{align}
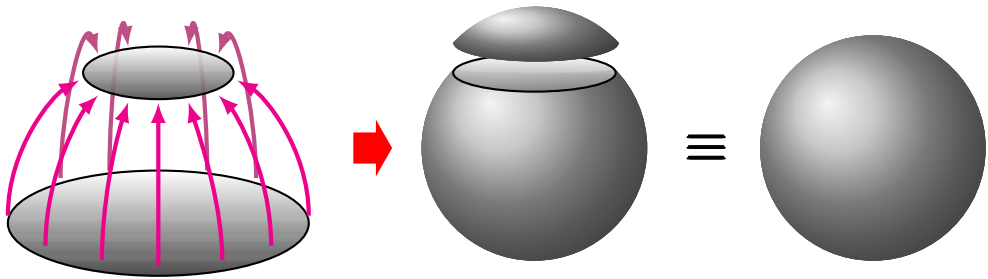
\begin{figure}[!htbp]
\begin{center}
\begin{tikzpicture}[>=latex,decoration={zigzag,amplitude=.5pt,segment length=2pt}]
\draw [ultra thick, color=magenta!75!black, ->] (-6.3,0.6) .. controls (-6.3,1.5) and (-6,2.8) .. (-5.8,2.25);
\draw [ultra thick, color=magenta!75!black, ->] (-3.7,0.6) .. controls (-3.7,1.5) and (-4,2.8) .. (-4.2,2.25);
\draw [ultra thick, color=magenta!75!black, ->] (-5.65,0.7) .. controls (-5.65,1.6) and (-5.5,3) .. (-5.4,2.35);
\draw [ultra thick, color=magenta!75!black, ->] (-4.35,0.7) .. controls (-4.35,1.6) and (-4.5,3) .. (-4.6,2.35);
\draw [thick, top color= white, bottom color=black!70!white] (-5,2) ellipse (1 and 0.35);
\draw [thick, top color= white, bottom color=black!70!white] (-5,0) ellipse (2 and 0.7);
\draw [ultra thick, magenta, ->] (-6.5,-0.3) .. controls (-6.5,0.7) and (-6,1.5) .. (-5.8,1.7);
\draw [ultra thick, magenta, ->] (-3.5,-0.3) .. controls (-3.5,0.7) and (-4,1.5) .. (-4.2,1.7);
\draw [ultra thick, magenta, ->] (-5,-0.57) -- (-5,1.6);
\draw [ultra thick, magenta, ->] (-5.75,-0.49) .. controls (-5.75,0.2) and (-5.6,1) .. (-5.4,1.61);
\draw [ultra thick, magenta, ->] (-4.15,-0.49) .. controls (-4.15,0.2) and (-4.4,1) .. (-4.6,1.61);
\draw [ultra thick, magenta, ->] (-7,0.1) .. controls (-7,1) and (-6.4,1.7) .. (-6.05,1.9);
\draw [ultra thick, magenta, ->] (-3,0.1) .. controls (-3,1) and (-3.6,1.7) .. (-3.95,1.9);
\shade[ball color=black!30!white] (0,1) circle (1.5);
\filldraw [white] (0,2.3) ellipse (2 and 0.3);
\draw [thick, top color=white,bottom color=gray] (0,2) ellipse (1.08 and 0.25); 
\shade [ball color=black!30!white] (-1.08,2.4) arc (90+47.2026:90-47.2026:1.5) arc (0:-180:1.1 and 0.25) -- cycle;
\shade[ball color=black!30!white] (4.5,1) circle (1.5);
\filldraw [red] (-2.4,1.2) -- (-2.1,1.2) -- (-2.1,1.35) -- (-1.9,1) -- (-2.1,0.65) -- (-2.1, 0.8) -- (-2.4,0.8) -- cycle;
\node at (2.3,1) {\huge$\pmb {\equiv}$};
\end{tikzpicture}
\caption{A sphere is obtained adding a face to a disc: the Euler characteristics increases by 1.}
\label{fig:gluing-sphere}
\end{center}
\end{figure}
Each time we make a hole on the sphere, we diminish $\chi_{S^2}$ by 1.  What does it change if we pass from a sphere to a torus? We can get a torus from a sphere in the following way. First, we can cut two spherical caps out of the sphere (say along the arctic
and antarctic polar circles). What remains is a deformed piece of a cylinder and since it is like a sphere with two holes it has $\chi=0$. The torus is then obtained by gluing two of such pieces along the circles
(Fig. \ref{fig:gluing-torus}). Since both have $\chi=0$, we see that for the torus
$T^2$ we have 
\begin{align}
 \chi_{T^2}=0.
\end{align}
\begin{figure}[!htbp]
\begin{center}
\begin{tikzpicture}[>=latex,decoration={zigzag,amplitude=.5pt,segment length=2pt}]
\draw [dashed] (-4.5,1) circle (1.49);
\draw [dashed] (-8.5,1) circle (1.49);
\shade [ball color=black!30!white] (0.679408*1.5-4.5,1-1.10064) arc (-47.2026:47.2026:1.5) arc (0:-180:1.01911 and 0.25) arc (180-47.2026:180+47.2026:1.5) arc (190:350:1.03483 and 0.25) -- cycle; 
\draw [thick, top color=white,bottom color=gray] (-4.5,2) ellipse (1.08 and 0.25); 
\draw [thick, red] (-4.5,2) ellipse (1.08 and 0.25); 
\draw [red, thick] (0.679408*1.5-4.5,1.01-1.10064) arc (350:190:1.03483 and 0.25);
\filldraw [red] (-2.4,1.2) -- (-2.1,1.2) -- (-2.1,1.35) -- (-1.9,1) -- (-2.1,0.65) -- (-2.1, 0.8) -- (-2.4,0.8) -- cycle;
\shade [ball color=black!30!white] (0.679408*1.5-8.5,1-1.10064) arc (-47.2026:47.2026:1.5) arc (0:-180:1.01911 and 0.25) arc (180-47.2026:180+47.2026:1.5) arc (190:350:1.03483 and 0.25) -- cycle; 
\draw [thick, top color=white,bottom color=gray] (-8.5,2) ellipse (1.08 and 0.25); 
\draw [red, thick] (0.679408*1.5-8.5,1.01-1.10064) arc (350:190:1.03483 and 0.25);
\draw [thick, red] (-8.5,2) ellipse (1.08 and 0.25); 
\shade [ball color=black!30!white] (0,1.05) ellipse (1.605 and .905);
\begin{scope}[scale=.8]
\filldraw[white,rounded corners=24pt] (-.9,1.3)--(0,1.9)--(.9,1.3);
\filldraw[white,rounded corners=24pt] (-.9,1.3045)--(0,0.82)--(.9,1.3045);
\path[rounded corners=24pt] (-.9,1.3)--(0,1.9)--(.9,1.3) (-.9,1.3)--(0,.92)--(.9,1.3);
\draw[rounded corners=28pt] (-1.1,1.42)--(0,.72)--(1.1,1.42);
\draw[rounded corners=24pt] (-.9,1.3)--(0,1.9)--(.9,1.3);
\end{scope}
\draw[opacity=0.5,thick,red,densely dashed] (0,0.152) arc (270:90:.2 and .365);
\draw [red, thick] (0,0.152) arc (-90:90:.2 and .365);
\draw [red,thick] (0,1.945) arc (90:270:.2 and .348);
\draw [opacity=0.5,thick,red,densely dashed] (0,1.945) arc (90:-90:.2 and .348);
\draw [red,ultra thick,latex-latex] (-8,1.82) .. controls (-7.5,2.32) and (-5.5,2.32) .. (-5,1.82);
\draw [red,ultra thick,latex-latex] (-8,-0.3) .. controls (-7.5,-0.8) and (-5.5,-0.8) .. (-5,-0.3);
\end{tikzpicture}
\caption{A torus is obtained by gluing two cylinders along the boundaries.}
\label{fig:gluing-torus}
\end{center}
\end{figure}
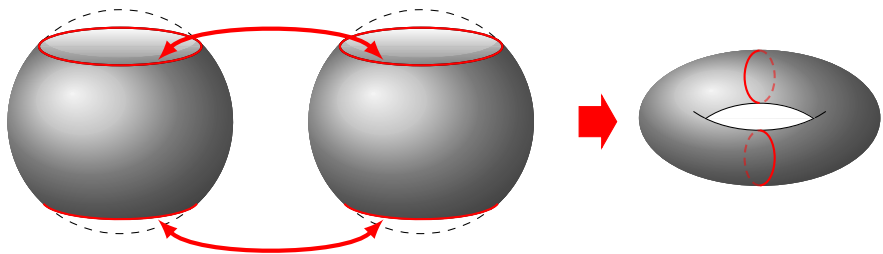
If we practice $k$ holes on the torus we get $\chi=-k$.
If we want to pass to a surface of genus $2$, we have to practice two holes in the torus surface and glue the extremities of a piece of cylinder to the holes. The holes diminish $\chi$ by two, while the cylinder is harmless, so $\chi=-2$ in this case. With the same
construction we see that for a surface $K_{g,k}$ of genus $g$ and $k$ holes on the surface we have
\begin{align}
 \chi_{K_{g,k}}=2-2g-k.
\end{align}
This topological number is called the {\it Euler characteristic of the surface.}
\subsection{Simplicial (co)homology}
The simplest tessellation we can think for a surface is a triangulation. The surface is then equivalent to the union of triangles, which are equivalent to two dimensional simplexes (the convex hull generated by three non aligned points in $\mathbb R^N$).
A triangle has the union of three segments (and three points) as boundary. Each segment is a one dimensional simplex and each point a zero dimensional simplex. Two points are the boundary of a 1-simplex. This way, we see the surface as a collection of 
simplexes glued together. If the triangulation is fine enough, we can guarantee that two simplexes of given dimension meet at most at a simplex of lower dimension (a face). Moreover, all sub simplexes of any simplex is also in the collection. Such a
collection is called a {\it simplicial complex}. Thus, we can see a surface as a simplicial complex. \\
Suppose we want to compute the Euler characteristic starting from it. We have to count the number of faces, edges and vertices. But now, for example, let us consider a face. It is a filled triangle whose boundary is an empty triangle. So it has one face, three edges
and three vertices. If we squash the triangle along an edge until it collapses over the other two edges, the face disappears and we are left with two edges and three vertices (Fig. \ref{fig:squashing-edge}). 
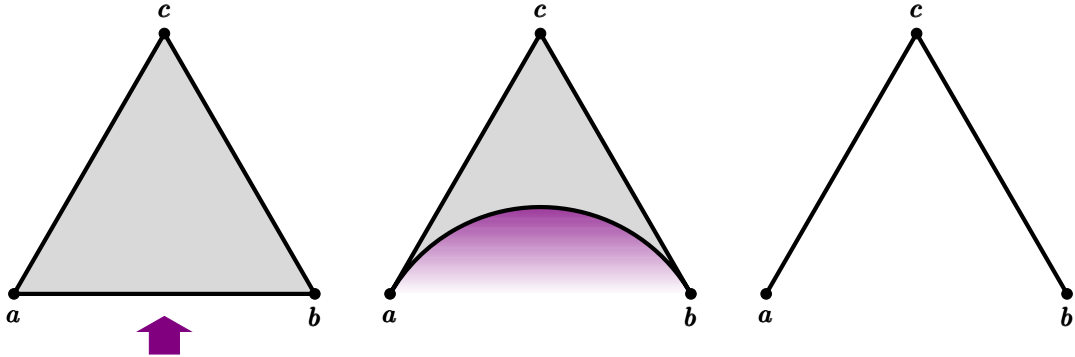
\begin{figure}[!htbp]
\begin{center}
\begin{tikzpicture}[>=latex,decoration={zigzag,amplitude=.5pt,segment length=2pt}]
\filldraw [violet,rotate=90] (-2.4,1.2) -- (-2.1,1.2) -- (-2.1,1.35) -- (-1.9,1) -- (-2.1,0.65) -- (-2.1, 0.8) -- (-2.4,0.8) -- cycle;
\draw [ultra thick,fill=white!70!gray] (-3,-1.6) -- (1,-1.6) -- (-1,2*1.73205-1.6) -- cycle;
\draw [white,top color=violet, bottom color=white] (6,-1.6) arc (30:150:4/1.73205) -- cycle;
\draw [ultra thick,fill=white!70!gray] (2,-1.6) -- (4,2*1.73205-1.6) -- (6,-1.6) arc (30:150:4/1.73205) -- cycle;
\draw [ultra thick] (7,-1.6) -- (9,2*1.73205-1.6) -- (11,-1.6);
\filldraw (-3,-1.6) circle (2pt) (1,-1.6) circle (2pt) (-1,2*1.73205-1.6) circle (2pt) (2,-1.6) circle (2pt) (6,-1.6) circle (2pt) (4,2*1.73205-1.6) circle (2pt) (7,-1.6) circle (2pt) (11,-1.6) circle (2pt) (9,2*1.73205-1.6) circle (2pt);
\node at (-3,-1.9) {$\pmb a$}; \node at (1,-1.9) {$\pmb b$}; \node at (-1,2*1.73205-1.3) {$\pmb c$};
\node at (2,-1.9) {$\pmb a$}; \node at (6,-1.9) {$\pmb b$}; \node at (4,2*1.73205-1.3) {$\pmb c$};
\node at (7,-1.9) {$\pmb a$}; \node at (11,-1.9) {$\pmb b$}; \node at (9,2*1.73205-1.3) {$\pmb c$};
\end{tikzpicture}
\caption{The edge $ab$ is squashed until the triangle reduces to two edges.}
\label{fig:squashing-edge}
\end{center}
\end{figure}
But $\chi$ is left invariant! The idea is then to say that the squashed edge is equivalent
to the other two. If we orient the triangle, call it $(abc)$, so that it has a well defined travelling direction, what we have done can be reformulated by saying that $(ab)$ can be squashed to $(ac)\cup (cb)$ after elimination of the face. In practice one realizes
this by saying that the boundary of the 2-simplex $(abc)$ is 
\begin{align}
 \partial (abc)=(bc)-(ac)+(ab),
\end{align}
where the sign respects the orientation, and stating that {\it if a loop is a boundary, it is trivial}: $\partial (abc)\sim 0$, which means $(ac)\sim (bc)+(ab)$. Also, if we take a union of such triangles, we see that we can eliminate the common edges (intersections
between simplexes) without changing $\chi$, so to obtaining a simply connected polygon $P$ that is a boundary of a region. Again, we will get that one of the edges of the polygon is equivalent to the sum (union) of the remaining ones. It is easy to check that 
this can be
written as
\begin{align}
 P=& \sum_j T_j, \\
 \partial P=&\sum_j \partial T_j \sim 0,
\end{align}
where $T_j$ are the simplexes composing $P$. However, it could happen that the polygon so obtained is not simply connected but it has a number of holes. The best we can say this way is that, after collapsing the triangles, the ``more external boundary'' of
the polygon is equivalent to the sum of boundaries of the internal holes, se Fig. \ref{fig:squashing-edges}. 
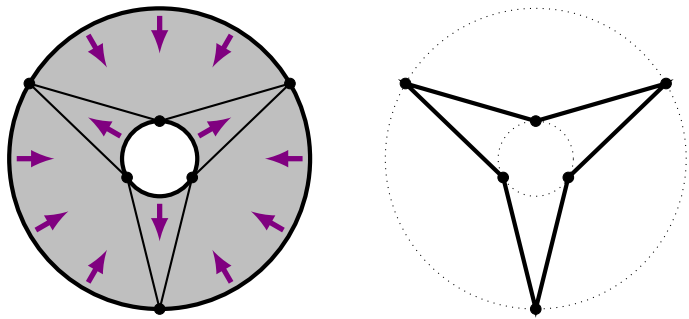
\begin{figure}[!htbp]
\begin{center}
\begin{tikzpicture}[>=latex,decoration={zigzag,amplitude=.5pt,segment length=2pt}]
\draw [ultra thick, fill=gray!50!white] (0,0) circle (2);
\draw [ultra thick, fill=white] (0,0) circle (0.5);
\draw [thick] (0,-2) -- (0.866025/2,-0.25) -- (0.866025*2,1) -- (0,0.5) -- (-0.866025*2,1) -- (-0.866025/2,-0.25) -- cycle;
\filldraw (0,0.5) circle (2pt) (-0.866025/2,-0.25) circle (2pt) (0.866025/2,-0.25) circle (2pt) (0,-2) circle (2pt) (-0.866025*2,1) circle (2pt) (0.866025*2,1) circle (2pt);
\draw [dotted] (5,0) circle (2);
\draw [dotted] (5,0) circle (0.5);
\draw [ultra thick] (5,-2) -- (5+0.866025/2,-0.25) -- (5+0.866025*2,1) -- (5,0.5) -- (5-0.866025*2,1) -- (5-0.866025/2,-0.25) -- cycle;
\filldraw (5,0.5) circle (2pt) (5-0.866025/2,-0.25) circle (2pt) (5+0.866025/2,-0.25) circle (2pt) (5,-2) circle (2pt) (5-0.866025*2,1) circle (2pt) (5+0.866025*2,1) circle (2pt);
\draw [violet,line width=2pt, ->] (0,1.9) -- (0,1.4); 
\draw [rotate=30,violet,line width=2pt, ->] (0,1.9) -- (0,1.4); 
\draw [rotate=-30,violet,line width=2pt, ->] (0,1.9) -- (0,1.4); 
\draw [rotate=90,violet,line width=2pt, ->] (0,1.9) -- (0,1.4); 
\draw [rotate=-90,violet,line width=2pt, ->] (0,1.9) -- (0,1.4); 
\draw [rotate=120,violet,line width=2pt, ->] (0,1.9) -- (0,1.4); 
\draw [rotate=-120,violet,line width=2pt, ->] (0,1.9) -- (0,1.4); 
\draw [rotate=150,violet,line width=2pt, ->] (0,1.9) -- (0,1.4); 
\draw [rotate=-150,violet,line width=2pt, ->] (0,1.9) -- (0,1.4); 
\draw [rotate=60,violet,line width=2pt, ->] (0,0.6) -- (0,1.1); 
\draw [rotate=-60,violet,line width=2pt, ->] (0,0.6) -- (0,1.1); 
\draw [rotate=180,violet,line width=2pt, ->] (0,0.6) -- (0,1.1); 
\end{tikzpicture}
\caption{All triangle are squashed so the ring is reduced to a closed path homotopic to a circle.}
\label{fig:squashing-edges}
\end{center}
\end{figure}
Strictly speaking, however, these are not really boundaries, since there is indeed a hole and not a face of which these are boundaries.\footnote{here one has to be careful: if we take 
a sphere with a hole, its boundary is not a boundary of the hole, but it is for its complement. In the above example it is not so.} This construction shows that in place of counting all edges, one is interested in counting how much unions of edges forming closed paths
are independent in the above sense (so unions of closed paths are not boundaries). \\
The same construction can be done for counting vertices. If two vertices $(a)$ and $(b)$ are the two tips of an edge, then, in place of counting 2 vertices and 1 edge we can just count 1 vertex and 0 edges, without changing $\chi$. In this case
we can restate this by saying that $(a)\sim (b)$
\begin{align}
\partial (ab)=(b)-(a)
\end{align}
and saying that, again, boundaries are trivial: $\partial (ab)\sim 0$, $(a)\sim (b)$. Notice that if a union $L=\sum_j e_j$ of edges forms a closed path than $\partial  L=0$. In particular, then, $\partial (\partial P)=0$ for any polygon $P$. Finally, we notice that
if a connected surface has a boundary, then by collapsing triangles starting from the boundary, one can always reduce to zero the number of faces, while if it has no boundary (like a sphere or a torus) than we can reduce the number of faces to one. This is exactly 
the number of closed surfaces that are not boundaries (obviously for dimensional reasons, things could change if we worked, for example, in three dimensions).\\
We can summarize such discussion as follows. We call a $j-chain$ the finite union of $j$-dimensional simplexes, with relative coefficients:
\begin{align}
 c=\sum_{a=1}^N b_a s_a, \qquad {\rm dim}(s_a)=j.
\end{align}
The sign of $b_a$ establishes the orientation and its modulus is the ``number of repetitions''. So, the set of chains is a linear space over $\mathbb Z$. We say that $c$ is closed if $\partial c=0$ and that it is exact if $c=\partial b$ for $b$ a $(j+1)$-chain.
We are then interested in counting the closed chains that are independent with respect to the relation that two close chains are equivalent if they differ by an exact chain. This space is called the $j$-th simplicial homology group
\begin{align}
 H_j(S,\mathbb Z),
\end{align}
of the surface $S$. Then, we can define the {\it Betty numbers} $b_j={\rm dim}H_j$ so that the above reasoning leads us to the result
\begin{align}
 \chi_S=\sum_{j=0}^2 (-1)^j b_j.
\end{align}
This result does not change if we change the coefficients allowing $b_a$ to take value in $\mathbb K=\mathbb Q, \mathbb R, \mathbb C$. In this case the Homology groups become vector spaces. Also, we can consider the cohomology groups by replacing
simplexes $s_k$ with continuous maps ($j$-cochain)
\begin{align}
 \sigma_k:s_k\rightarrow \mathbb K,
\end{align}
and the boundary operators $\partial$ with the coboundary operators $\delta$, sending $j$-cochains in $(j+1)$-cochains, defined by
\begin{align}
 \delta(\sigma_j) (c_{j+1})\equiv \sigma_j(\partial c_{j+1}).
\end{align}
This defines the cohomology groups $H^j(S,\mathbb K)$. It follows that $H^j(S,\mathbb K)$ is the linear dual of $H_j(S,\mathbb K)$, so that are isomorphic as vector spaces.

\subsection{Morse theory}
Another way of understanding the topology of a differentiable manifold $M$ is to use properties of functions $f:M\rightarrow \mathbb R$. Assume these functions are smooth, with non degenerate isolated critical points. A critical point is a point $p$ where
$df(p)=0$. It is non degenerate if its Hessian is different from zero. We will not delve here into the details (see \cite{milnor2016morse} for these), but we want just to show how such function may capture the topological properties by looking at the explicit example of a torus. Referring to Fig. \ref{fig:morse}
let consider the map that at $p$ associates the quote $z=f(p)$. The critical points are the points where the horizontal plane is tangent to the surface. We see that there are 4 critical points: a maximum, a minimum and two saddle points.
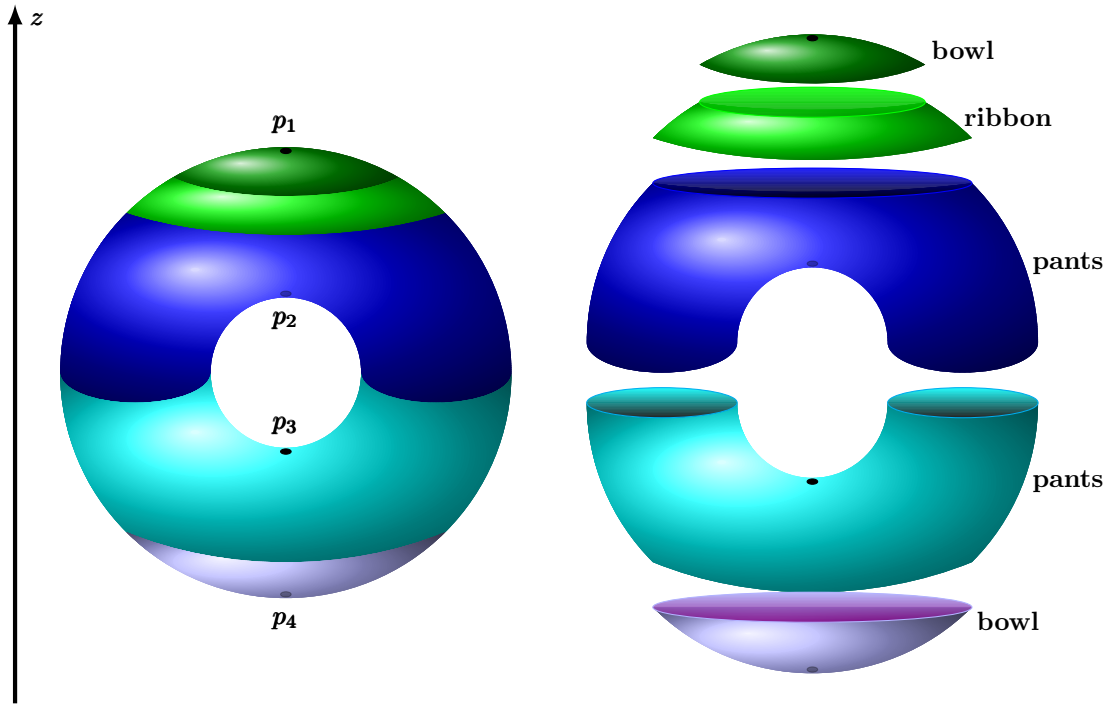
\begin{figure}[!htbp]
\begin{center}
\begin{tikzpicture}[>=latex,decoration={zigzag,amplitude=.5pt,segment length=2pt}]
\draw [ultra thick,->] (-3.6,-4.4) -- (-3.6,4.9);
\shade [ball color=green!60!black] (0,3) arc (90:120:3) arc (-150:-30:1.73205 and 0.5) arc (60:90:3) -- cycle;
\shade [ball color=green] (2.12132,2.12132) arc (45:60:3) arc (-30:-150:1.73205 and 0.5) arc (120:135:3) arc (-135:-45:3 and 1) -- cycle;
\shade [ball color=blue] (2.12132,2.12132) arc (-45:-135:3 and 1) arc (135:180:3) arc (180:360:1 and 0.4) arc (180:0:1) arc (180:360:1 and 0.4) arc (0:45:3) -- cycle;
\shade [ball color=cyan] (-3,0) arc (180:360:1 and 0.4) arc (180:360:1) arc (180:360:1 and 0.4) arc (0:-45:3) arc (-60:-120:2*2.12132 and 3) arc (225:180:3) -- cycle;
\shade [ball color=blue!30!white] (2.12132,-2.12132) arc (-60:-120:2*2.12132 and 3) arc (225:315:3) -- cycle;
\shade [ball color=blue] (9.12132,2.52132) arc (-45:-135:3 and 0.4) arc (135:180:3) arc (180:360:1 and 0.4) arc (180:0:1) arc (180:360:1 and 0.4) arc (0:45:3) -- cycle;
\shade [ball color=cyan] (4,-0.4) arc (180:360:1 and 0.2) arc (180:360:1) arc (180:360:1 and 0.2) arc (0:-45:3) arc (-60:-120:2*2.12132 and 3) arc (225:180:3) -- cycle;
\shade [ball color=blue!30!white] (9.12132,-3.12132) -- (7-2.12132,-3.12132)  arc (225:315:3) -- cycle;
\shade [ball color=green] (9.12132,3.12132) arc (45:60:3) arc (-30:-150:1.73205 and 0.2) arc (120:135:3) arc (-135:-45:3 and 1) -- cycle;
\shade [ball color=green!60!black] (7,4.5) arc (90:120:3) arc (-150:-30:1.73205 and 0.5) arc (60:90:3) -- cycle;
\draw [cyan, top color=cyan, bottom color=black](5,-0.4) ellipse (1 and 0.2);
\draw [cyan, top color=cyan, bottom color=black](9,-0.4) ellipse (1 and 0.2);
\draw [blue!30!white, top color=blue!30!white, bottom color=violet](7,-3.12132) ellipse (2.12132 and 0.2);
\draw [blue!, top color=blue, bottom color=black](7,2.52132) ellipse (2.12132 and 0.2);
\draw [green, top color=green, bottom color=green!60!black] (7,1.5*1.73205+1) ellipse (1.5 and 0.2);
\filldraw (0,2.95) ellipse (2pt and 1pt);
\filldraw [opacity=0.3] (0,1.05) ellipse (2pt and 1pt);
\filldraw (0,-1.05) ellipse (2pt and 1pt);
\filldraw [opacity=0.3] (0,-2.95) ellipse (2pt and 1pt);
\node at (-3.3,4.7) {$\pmb z$};
\node at (0,3.3) {$\pmb {p_1}$};
\node at (0,0.7) {$\pmb {p_2}$};
\node at (0,-0.7) {$\pmb {p_3}$};
\node at (0,-3.3) {$\pmb {p_4}$};
\filldraw (7,4.45) ellipse (2pt and 1pt);
\filldraw [opacity=0.3] (7,1.45) ellipse (2pt and 1pt);
\filldraw (7,-1.45) ellipse (2pt and 1pt);
\filldraw [opacity=0.3] (7,-3.95) ellipse (2pt and 1pt);
\node at (9,4.3) {\bf{bowl}};
\node at (9.6,3.4) {\bf{ribbon}};
\node at (10.4,1.45) {\bf{pants}};
\node at (10.4,-1.45) {\bf{pants}};
\node at (9.6,-3.3) {\bf{bowl}};
\end{tikzpicture}
\caption{A torus is cut in pieces with different topologies. Notice that the ribbon can be pasted to the bowl or to the pants without change the topology neither of the bowl or the pants: cutting away pieces not containing critical points does not change the topology.}
\label{fig:morse}
\end{center}
\end{figure}
Let us see how this function can give us a new way of reconstruct the surfaces starting from pieces. Starting from above, assume we use the horizontal plane to cut the surface. First, we will meet the top $p_1$ of the surface, corresponding to the maximum.
If we cut a little bit below, we get a small bowl faced down. Then, let us move below. If we cut before meeting the second critical point, we cut out a sort of ribbon, whose boundaries are homologically equivalent. This will happen each time we take two cuts
not containing any critical point in between. Then we move even below until meeting the saddle point. If we cut a little bit below we get a piece that looks like a pair of pants. If we go even more below, until passing the second saddle point $p_3$, we get a second
pair of pants (with the legs upside).\footnote{If we cut before passing the saddle point we get two cylinders, which correspond to trivial pieces} Going below the point $p_4$ we are left with a second bowl. Forgetting the trivial pieces, we see that we can 
reconstruct the surface gluing together the shapes obtained after cutting horizontally around the critical points. A maximum gives a bowl down, any saddle point gives a pair of pants and a minimum gives a bowl up. Each of these shapes is understood by looking
at the signature of the Hessian at each critical point. Starting from above, we have just to look at the principal ways to go down (Fig \ref{fig:morse-1}): 
\begin{itemize}
 \item at $p_1$ the Hessian is negative, so it has two principal directions going down (the eigenvectors). We say that $p_1$ has Morse index $2$;
 \item at $p_2$ and $p_3$ the Hessian is indefinite, it has only one negative eigenvalue at which it corresponds a descending direction. We say that $p_2$ and $p_3$ have Morse index $1$;
 \item at $p_4$ the Hessian is positive. There are not descending directions and we say that $p_4$ has Morse index $0$.
\end{itemize}
If $m(n)$ is the number of critical points with Morse index $n$, then we can define the Euler number associated to our surface $S$ ($M$ stays for Morse):
\begin{align}
 \chi_M (S)=\sum_{n=0}^2 (-1)^n m(n).
\end{align}
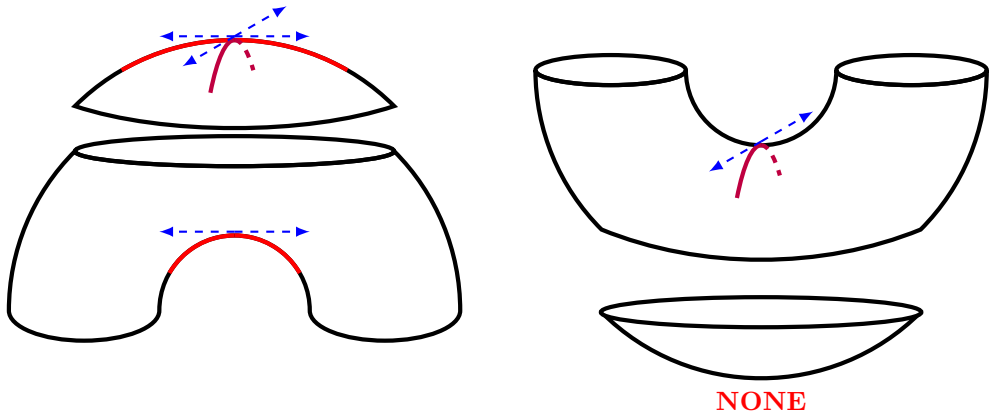
\begin{figure}[!htbp]
\begin{center}
\begin{tikzpicture}[>=latex,decoration={zigzag,amplitude=.5pt,segment length=2pt},scale=1]
\draw [ultra thick](2.12132,2.52132) arc (0:-180:2.12132 and 0.2 ) arc (135:180:3) arc (180:360:1 and 0.4) arc (180:0:1) arc (180:360:1 and 0.4) arc (0:45:3) -- cycle;
\draw [ultra thick](4,3.6) arc (180:360:1 and 0.2) arc (180:360:1) arc (180:360:1 and 0.2) arc (0:-45:3) arc (-60:-120:2*2.12132 and 3) arc (225:180:3) -- cycle;
\draw [ultra thick](7-2.12132,0.5-0.12132)  arc (225:315:3);
\draw [ultra thick] (2.12132,3.12132) arc (45:135:3) arc (-135:-45:3 and 1) -- cycle;
\draw [ultra thick](5,3.6) ellipse (1 and 0.2);
\draw [ultra thick](9,3.6) ellipse (1 and 0.2);
\draw [ultra thick,red] (0,1.4) arc (90:30:1);
\draw [ultra thick,red] (0,1.4) arc (90:150:1);
\draw [thick,blue,dashed,->] (0,1.45) -- (1,1.45);
\draw [thick,blue,dashed,->] (0,1.45) -- (-1,1.45);
\draw [ultra thick,red] (0,4) arc (90:60:3);
\draw [ultra thick,red] (0,4) arc (90:120:3);
\draw [ultra thick,purple] (0,4) arc (90:130:0.5 and 3);
\draw [ultra thick,purple, dashed] (0,4) arc (90:60:0.5 and 3);
\draw [thick,blue,dashed,->] (0,4.05) -- (1,4.05);
\draw [thick,blue,dashed,->] (0,4.05) -- (-1,4.05);
\draw [thick,blue,dashed,->] (0,4.05) -- (0.8*0.866,4.05+0.4);
\draw [thick,blue,dashed,->] (0,4.05) -- (-0.8*0.866,4.05-0.4);
\draw [ultra thick,purple] (7,2.6) arc (90:130:0.5 and 3);
\draw [ultra thick,purple, dashed] (7,2.6) arc (90:60:0.5 and 3);
\draw [thick,blue,dashed,->] (7,2.65) -- (7+0.8*0.866,2.65+0.4);
\draw [thick,blue,dashed,->] (7,2.65) -- (7-0.8*0.866,2.65-0.4);
\draw [ultra thick,fill=white](7,3.5-3.12132) ellipse (2.13 and 0.2);
\draw [ultra thick](0,2.52132) ellipse (2.12132 and 0.2);
\node [red] at (7,-0.8) {\bf {NONE}};
\end{tikzpicture}
\caption{Principal descending directions from the critical points at any piece.}
\label{fig:morse-1}
\end{center}
\end{figure}
Notice that in our example we get $\chi_M(S)=0$. If we deform the surface without changing the topology, like in figure \ref{fig:deform-torus}, the number of critical points can change, but $\chi_M(S)$ remains unchanged. 
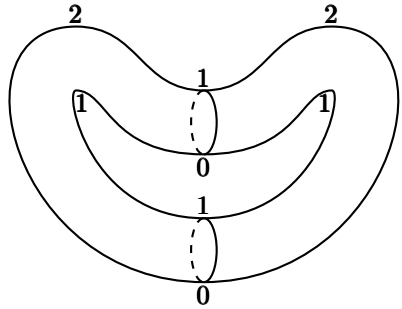
\begin{figure}[!htbp]
\begin{center}
\begin{tikzpicture}[>=latex,decoration={zigzag,amplitude=.5pt,segment length=2pt},scale=0.85]
\draw [thick] (0,0) .. controls (1,0) and (1,1) .. (2,1) .. controls (4,1) and (3,-3) .. (0,-3) .. controls (-3,-3) and (-4,1) .. (-2,1) .. controls (-1,1) and (-1,0) .. (0,0) -- cycle;
\draw [thick] (0,-1) .. controls (1.5,-1) and (1.5,0) .. (2,0)  .. controls (2.2,0) and (1.8,-2) .. (0,-2) .. controls (-1.8,-2) and (-2.2,0) .. (-2,0) .. controls (-1.5,0) and (-1.5,-1) .. (0,-1) -- cycle;
\node at (0,-3.2) {$\pmb 0$};
\node at (2,1.2) {$\pmb 2$};
\node at (-2,1.2) {$\pmb 2$};
\node at (0,0.2) {$\pmb 1$};
\node at (0,-1.2) {$\pmb 0$};
\node at (0,-1.8) {$\pmb 1$};
\node at (-1.9,-0.2) {$\pmb 1$};
\node at (1.9,-0.2) {$\pmb 1$};
\draw [thick] (0,0) arc (90:-90:0.2 and 0.5);
\draw [thick,dashed] (0,0) arc (90:270:0.2 and 0.5);
\draw [thick] (0,-2) arc (90:-90:0.2 and 0.5);
\draw [thick,dashed] (0,-2) arc (90:270:0.2 and 0.5);
\end{tikzpicture}
\caption{Deforming the torus does not changes the Morse index.}
\label{fig:deform-torus}
\end{center}
\end{figure}
It is a topological invariant and, in this case,
it is zero, equal to the Euler characteristic of the torus. It is possible to prove that $\chi_M(S)$ is always equal to the Euler characteristic, in any dimension, for any closed compact manifold. In the next picture we can see that, indeed, for a sphere $S^2$ we get 
$\chi_M(S^2)=2$ and for a surface $\Sigma_g$ of genus $g$ one has $\chi_M(\Sigma_g)=2-2g$. 

\subsection{Cellular (co)homology}
Another way to look at homology (and then cohomology by duality, as usual) is to look for a cellular decomposition of our surface, \cite{dubrovin1984modern}. 
This means that we have to obtain the surface (up to deformations preserving the topology) by gluing discs of different dimensions.
A disc is meant to be a full ball in given dimension. So a 0-dimensional disc is a point and a 1-dimensional disc is a segment. The rule is that one starts from the lower dimensional discs, then glues the boundaries of the successive dimensional discs to the 
lower dimensional structure obtained previously. For example, one can get a sphere starting from a 0-cell and a 2-cell, gluing the boundary of the two cell to the point (Fig. \ref{fig:cellular-sphere1}). 
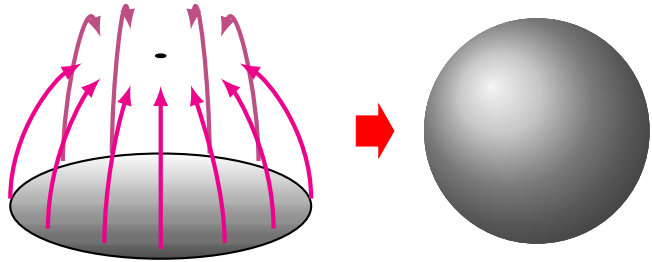
\begin{figure}[!htbp]
\begin{center}
\begin{tikzpicture}[>=latex,decoration={zigzag,amplitude=.5pt,segment length=2pt}]
\draw [ultra thick, color=magenta!75!black, ->] (-6.3,0.6) .. controls (-6.3,1.5) and (-6,2.8) .. (-5.8,2.25);
\draw [ultra thick, color=magenta!75!black, ->] (-3.7,0.6) .. controls (-3.7,1.5) and (-4,2.8) .. (-4.2,2.25);
\draw [ultra thick, color=magenta!75!black, ->] (-5.65,0.7) .. controls (-5.65,1.6) and (-5.5,3) .. (-5.4,2.35);
\draw [ultra thick, color=magenta!75!black, ->] (-4.35,0.7) .. controls (-4.35,1.6) and (-4.5,3) .. (-4.6,2.35);
\filldraw (-5,2) ellipse (2pt and 0.7pt);
\draw [thick, top color= white, bottom color=black!70!white] (-5,0) ellipse (2 and 0.7);
\draw [ultra thick, magenta, ->] (-6.5,-0.3) .. controls (-6.5,0.7) and (-6,1.5) .. (-5.8,1.7);
\draw [ultra thick, magenta, ->] (-3.5,-0.3) .. controls (-3.5,0.7) and (-4,1.5) .. (-4.2,1.7);
\draw [ultra thick, magenta, ->] (-5,-0.57) -- (-5,1.6);
\draw [ultra thick, magenta, ->] (-5.75,-0.49) .. controls (-5.75,0.2) and (-5.6,1) .. (-5.4,1.61);
\draw [ultra thick, magenta, ->] (-4.15,-0.49) .. controls (-4.15,0.2) and (-4.4,1) .. (-4.6,1.61);
\draw [ultra thick, magenta, ->] (-7,0.1) .. controls (-7,1) and (-6.4,1.7) .. (-6.05,1.9);
\draw [ultra thick, magenta, ->] (-3,0.1) .. controls (-3,1) and (-3.6,1.7) .. (-3.95,1.9);
\shade[ball color=black!30!white] (0,1) circle (1.5);
\filldraw [red] (-2.4,1.2) -- (-2.1,1.2) -- (-2.1,1.35) -- (-1.9,1) -- (-2.1,0.65) -- (-2.1, 0.8) -- (-2.4,0.8) -- cycle;
\end{tikzpicture}
\caption{A cellular decomposition of the sphere: a disc is glued to a point.}
\label{fig:cellular-sphere1}
\end{center}
\end{figure}
Or starting from a point, a segment and 2 discs, first gluing the boundaries of the segments 
to the point (getting $S^1$) and gluing the boundaries of the two 2-cells to $S^1$ (Fig \ref{fig:cellular-sphere2}). 
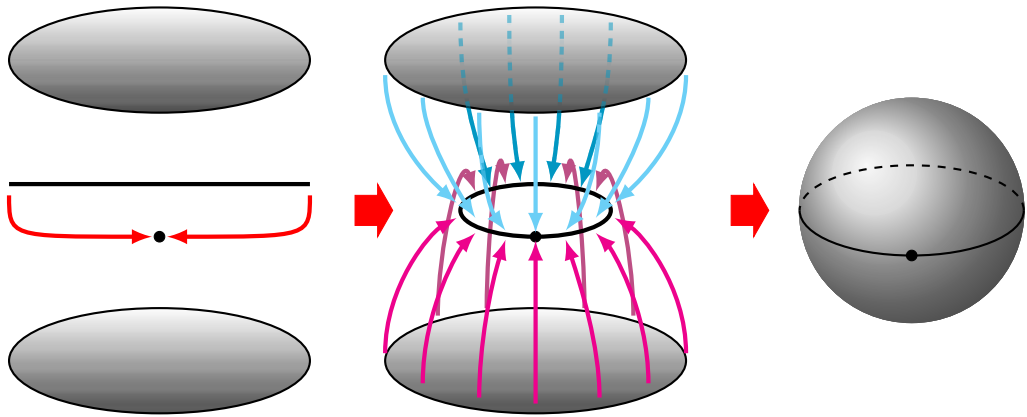
\begin{figure}[!htbp]
\begin{center}
\begin{tikzpicture}[>=latex,decoration={zigzag,amplitude=.5pt,segment length=2pt}]
\filldraw (-5,1.65) circle (2pt);
\draw [ultra thick, color=cyan!75!black, ->] (-4.65,4.5) .. controls (-4.65,3.5) and (-4.7,3) .. (-4.8,2.38);
\draw [ultra thick, color=cyan!75!black, ->] (-5.35,4.5) .. controls (-5.35,3.5) and (-5.3,3) .. (-5.2,2.38);
\draw [ultra thick, color=cyan!75!black, ->] (-4,4.5) .. controls (-4,3.5) and (-4.2,3) .. (-4.4,2.3);
\draw [ultra thick, color=cyan!75!black, ->] (-6,4.5) .. controls (-6,3.5) and (-5.8,3) .. (-5.6,2.3);
\draw [ultra thick, color=magenta!75!black, ->] (-6.3,0.6) .. controls (-6.3,1.5) and (-6,2.8) .. (-5.8,2.25);
\draw [ultra thick, color=magenta!75!black, ->] (-3.7,0.6) .. controls (-3.7,1.5) and (-4,2.8) .. (-4.2,2.25);
\draw [ultra thick, color=magenta!75!black, ->] (-5.65,0.7) .. controls (-5.65,1.6) and (-5.5,3) .. (-5.4,2.35);
\draw [ultra thick, color=magenta!75!black, ->] (-4.35,0.7) .. controls (-4.35,1.6) and (-4.5,3) .. (-4.6,2.35);
\draw [ultra thick] (-5,2) ellipse (1 and 0.35);
\draw [thick, top color= white, bottom color=black!70!white] (-5,0) ellipse (2 and 0.7);
\draw [thick, top color= white, bottom color=black!70!white] (-5,4) ellipse (2 and 0.7);
\draw [ultra thick, magenta, ->] (-6.5,-0.3) .. controls (-6.5,0.7) and (-6,1.5) .. (-5.8,1.7);
\draw [ultra thick, magenta, ->] (-3.5,-0.3) .. controls (-3.5,0.7) and (-4,1.5) .. (-4.2,1.7);
\draw [ultra thick, magenta, ->] (-5,-0.57) -- (-5,1.6);
\draw [ultra thick, magenta, ->] (-5.75,-0.49) .. controls (-5.75,0.2) and (-5.6,1) .. (-5.4,1.61);
\draw [ultra thick, magenta, ->] (-4.15,-0.49) .. controls (-4.15,0.2) and (-4.4,1) .. (-4.6,1.61);
\draw [ultra thick, magenta, ->] (-7,0.1) .. controls (-7,1) and (-6.4,1.7) .. (-6.05,1.9);
\draw [ultra thick, magenta, ->] (-3,0.1) .. controls (-3,1) and (-3.6,1.7) .. (-3.95,1.9);
\draw [ultra thick, cyan!50!white, ->] (-6.5,3.5) .. controls (-6.5,3) and (-6,2.2) .. (-5.8,1.9);
\draw [ultra thick, cyan!50!white, ->] (-3.5,3.5) .. controls (-3.5,3) and (-4,2.2) .. (-4.2,1.9);
\draw [ultra thick, cyan!50!white, ->] (-5,3.25) -- (-5,1.7);
\draw [ultra thick, cyan!50!white, ->] (-5.75,3.3) .. controls (-5.75,2.5) and (-5.6,2) .. (-5.4,1.73);
\draw [ultra thick, cyan!50!white, ->] (-4.15,3.3) .. controls (-4.15,2.5) and (-4.4,2) .. (-4.6,1.73);
\draw [ultra thick, cyan!50!white, ->] (-7,3.8) .. controls (-7,3) and (-6.4,2.4) .. (-6.05,2.1);
\draw [ultra thick, cyan!50!white, ->] (-3,3.8) .. controls (-3,3) and (-3.6,2.4) .. (-3.95,2.1);
\draw [opacity=0.5,ultra thick, color=cyan!75!black,dashed, ->] (-4.65,4.6) .. controls (-4.65,3.5) and (-4.7,3) .. (-4.8,2.38);
\draw [opacity=0.5,ultra thick, color=cyan!75!black,dashed, ->] (-5.35,4.6) .. controls (-5.35,3.5) and (-5.3,3) .. (-5.2,2.38);
\draw [opacity=0.5,ultra thick, color=cyan!75!black,dashed, ->] (-4,4.5) .. controls (-4,3.5) and (-4.2,3) .. (-4.4,2.3);
\draw [opacity=0.5,ultra thick, color=cyan!75!black,dashed, ->] (-6,4.5) .. controls (-6,3.5) and (-5.8,3) .. (-5.6,2.3);
\shade[ball color=black!20!white] (0,2) circle (1.5);
\draw [thick] (1.49,2) arc (0:-180:1.49 and 0.6);
\draw [thick,dashed] (1.49,2) arc (0:180:1.49 and 0.6);
\filldraw (0,1.4) circle (2pt);
\filldraw [red] (-2.4,2.2) -- (-2.1,2.2) -- (-2.1,2.35) -- (-1.9,2) -- (-2.1,1.65) -- (-2.1, 1.8) -- (-2.4,1.8) -- cycle;
\draw [thick, top color= white, bottom color=black!70!white] (-10,0) ellipse (2 and 0.7);
\draw [thick, top color= white, bottom color=black!70!white] (-10,4) ellipse (2 and 0.7);
\filldraw (-10,1.65) circle (2pt);
\draw [ultra thick] (-12,2.35) -- (-8,2.35);
\draw [red,ultra thick,->] (-12,2.2) .. controls (-12,1.65) and (-12,1.65) .. (-10.1,1.65); 
\draw [red,ultra thick,->] (-8,2.2) .. controls (-8,1.65) and (-8,1.65) .. (-9.9,1.65); 
\filldraw [red] (-7.4,2.2) -- (-7.1,2.2) -- (-7.1,2.35) -- (-6.9,2) -- (-7.1,1.65) -- (-7.1, 1.8) -- (-7.4,1.8) -- cycle;
\end{tikzpicture}
\caption{A different cellular decomposition of the sphere.}
\label{fig:cellular-sphere2}
\end{center}
\end{figure}
These are different constructions of $S^2$ gluing discs, and we can think many others. The point is that if we indicate with $d(n)$ the numbers of cells of dimension $n$
then we get for the sphere
\begin{align}
 \chi_{CW}(S^2):=\sum_{n=0}^2 (-1)^n d(n)=2
\end{align}
independently on the construction. Notice that, again, it coincides with the Euler characteristic, and, again, it is not by chance but it is a general result: the cellular decomposition gives us another way to compute the Euler characteristic. In Figure \ref{fig:cell-g} it
is shown how to get a genus $g$ surface $\Sigma_g$ by gluing a point, $2g$ segments and a 2-disc. 

\begin{figure}[!htbp]
\begin{center}
\begin{tikzpicture}[>=latex,decoration={zigzag,amplitude=.5pt,segment length=2pt},scale=0.85]
\filldraw [red](-5,7.5) circle (2pt); \node [red] at (-4.7,7.5) {$\pmb p$};
\draw [ultra thick, green!50!blue] (-6.5,7) -- (-3.5,7); \node [green!50!blue] at (-3.2,7) {$\pmb {a_1}$}; \node [green!50!blue] at (-1.7,6.5) {$\pmb {a_1}$};
\draw [ultra thick, green!50!blue] (-6.5,6.5) -- (-3.5,6.5); \node [green!50!blue] at (-3.2,6.5) {$\pmb {a_2}$}; \node [green!50!blue] at (0.3,7) {$\pmb {a_2}$};
\draw [ultra thick, green!50!blue] (-6.5,6) -- (-3.5,6); \node [green!50!blue] at (-3.2,6) {$\pmb {a_3}$}; \node [green!50!blue] at (1.3,5.6) {$\pmb {a_g}$};
\draw [thick,dotted, green!50!blue] (-6.5,5.5) -- (-3.5,5.5);
\draw [ultra thick, green!50!blue] (-6.5,5) -- (-3.5,5); \node [green!50!blue] at (-3.2,5) {$\pmb {a_g}$};
\draw [thick,top color=white, bottom color=gray] (-5,2.5) circle (2);
\draw [ultra thick,cyan] (-6.5,0) -- (-3.5,0); \node [cyan] at (-3.2,0) {$\pmb {b_1}$};
\draw [ultra thick,cyan] (-6.5,-0.5) -- (-3.5,-0.5); \node [cyan] at (-3.2,-0.5) {$\pmb {b_2}$};
\draw [ultra thick,cyan] (-6.5,-1) -- (-3.5,-1); \node [cyan] at (-3.2,-1) {$\pmb {b_3}$};
\draw [thick,dotted,cyan] (-6.5,-1.5) -- (-3.5,-1.5);
\draw [ultra thick,cyan] (-6.5,-2) -- (-3.5,-2); \node [cyan] at (-3.2,-2) {$\pmb {b_g}$};
\draw [white,top color=white, bottom color=gray] (0,0) circle (2);
\draw [ultra thick, green!50!blue](0,4.2) .. controls (-4,4.2) and (-1,8.5) .. (0,4.2);
\draw [ultra thick, green!50!blue](0,4.2) .. controls (-1,7.5) and (1,7.5) .. (0,4.2);
\draw [ultra thick, green!50!blue](0,4.2) .. controls (1,6.5) and (3,4.4) .. (0,4.2);
\draw [ultra thick, cyan](0,4.2) .. controls (-4,4.2) and (-1,2) .. (0,4.2);
\draw [ultra thick, cyan](0,4.2) .. controls (-1,1.9) and (1,1.9) .. (0,4.2);
\draw [ultra thick, cyan](0,4.2) .. controls (1,1.9) and (3,4) .. (0,4.2);
\node [cyan] at (-1.7,2.95) {$\pmb {b_1}$};
\node [cyan] at (-0.6,2.86) {$\pmb {b_2}$};
\node [cyan] at (1.58,2.82) {$\pmb {b_g}$};
\filldraw [red](0,4.2) circle (2pt); \node [red] at (-0.3,4.4) {$\pmb p$};
\draw [ultra thick, cyan] (0,2) arc (90:60:2);
\draw [rotate=60,ultra thick, cyan] (0,2) arc (90:60:2);
\draw [rotate=-60,ultra thick, cyan] (0,2) arc (90:60:2);
\draw [rotate=-120,ultra thick, cyan,dashed] (0,2) arc (90:60:2);
\draw [rotate=30, ultra thick, green!50!blue] (0,2) arc (90:60:2);
\draw [rotate=-30, ultra thick, green!50!blue] (0,2) arc (90:60:2);
\draw [rotate=-90,ultra thick, green!50!blue] (0,2) arc (90:60:2);
\draw [rotate=90,ultra thick, green!50!blue, dashed] (0,2) arc (90:60:2);
\filldraw [red] (0,2) circle (2pt);
\filldraw [rotate=30,red] (0,2) circle (2pt);
\filldraw [rotate=60,red] (0,2) circle (2pt);
\filldraw [rotate=90,red] (0,2) circle (2pt);
\filldraw [rotate=-30,red] (0,2) circle (2pt);
\filldraw [rotate=-60,red] (0,2) circle (2pt);
\filldraw [rotate=-90,red] (0,2) circle (2pt);
\filldraw [rotate=-120,red] (0,2) circle (2pt);
\filldraw [rotate=-150,red] (0,2) circle (2pt);
\draw [green!50!blue, rotate=15,ultra thick,->] (0,2) -- (0.2,2); 
\draw [cyan, rotate=-15,ultra thick,->] (0,2) -- (0.2,2); 
\draw [green!50!blue, rotate=-105,ultra thick,->] (0,2) -- (0.2,2); 
\draw [cyan, rotate=-135,ultra thick,->] (0,2) -- (0.2,2); 
\draw [cyan, rotate=45,ultra thick,->] (0,2) -- (-0.2,2); 
\draw [green!50!blue, rotate=75,ultra thick,->] (0,2) -- (-0.2,2); 
\draw [green!50!blue, rotate=-45,ultra thick,->] (0,2) -- (-0.2,2); 
\draw [cyan, rotate=-75,ultra thick,->] (0,2) -- (-0.2,2); 
\draw [ultra thick,green!50!blue,->] (1.1,5.25) -- (1.3,5.25); 
\draw [ultra thick,green!50!blue,->] (-1.4,6.1) -- (-1.2,6.1); 
\draw [ultra thick,green!50!blue,->] (0.3,5.9) -- (0.3,5.7); 
\draw [ultra thick,cyan,->] (1.2,3.14) -- (1,3.14); 
\draw [ultra thick,cyan,->] (-1.2,3.22) -- (-1.4,3.22); 
\draw [ultra thick,cyan,->] (0.3,3) -- (0.3,2.8); 
\draw [dotted,thick,green!50!blue] (0.5,6) .. controls (0.7,5.8) and (0.9,5.4) .. (0.9,5.3);
\draw [dotted,thick,cyan] (0.4,2.7) .. controls (0.7,2.8) and (0.9,3.2) .. (0.9,3.3);
\draw [white, top color=gray, bottom color=white](5,2) .. controls (2,-1) and (2,6.5) .. (5,3.5) .. controls (6,2.5) and (7,2.5) .. (6,3.5) .. controls (3,6.5) and (10.5,6.5) .. (7.5,3.5) .. controls (6.5,2.5) and (7,2.5) .. (8,3.5) .. controls (11,6.5) and (11,-1) .. 
(8,2) .. controls (7.5,2.5) and (8,1) .. (6.5,1) .. controls (5,1) and (5.5,2.5) .. (5,2) -- cycle;
\draw [thick](5,2) .. controls (2,-1) and (2,6.5) .. (5,3.5) .. controls (6,2.5) and (7,2.5) .. (6,3.5) .. controls (3,6.5) and (10.5,6.5) .. (7.5,3.5) .. controls (6.5,2.5) and (7,2.5) .. (8,3.5) .. controls (11,6.5) and (11,-1) .. (8,2);
\draw [dotted,thick] (5,2) .. controls (5.5,2.5) and (5,1) .. (6.5,1) .. controls (8,1) and (7.5,2.5) .. (8,2); 
\filldraw [white](6.197,4.63) .. controls (6.5,4.8) and (7,4.8) .. (7.303,4.63) .. controls (6.9,4.37) and (6.6,4.37) .. (6.197,4.63) -- cycle;
\draw (6.197,4.63) .. controls (6.5,4.8) and (7,4.8) .. (7.303,4.63);
\draw (6,4.8) .. controls (6.5,4.3) and (7,4.3) .. (7.5,4.8);
\filldraw [white](3.647,3) .. controls (3.9,3.4) and (4.4,2.6) .. (4.1,2.471) .. controls (3.7,2.26) and (3.4,2.5) .. (3.647,3) -- cycle;
\draw (3.647,3) .. controls (3.9,3.4) and (4.4,2.6) .. (4.1,2.471);
\draw (3.9,3.3) .. controls (3.3,2.8) and (3.5,2) .. (4.3,2.6);
\filldraw [white](13-3.647,3) .. controls (13-3.9,3.4) and (13-4.4,2.6) .. (13-4.1,2.471) .. controls (13-3.7,2.26) and (13-3.4,2.5) .. (13-3.647,3) -- cycle;
\draw (13-3.647,3) .. controls (13-3.9,3.4) and (13-4.4,2.6) .. (13-4.1,2.471);
\draw (13-3.9,3.3) .. controls (13-3.3,2.8) and (13-3.5,2) .. (13-4.3,2.6);
\draw [thick,cyan](6.7,2.2) .. controls (6.7,3.8) and (6.2, 4) .. (6.49,4.476);
\draw [thick,cyan](6.7,2.2) .. controls (6.6,3) and (5.9,4.1) .. (5.65,3.873);
\draw [thick,cyan,opacity=0.3](5.645,3.9) .. controls (6.35,3.6) and (6.55,3.9) .. (6.49,4.476);
\draw [thick,cyan](6.7,2.2) .. controls (13-5.8,2.6) and (13-4.5,3.1) .. (13-4.87,3.645);
\draw [thick,cyan](6.7,2.2) .. controls (13-5.8,2) and (13-4.5,3.1) .. (13-4.14,2.82);
\draw [thick,cyan,opacity=0.3](13-4.16,2.83) .. controls (13-4.6,2.7) and (13-5.2,3.18) .. (13-4.86,3.64);
\draw [thick,cyan](6.7,2.2) .. controls (5.8,2.6) and (4.5,3.1) .. (4.87,3.645);
\draw [thick,cyan](6.7,2.2) .. controls (5.8,2) and (4.5,3.1) .. (4.14,2.82);
\draw [thick,cyan,opacity=0.3](4.16,2.83) .. controls (4.6,2.7) and (5.2,3.18) .. (4.86,3.64);
\draw [thick, green!50!blue](6.7,2.2) .. controls (6.7,4.2) and (5.2,4.5) .. (5.8,5.1) .. controls (6.2,5.4) and (7.2,5.4) .. (7.6,5.1) .. controls (8.2,4.5) and (6.7,4.2) .. (6.7,2.2);
\draw [thick, green!50!blue](6.7,2.2) .. controls (5.8,1.6) and (5.3,2.7) .. (4.7,2.3) .. controls (2.6,0.8) and (3.1,5) .. (4.7,3.1) .. controls (5.1,2.725) and (5.7,2.35) .. (6.7,2.2);
\draw [thick, green!50!blue](6.7,2.2) .. controls (13-5.8,1.6) and (13-5.3,2.7) .. (13-4.7,2.3) .. controls (13-2.6,0.8) and (13-3.1,5) .. (13-4.7,3.1) .. controls (13-5.1,2.725) and (13-5.7,2.35) .. (6.7,2.2);
\filldraw [red] (6.7,2.2) circle (2pt);
\node [green!50!blue] at (-0.4,1.6) {$\pmb {a_1}$}; \node [cyan] at (0.5,1.6) {$\pmb {b_1}$};
\node [green!50!blue] at (1.2,1.25) {$\pmb {a_1}$}; \node [cyan] at (-1.05,1.2) {$\pmb {b_g}$};
\node [green!50!blue] at (1.6,-0.4) {$\pmb {a_2}$}; \node [cyan] at (1.65,0.5) {$\pmb {b_1}$};
\node [green!50!blue] at (-1.6,0.4) {$\pmb {a_g}$}; \node [cyan] at (1.2,-1.1) {$\pmb {b_2}$};
\node [green!50!blue] at (6.7,5.5) {$\pmb {a_1}$};
\node [green!50!blue] at (9.95,2.7) {$\pmb {a_2}$};
\node [green!50!blue] at (3.05,2.7) {$\pmb {a_g}$};
\node [cyan] at (6.8,4) {$\pmb {b_1}$};
\node [cyan] at (8.41,3.6) {$\pmb {b_2}$};
\node [cyan] at (13-8.41,3.6) {$\pmb {b_g}$};
\node [red] at (-2.3,0) {$\pmb p$};
\node [red] at (2.3,0) {$\pmb p$};
\node [red] at (-2,1) {$\pmb p$};
\node [red] at (2,1) {$\pmb p$};
\node [red] at (2,-1) {$\pmb p$};
\node [red] at (-1.2,1.9) {$\pmb p$};
\node [red] at (1.3,1.9) {$\pmb p$};
\node [red] at (1.3,-1.9) {$\pmb p$};
\node [red] at (0,2.25) {$\pmb p$};
\node [red] at (6.7,1.85) {$\pmb p$};
\end{tikzpicture}
\caption{A cellular decomposition of a genus $g$ surface. All the boundaries of the $2g$ one dimensional cells are glued to a point, giving a one dimensional skeleton. Then the boundary of the two dimensional disc is glued to the skeleton, gluing the red points to p and the remaining parts to the corresponding lines (two segments for each curve) respecting the drawn orientations.}
\label{fig:cell-g}
\end{center}
\end{figure}
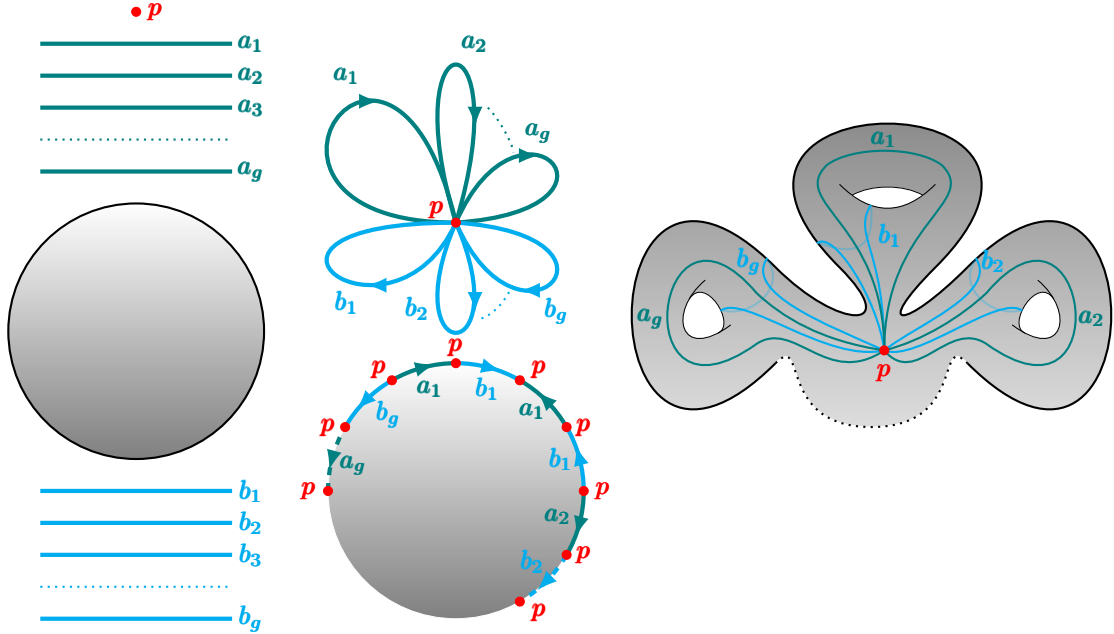

Again, we see that
\begin{align}
 \chi_{CW}(\Sigma_g)=\chi_E(\Sigma_g)=\chi_M(\Sigma_g).
\end{align}

\subsection{De Rham cohomology}
Finally, we want just to recall what is probably the most commonly known cohomology, which is the de Rham cohomology, \cite{madsen1997from}. On a smooth surface we can have 0-forms (functions), 1-forms and 2-forms. The external differential maps $k$-forms into $(k+1)$-forms.
A $k$-form $\phi$ is said to be closed if $d\phi=0$, while it is called exact if $\phi=d\psi$ for $\psi$ a $(k-1)$-form. An exact form is also closed since $d^2=0$. This is an obvious consequence of the Schwartz's lemma for double derivatives. So, exact forms
generate a subspace of closed forms. The de Rham cohomology group of degree $k$ is the space of closed $k$-forms identified up to exact forms. If $\Omega^k(S)$ are the $k$-forms:
\begin{align}
 H^k_{dR}(S)=\{ \omega \in \Omega^k(S)| d\omega=0 \}/\sim,
\end{align}
where $\omega_1\sim\omega_2$ if and only if $\omega_1-\omega_2=d\lambda$ for some $\lambda  \in \Omega^{k-1}(S)$. It turns out that, under reasonable hypotheses, $H^k_{dR}(S)$ has finite dimension as a real vector space. Its dimension 
$b_k={\dim} H^k_{dR}(S)$ is called the $k$-th Betty number of $S$ and is a topological invariant. For $k=0$ the closed forms are the locally constant functions. They are closed but not exact, so $H^0_{dR}(S)$ is the space of locally constant functions. 
If $S$ is connected, any constant function is just identified by its value, so $H^0_{dR}(S)=\mathbb R$ and $b_0=1$. On a bidimensional surface, any 2-form is closed. Suppose that $S$ is connected and closed (i.e. has no boundary) and that 
$\omega_1\sim\omega_2$ are two equivalent 2-forms. Then, $\omega_2-\omega_1=d\psi$ for $\psi$ a 1-form, and using Stokes theorem:
\begin{align}
 \int_S \omega_2-\int_S\omega_1=\int_Sd\psi=\int_{\partial S}\psi=0,
\end{align}
so two equivalent forms have the same integral. The vice versa is also true and we get that for a closed connected surface also $b_2=1$. It is difficult but possible to prove that if the surface has genus $g$, then $b_1=2g$. Since $b_j$ are invariants, it
follows that also 
\begin{align}
\chi_{dR}(S):=\sum_{n=0}^2 (-1)^n b_n 
\end{align}
is an invariant. We see that if $S=\Sigma_g$, we get 
\begin{align}
\chi_{dR}(\Sigma_g)=2-2g= \chi_{CW}(\Sigma_g)=\chi_E(\Sigma_g)=\chi_M(\Sigma_g).
\end{align}
Once more, this is a general result. We can now stop here, we will discuss a little bit further about the relation of the de Rham cohomology with other cohomologies in a more general setting (where we will also see a further cohomology, the singular cohomology).

\section{Integrals and Cohomologies}\label{HardMaths}
While calculus with functions of one real variable gives us the impression of a strict relation between differential and integral calculus, it is pretty evident that when passing to more variables, or just to one complex variable, things change drastically, and the
theory of integration looks deeply different. Indeed, the sensitivity of integration to global questions becomes quickly manifest. Integration looks more related to (co)homology and Hodge theory. In this section we want to summarise some facts we probably 
will have to consider in a full program of investigation about the geometry underlying Feynman integrals. Of course, we cannot be exhaustive, but our aim is to be at least suggestive. \\
A nice way to see it is to pass through complete elliptic integrals. They are one dimensional integrals whose structure resembles that of Feynman integrals, hiding the main characterizations we need, and are also the origin of larger dimensional integration 
theory, starting from Poincar\'e and developed by Lefschetz and Hodge \cite{alma991027076759703276,hodge1989the,Hodge1955IntegralsOT,MR0268189,whitney1957geometric}. 
\subsection{Elliptic integrals}
Consider the elliptic integrals of first and second kind, \cite{Hancock:1958:EI}:
\begin{align}
 K(k)&=\int_0^1 \frac {dx}{\sqrt{(1-x^2)(1-k^2x^2)}}=\frac 12 \int_{-1}^1 \frac {dx}{\sqrt{(1-x^2)(1-k^2x^2)}}, \\
 E(k)&=\int_0^1 \frac {\sqrt{1-k^2 x^2}}{\sqrt{1-x^2}}\ dx =\frac 12 \int_{-1}^1 \frac {\sqrt{1-k^2 x^2}}{\sqrt{1-x^2}}\ dx.
\end{align}
It is known that these integrals cannot be expressed in terms of elementary functions of $k$. In both cases the last integral extends between two branch points of the integrands. There are four branch points, at $z_\pm=\pm 1$ and $\tilde z_\pm=\pm k^{-1}$,
$k\neq \pm1$. As the integrand is a multivalued function, we can extend it to a single valued function on a Riemann surface. To this end, we can first extend $x$ to $z$ on the complex plane, with two cuts, one from $-1$ to $1$, the other from $-k^{-1}$
to $k^{-1}$. Notice that in both cases the integrand is regular at infinity, so we can think at it as defined on the Riemann sphere $\mathbb P^1$ with the two given cuts. If we try to cross one of the cuts, the function flips the sign and the only way to keep
the function single valued is to assume that we end up to a second copy of the cut Riemann sphere, where at the doubled point it takes the same value as on the original sphere, but with the opposite sign. With this standard construction we end up with
a well defined single valued function on a two dimensional surface obtained gluing the edges of each cut of a sphere to the edges of the corresponding cut on the second sphere. This is topologically a torus, that is a Riemann surface of genus two or, from 
the complex point of view, an elliptic curve, Fig. \ref{fig:elliptic curve}. 
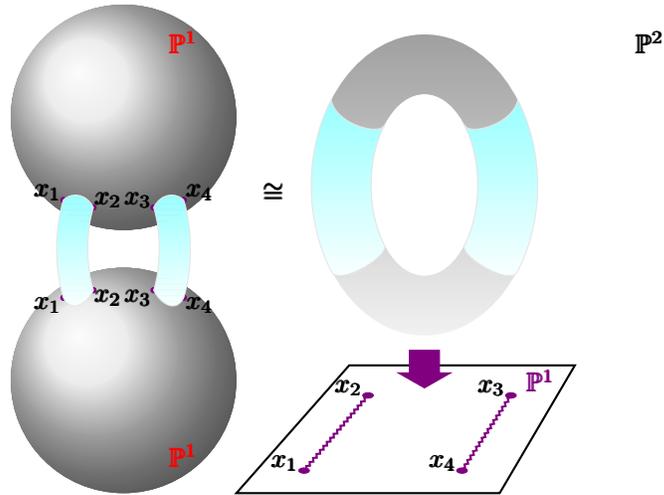
\begin{figure}[!htbp]
\begin{center}
\begin{tikzpicture}[>=latex,decoration={zigzag,amplitude=.5pt,segment length=2pt}]
\shade[ball color=gray!20!white] (-3,2) circle (1.5);
\shade[ball color=gray!20!white] (-3,5.5) circle (1.5);
\filldraw [violet] (-3.8,4.4) circle (1pt) (-3.4,4.3) circle (1pt) (-3.8,3.1) circle (1pt) (-3.4,3.2) circle (1pt);
\filldraw [violet] (-6+3.8,4.4) circle (1pt) (-6+3.4,4.3) circle (1pt) (-6+3.8,3.1) circle (1pt) (-6+3.4,3.2) circle (1pt);
\draw [gray!20!white, top color=cyan!40!white](-6+3.8,4.4) .. controls (-6+3.7,4.5) and (-6+3.5,4.5) .. (-6+3.4,4.3) .. controls (-6+3.5,4) and (-6+3.5,3.5) .. (-6+3.4,3.2) .. controls (-6+3.5,3) and (-6+3.7,2.9) .. (-6+3.8,3.1) .. controls (-6+3.92,3.5) and 
(-6+3.92,4) .. (-6+3.8,4.4); 
\draw [gray!20!white, top color=cyan!40!white](-3.8,4.4) .. controls (-3.7,4.5) and (-3.5,4.5) .. (-3.4,4.3) .. controls (-3.5,4) and (-3.5,3.5) .. (-3.4,3.2) .. controls (-3.5,3) and (-3.7,2.9) .. (-3.8,3.1) .. controls (-3.92,3.5) and (-3.92,4) .. (-3.8,4.4); 
\node at (-4,4.5) {$\pmb {x_1}$};
\node at (-3.2,4.4) {$\pmb {x_2}$};
\node at (-1.98,4.5) {$\pmb {x_4}$};
\node at (-2.8,4.4) {$\pmb {x_3}$};
\node at (-4,3) {$\pmb {x_1}$};
\node at (-3.2,3.1) {$\pmb {x_2}$};
\node at (-1.98,3) {$\pmb {x_4}$};
\node at (-2.8,3.1) {$\pmb {x_3}$};
\draw [white,top color=gray!75!white, bottom color=gray!10!white] (1,4.6) ellipse (1.51 and 2.01);
\filldraw [white] (1,4.6) ellipse (0.7 and 1.2);
\draw [gray!20!white, top color=cyan!40!white](0.47,5.4) arc (138:222.6:0.7 and 1.2) .. controls (0.47,3.6) and (-0.17,3.25) .. (-0.240,3.456) arc (214.6:145.4:1.5 and 2) .. controls (-0.17,5.4) and (0.42,5.25) .. (0.47,5.4) -- cycle;
\draw [gray!20!white, top color=cyan!40!white](2-0.47,5.4) arc (180-138:180-222.6:0.7 and 1.2) .. controls (2-0.47,3.6) and (2+0.17,3.25) .. (2+0.240,3.456) arc (180-214.6:180-145.4:1.5 and 2) .. controls (2+0.17,5.4) and (2-0.42,5.25) .. (2-0.47,5.4) -- cycle;
\draw [thick] (0,2.2) -- (3,2.2) -- (2,0.5) -- (-1.5,0.5) -- cycle;  
\filldraw [violet] (0.25,1.8) ellipse (2pt and 1pt) (2.15,1.8) ellipse (2pt and 1pt) (1.5,0.8) ellipse (2pt and 1pt) (-0.6,0.8) ellipse (2pt and 1pt);
\draw [thick,violet,decorate] (0.25,1.8) -- (-0.6,0.8); \node at (0,1.9) {$\pmb {x_2}$}; \node at (1.9,1.9) {$\pmb {x_3}$}; \node [violet] at (2.55,2) {$\pmb {\mathbb P^1}$};
\draw [thick,violet,decorate] (2.15,1.8) -- (1.5,0.8); \node at (-0.85,0.9) {$\pmb {x_1}$}; \node at (1.25,0.9) {$\pmb {x_4}$};
\filldraw [violet,rotate=-90] (-2.4,1.2) -- (-2.1,1.2) -- (-2.1,1.35) -- (-1.9,1) -- (-2.1,0.65) -- (-2.1, 0.8) -- (-2.4,0.8) -- cycle;
\node [red] at (-2.2,6.5) {$\pmb {\mathbb P^1}$}; \node at (4,6.5) {$\pmb {\mathbb P^2}$};
\node [red] at (-2.2,1) {$\pmb {\mathbb P^1}$};
\node at (-1,4.5) {$\pmb {\cong}$};
\end{tikzpicture}
\caption{Two double cut $\mathbb P^1$ planes glued along the cuts are topologically equivalent to a torus.}
\label{fig:elliptic curve}
\end{center}
\end{figure}
We see that the integrals can then be thought as along the edge of the cut where the root is positive. If we go back integrating along the other edge of the cut, we get the same result. Since the whole cut traveled this way is a closed path $\gamma$ along the 
elliptic curve, we can write 
\begin{align}
K(k)&=\frac 14 \oint_\gamma \frac {dx}{\sqrt{(1-x^2)(1-k^2x^2)}}, \\
 E(k)&=\frac 14 \oint_\gamma \frac {\sqrt{1-k^2 x^2}}{\sqrt{1-x^2}}\ dx. 
\end{align}
Here, $\gamma$ is one of the generators of the first homology group of the elliptic curve. In this form the elliptic integrals looks like a 1-form integrated over a closed cycle.
Since the first homology group of an elliptic curve is two dimensional, there is a second independent curve, which we can represent, for example, as the path going from $z_+$ to $\tilde z_+$ on one of the spheres and coming back from the other sphere. 
This second path,
say $\beta$, intersects $\gamma$ transversally in a point. After having chosen an orientation we can define the intersection product $\beta\cdot \gamma=1$, meaning that the velocities $\dot\beta$ and $\dot \gamma$ at the intersection point form a basis for the tangent space correctly
oriented. Of course this means that $\gamma \cdot \beta=-1$. However, in place of considering the same integral changing $\gamma$ with $\beta$, let us first consider the other canonical integrals $K(k')$ and $E(k')$, with $k'=\sqrt {1-k^2}$. These integrals
are related to the previous ones by the Legendre's quadratic relation
\begin{align}
 K(k)E(k')+E(k)K(k')-K(k)K(k')=\frac \pi2.
\end{align}
To understand the meaning of this quadratic relation it is convenient to pass to a more canonical description of the integrals of first and second kind. Consider an integral of the form
\begin{align}
 I_1=\int_{x_1}^{x_2} \frac {dx}{\sqrt {P(x)}},\label{P1form}
\end{align}
where $P$ is a fourth order polynomial with simple roots among which there are $x_1$ and $x_2$ (not necessarily real). Since the integrand is not single valued, this integral depends on the specification of the path connecting the two roots.
However, it is clear that we can work exactly as before, so that we can write\footnote{$\gamma$ runs along $[-1,1]$ twice, so a further factor $\frac 12$ appears}
\begin{align}
I_1=\frac 12 \int_{\gamma} \frac {dx}{y}, 
\end{align}
where $\gamma$ is a closed path on the genus one complex line defined by the equation 
\begin{align}
 y^2=P(x) 
\end{align}
in $\mathbb P^2$.\footnote{$(x,y)$ are non homogeneous coordinates on $\mathbb P^2$ that can be related to homogeneous ones $(z_0:z_1:z_2)$ by $x=z_1/z_0$, $y=z_2/z_0$ on the patch $z_0\neq 0$.} Notice that in the form (\ref{P1form}) we
are working just on a $\mathbb P^1$ projected component $\mathbb P^2\mapsto \mathbb P^1$ (this is the reason why the branch points appear). On $\mathbb P^1$ there is the action of the M\"obius fractional group $PGL(2,\mathbb C)$, which allows
to move three points in any desired position. It works as follows. Let $z$ be the inhomogeneous coordinate on $\mathbb P^1$, and 
\begin{align}
 A=
\begin{pmatrix}
 a & b\\ c & d
\end{pmatrix}
\in GL(2,\mathbb C).
\end{align}
Its action on $z$ is thus defined by
\begin{align}
 Az=\frac {az+b}{cz+d}.
\end{align}
For $\lambda\in \mathbb C_*$, we se that $\lambda A$ gives the same action as $A$, so it reduces to an action of the projective group $PGL(2,\mathbb C)$.
Moving one of the roots at infinity, this way it is always possible to bring the integral in the form $I_1=\chi J_1$, where 
\begin{align}
J_1= \int_{\gamma} \frac {dx}{y}, 
\end{align}
where now the elliptic curve is defined by the Weierstrass normal  form\footnote{For example, assume $P(x)=a(x-x_1)(x-x_2)(x-x_3)(x-x_4)$, and take the transformation 
\begin{align}
x(z)=\frac {x_1 z-b}{z-d}, \qquad z(x)=\frac {xd-b}{x-x_1}, 
\end{align}
which sends $x_1$ to $\infty$. Then, the integral (\ref{P1form}) takes the form
\begin{align}
 I_1=2\sqrt {\frac {x_1d-b}{a\prod_{j=2}^4 (x_1-x_j)}} \int_{z(x_2)}^\infty \frac {dz}{\sqrt{4(z-z(x_2))(z-z(x_3))(z-z(x_4))}}.
\end{align}
Finally, fixing $b,d$ such that $\sum_{j=2}^4 z(x_j)=0$, we get the Weierstrass normal form.
} \cite{shafar}
\begin{align}
 y^2=4x^3-g_2 x -g_3.
\end{align}
Another choice is to bring it in the Legendre normal form 
\begin{align}
 y^2=(1-x^2)(1-k^2x^2),
\end{align}
so relating it to the complete elliptic integral of the first kind. We are interested in the Weierstrass form now. Associated to it, the elliptic integrals of the second kind can be written in the form 
\begin{align}
J_2= \int_{\gamma} \frac {xdx}{y}. 
\end{align}
Let us now quickly see how $J_1$ and $J_2$ are related to the theory of Weierstrass $\wp$ and $\zeta$ elliptic functions.
For $\omega_1$ and $\omega_2$ two complex numbers such that
\begin{align}
 \tau:=\frac {\omega_2}{\omega_1}
\end{align}
has strictly positive imaginary part, one defines the Weierstrass gamma function $\wp$ as \cite{chandrasekharan1985elliptic}
\begin{align}
 \wp(z)=\frac 1{z^2}+\sum_{(m,n)\in \mathbb Z^2_0} \left( \frac 1{(z-m\omega_1-n\omega_2)^2} -\frac 1{(m\omega_1+n\omega_2)^2} \right),
\end{align}
where $m,n$ are relative integers with $(m,n)\neq (0,0)$. It is biperiodic, with periods $\omega_1$ and $\omega_2$, and satisfies the differential equation
\begin{align}
(\wp')^2=4\wp^3-g_2 \wp -g_3, 
\end{align}
with
\begin{align}
g_2&= \sum_{(m,n)\in \mathbb Z^2_0} \frac {60}{(m\omega_1+n\omega_2)^4},\\
g_3&=\sum_{(m,n)\in \mathbb Z^2_0} \frac {140}{(m\omega_1+n\omega_2)^6}.
\end{align}
The Weierstrass $\zeta$ function is
\begin{align}
 \zeta(z)=\frac 1z+\sum_{(m,n)\in \mathbb Z^2_0} \left( \frac 1{z-m\omega_1-n\omega_2}+\frac 1{m\omega_1+n\omega_2} +\frac z{(m\omega_1+n\omega_2)^2}\right).
\end{align}
It is then clear that $\zeta'(z)=-\wp(z)$. However, $\zeta(z)$ is not biperiodic, but it satisfies the relations
\begin{align}
 \zeta(z+\omega_j)=\zeta(z)+2\eta_j, \quad \eta_j=\zeta(\frac {\omega_j}2), \quad j=1,2,3,
\end{align}
with $\omega_3=\omega_1+\omega_2$. In this way, the Legendre's quadratic relation takes the form
\begin{align}
 \eta_1\omega_2-\omega_1\eta_2=i\pi,
\end{align}
which can be easily proven after integrating $\zeta$ along a fundamental parallelogram centred in $0$, with fundamental periods as edges \cite{chandrasekharan1985elliptic}. This means that the rectangle is $(abcd)$, with
\begin{align}
 (a,b,c,d)=(-\frac {\omega_1+\omega_2}2, \frac {\omega_1-\omega_2}2, \frac {\omega_1+\omega_2}2, \frac {-\omega_1+\omega_2}2), 
\end{align}
see figure \ref{fig:chandrasekaran}.
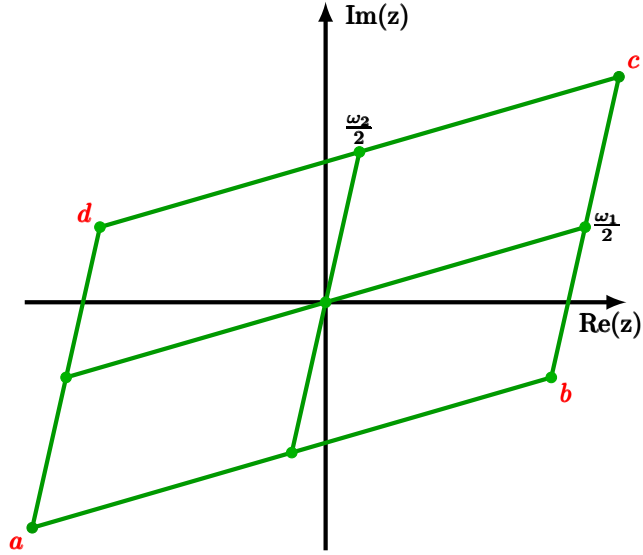
\begin{figure}[!htbp]
\begin{center}
\begin{tikzpicture}[>=latex,decoration={zigzag,amplitude=.5pt,segment length=2pt}]
\draw [ultra thick,->] (-4,0) -- (4,0);
\draw [ultra thick,->] (0,-3.3) -- (0,4);
\draw [black!40!green,ultra thick] (-3.9,-3) -- (3,-1) -- (3.9,3) -- (-3,1) -- cycle;
\draw [black!40!green,ultra thick] (-3.45,-1) -- (3.45,1);
\draw [black!40!green,ultra thick] (-0.45,-2) -- (0.45,2);
\filldraw [black!40!green] (3.9,3) circle (2pt);
\filldraw [black!40!green] (-3.9,-3) circle (2pt); 
\filldraw [black!30!green] (3,-1) circle (2pt);
\filldraw [black!30!green] (-3,1) circle (2pt);
\filldraw [black!30!green] (0,0) circle (2pt);
\filldraw [black!30!green] (3.45,1) circle (2pt);
\filldraw [black!30!green] (-3.45,-1) circle (2pt);
\filldraw [black!30!green] (0.45,2) circle (2pt);
\filldraw [black!30!green] (-0.45,-2) circle (2pt);
\node at (3.8,-0.3) {$\pmb {\rm{Re}(z)}$};
\node at (0.7,3.8) {$\pmb {\rm{Im}(z)}$};
\node at (3.75,1) {$\pmb {\frac {\omega_1}2}$};
\node at (0.45,2.3) {$\pmb {\frac {\omega_2}2}$};
\node [red] at (-4.1,-3.2) {$\pmb {a}$};
\node [red] at (4.1,3.2) {$\pmb {c}$};
\node [red] at (3.2,-1.2) {$\pmb {b}$};
\node [red] at (-3.2,1.2) {$\pmb {d}$};
\end{tikzpicture}
\caption{Fundamental parallelogram.}
\label{fig:chandrasekaran}
\end{center}
\end{figure}
Now, in the fundamental region, the function $z$ has a unique simple pole, sited in $z=0$, with residue 1. Therefore, we can write
\begin{align}
 2\pi i=\int_{(abcd)} \zeta(z) d z=\int_a^b \zeta(z) dz+\int_b^c \zeta(z) dz+\int_c^d \zeta(z) dz+\int_d^a \zeta(z) dz. \label{2pi}
\end{align}
On the other hand, we have that
\begin{align}
 (d,c)=(a+\omega_2, b+\omega_2), \qquad (b,c)=(a+\omega_1, d+\omega_1),
\end{align}
which imply
\begin{align}
 \int_b^c \zeta(z) dz=& \int_{a+\omega_1}^{d+\omega_1} \zeta(z) dz=\int_{a}^{d} \zeta(t+\omega_1) dt=\int_{a}^{d} (\zeta(t)+2\eta_1) dt=\int_{a}^{d} \zeta(t) dt+2\eta_1(d-a)\cr
 =&\int_{a}^{d} \zeta(t) dt+2\eta_1\omega_2,
\end{align}
and
\begin{align}
 \int_c^d \zeta(z) dz=& -\int_d^c \zeta(z) dz=-\int_{a+\omega_2}^{b+\omega_2} \zeta(z) dz=-\int_{a}^{b} \zeta(t+\omega_2) dt=-\int_{a}^{b} (\zeta(t)+2\eta_2) dt\cr
 =&-\int_{a}^{b} \zeta(t) dt-2\eta_2(b-a)=-\int_{a}^{b} \zeta(t) dt-2\eta_2\omega_1.
\end{align}
Inserted in (\ref{2pi}) we get
\begin{align}
 2\pi i=2\eta_1\omega_2-2\eta_2\omega_1,
\end{align}
which is the assert.\\
Finally, it is also possible to prove (but we omit the proof) that 
\begin{align}
 \omega_1&=J_1, \qquad \eta_1=\frac 12 J_2,\\
 \omega_2&=\tilde J_1, \qquad \eta_2=\frac 12 \tilde J_2,
\end{align}
where $\tilde J_k$ are obtained replacing $\gamma$ with $\beta$, so we can further rewrite the Legendre's relation in the form
\begin{align}
 \tilde J_1 J_2-\tilde J_2 J_1=2\pi i. \label{Legendrepz}
\end{align}
It is in this form that the Legendre's relation has its deepest geometrical meaning.
\subsection{Riemann's bilinear relations and intersections}
To interpret the above result we pass to a more general situation and consider an oriented Riemann surface $\Sigma$ of genus $g>0$ (the surface of a donut with $g$ holes or a sphere with $g$ handles). It is well known, \cite{dubrovin1984modern}, Lemma 3.14,
that we can realize it from a $4g$-gone with edges,
in clockwise order, $a_1b_1a_1^{-1} b_1^{-1}\cdots a_gb_ga_g^{-1} b_g^{-1}$, where $a^{-1}$ means the edge is oriented counterclockwise, gluing each edge with its ``inverse'', according to the orientation. Any $a_j$ and any $b_j$ is a closed curve 
and they satisfy the intersection product relations
\begin{align}
a_j\cdot a_k=b_j\cdot b_k=0, \qquad a_j\cdot b_k=-b_k\cdot a_j=\delta_{jk}. 
\end{align}
These curves are homotopically distinct and are a (canonical) basis for the first homology group $H_1(\Sigma,\mathbb Z)\simeq \mathbb Z^{2g}$. Given a closed 1-form $\omega$ on $\Sigma$, one defines its periods relative to the given basis as
\begin{align}
 \pi^a_j (\omega) := \oint_{a_j} \omega, \qquad \pi^b_j (\omega) := \oint_{b_j} \omega.
\end{align}
Given two closed forms $\omega_1$ and $\omega_2$, using Stokes's theorem and closures, it is not difficult to prove that the following identity is always true 
\begin{align}
 \int_\Sigma \omega_1\wedge \omega_2= \sum_{j=1}^g \left[\pi^a_j(\omega_1) \pi^b_j(\omega_2)-\pi^b_j(\omega_1) \pi^a_j(\omega_2)\right].\label{RBR}
\end{align}
\begin{proof}
\noindent Assume that $\Sigma$ is a closed oriented Riemann surface of genus $g$. 
Now, let us notice that if we cut $\Sigma$ along the representatives $a_j$ and $b_j$ we get the $4g$-gon $\tilde \Sigma$ associated to the word $a_1 b_1 a_1^{-1} b_1^{-1}\cdots a_j b_j a_j^{-1} b_j^{-1}\cdots a_g b_g a_g^{-1} b_g^{-1}$. 
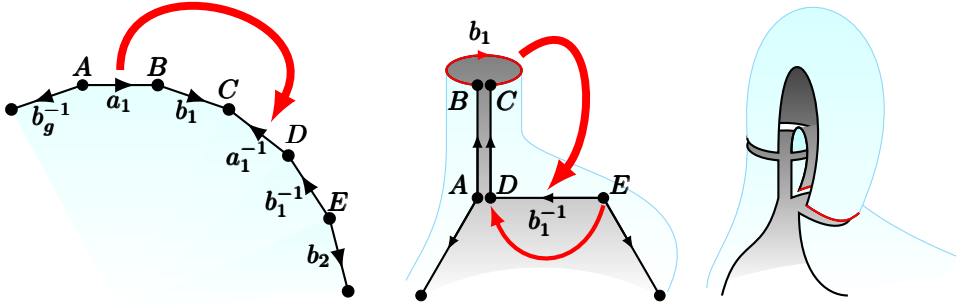
\begin{figure}[!htbp]
\begin{center}
\begin{tikzpicture}[>=latex,decoration={zigzag,amplitude=.5pt,segment length=2pt}]
\shade [top color=cyan!10!white] (0,0) -- (-0.5,3) -- (0.5,3) -- cycle;
\shade [top color=cyan!10!white,rotate=2*9.46232] (0,0) -- (-0.5,3) -- (0.5,3) -- cycle;
\shade [top color=cyan!10!white,rotate=-2*9.46232] (0,0) -- (-0.5,3) -- (0.5,3) -- cycle;
\shade [top color=cyan!10!white,rotate=-4*9.46232] (0,0) -- (-0.5,3) -- (0.5,3) -- cycle;
\shade [top color=cyan!10!white,rotate=-6*9.46232] (0,0) -- (-0.5,3) -- (0.5,3) -- cycle;
\shade [top color=cyan!10!white,rotate=-8*9.46232] (0,0) -- (-0.5,3) -- (0.5,3) -- cycle;
\draw [thick] (-0.5,3) -- (0.5,3);
\draw [thick,rotate=2*9.46232] (-0.5,3) -- (0.5,3);
\draw [thick,rotate=-2*9.46232] (-0.5,3) -- (0.5,3);
\draw [thick,rotate=-4*9.46232] (-0.5,3) -- (0.5,3);
\draw [thick,rotate=-6*9.46232] (-0.5,3) -- (0.5,3);
\draw [thick,rotate=-8*9.46232] (-0.5,3) -- (0.5,3);
\filldraw (-0.5,3) circle (2pt);
\filldraw [thick,rotate=2*9.46232] (-0.5,3) circle (2pt);
\filldraw [thick,rotate=-2*9.46232] (-0.5,3) circle (2pt);
\filldraw [thick,rotate=-4*9.46232] (-0.5,3) circle (2pt);
\filldraw [thick,rotate=-6*9.46232] (-0.5,3) circle (2pt);
\filldraw [thick,rotate=-8*9.46232] (-0.5,3) circle (2pt);
\filldraw [thick,rotate=-10*9.46232] (-0.5,3) circle (2pt);
\draw [ultra thick,->] (0,3) -- (0.2,3);
\draw [ultra thick,->,rotate=-2*9.46232] (0,3) -- (0.2,3);
\draw [ultra thick,->,rotate=-8*9.46232] (0,3) -- (0.2,3);
\draw [ultra thick,->,rotate=2*9.46232] (0,3) -- (-0.2,3);
\draw [ultra thick,->,rotate=-4*9.46232] (0,3) -- (-0.2,3);
\draw [ultra thick,->,rotate=-6*9.46232] (0,3) -- (-0.2,3);
\draw [line width=3pt,->,red] (0,3.2) .. controls (0,4.7) and (2.8,4) .. (2,2.5);
\shade [top color=cyan!10!white] (7+2*0.0868241,0.2) -- (6.25+2*0.0868241,0.2+1.5*0.866025) -- (4.75+2*0.0868241,0.2+1.5*0.866025) -- (4.75+2*0.0868241,0.2+1.5*0.866025+1.5) arc (-80:0:0.5 and 0.2) -- (4.75+0.5-0.0868241+2*0.0868241,0.2+2.2) 
.. controls (4.75+0.5-0.0868241+2*0.0868241,0.2+1.7) and (7+2*0.1368241+0.75,1.75) .. (7+2*0.1368241,0.25) -- cycle;
\draw [cyan!40!white] (7+2*0.0868241,0.2) -- (6.25+2*0.0868241,0.2+1.5*0.866025) -- (4.75+2*0.0868241,0.2+1.5*0.866025) -- (4.75+2*0.0868241,0.2+1.5*0.866025+1.5) arc (-80:0:0.5 and 0.2) -- (4.75+0.5-0.0868241+2*0.0868241,0.2+2.2) 
.. controls (4.75+0.5-0.0868241+2*0.0868241,0.2+1.7) and (7+2*0.1368241+0.75,1.75) .. (7+2*0.1368241,0.25);
\shade [top color=cyan!10!white] (4,0.2) -- (4.75,0.2+1.5*0.866025) -- (4.75,0.2+1.5*0.866025+1.5) arc (260:180:0.5 and 0.2) -- (4.75-0.5+0.0868241,2.2) .. controls (4.75-0.5+0.0868241,1.2) and (4,0.8) .. (3.8,0.4) -- cycle;
\draw [cyan!40!white] (4,0.2) -- (4.75,0.2+1.5*0.866025) -- (4.75,0.2+1.5*0.866025+1.5) arc (260:180:0.5 and 0.2) -- (4.75-0.5+0.0868241,2.2) .. controls (4.75-0.5+0.0868241,1.2) and (4,0.8) .. (3.8,0.4);
\shade [cyan!60!black] (4,0.2) -- (4.75,0.2+1.5*0.866025) -- (4.75,0.2+1.5*0.866025+1.5) arc (260:-80:0.5 and 0.2) -- (4.75+2*0.0868241,0.2+1.5*0.866025) -- (6.25+2*0.0868241,0.2+1.5*0.866025) -- (7+2*0.0868241,0.2) .. controls (7+2*0.0868241-1,0.2+0.5) 
and (4+1,0.2+0.5) .. (4,0.2);
\draw [thick] (4,0.2) -- (4.75,0.2+1.5*0.866025) -- (4.75,0.2+1.5*0.866025+1.5) arc (260:-80:0.5 and 0.2) -- (4.75+2*0.0868241,0.2+1.5*0.866025) -- (6.25+2*0.0868241,0.2+1.5*0.866025) -- (7+2*0.0868241,0.2);
\draw [red,thick] (4.75,0.2+1.5*0.866025+1.5) arc (260:-80:0.5 and 0.2);
\filldraw (4,0.2) circle (2pt) (4.75,0.2+1.5*0.866025) circle (2pt) (4.75,0.2+1.5*0.866025+1.5) circle (2pt) (4.75+2*0.0868241,1.7+1.5*0.866025) circle (2pt) (4.75+2*0.0868241,0.2+1.5*0.866025) circle (2pt) (6.25+2*0.0868241,0.2+1.5*0.866025) circle (2pt)
(7+2*0.0868241,0.2) circle (2pt);
\draw [ultra thick, red, ->] (6.25+2*0.0868241,0.1+1.5*0.866025) .. controls (-0.25+6.25+2*0.0868241,1.5*0.866025-0.8) and (4.85+2*0.0868241+0.25,1.5*0.866025-0.8) .. (4.75+2*0.0868241,0.1+1.5*0.866025);
\draw [line width=3pt,->,red] (4.75+2*0.0868241+0.4,1.7+1.5*0.866025+0.4) .. controls (4.75+2*0.0868241+1+0.4,1.7+1.5*0.866025+0.8+0.4) and (5.5+2*0.0868241+0.7,0.2+1.5*0.866025+0.2+0.7) .. (5.5+2*0.0868241,0.3+1.5*0.866025);
\draw [thick,->] (4.75,0.2+1.5*0.866025) -- (4.75,1.05+1.5*0.866025);
\draw [thick,->] (4.75,0.2+1.5*0.866025) -- (4.75-0.4,0.2+0.7*0.866025);
\draw [thick,->] (4.75+2*0.0868241,0.2+1.5*0.866025) -- (4.75+2*0.0868241,1.05+1.5*0.866025);
\draw [thick,->] (6.25+2*0.0868241,0.2+1.5*0.866025) -- (6.25+2*0.0868241-0.85,0.2+1.5*0.866025);
\draw [thick,->] (6.25+2*0.0868241,0.2+1.5*0.866025) -- (6.25+2*0.0868241+0.4,0.2+0.7*0.866025);
\draw [thick,->,red] (4.75+0.0868241,1.7+1.5*0.866025+0.4) -- (4.75+0.0868241+0.1,1.7+1.5*0.866025+0.4);
\shade [cyan!60!black] (7.8,0.3) .. controls (7.8,0.7) and (-0.5+0.0868241+8.75,1.1) .. (8.75+0.0868241-0.5,0.2+1.5*0.866025+0.5+0.196962) arc (180:0:0.5 and 0.2) .. controls (8.75+0.0868241+0.5,0.2+1.5*0.866025) and (10.5+2*0.0868241,0.9) .. 
(11+2*0.0868241,0.7) .. controls (9.5+2*0.0868241,0.7) and (8.8,0.7) .. (7.8,0.3);
\shade [top color=cyan!10!white] (8,0.2) .. controls (8+0.2,0.2+0.8) and (8.75-0.2,0.2+1.5*0.866025-0.8) .. (8.75,0.2+1.5*0.866025) -- (8.75,0.2+1.5*0.866025+0.5) arc (260:180:0.5 and 0.2) .. controls (-0.5+0.0868241+8.75,1.1) and (7.8,0.7) .. (7.8,0.3) -- cycle;
\draw [cyan!40!white] (8,0.2) .. controls (8+0.2,0.2+0.8) and (8.75-0.2,0.2+1.5*0.866025-0.8) .. (8.75,0.2+1.5*0.866025) -- (8.75,0.2+1.5*0.866025+0.5) arc (260:180:0.5 and 0.2) .. controls (-0.5+0.0868241+8.75,1.1) and (7.8,0.7) .. (7.8,0.3);
\shade [top color=cyan!10!white] (9.5+2*0.0868241,0.2) .. controls (9.2+2*0.0868241,0.2) and (8.85+2*0.0868241,1.5*0.866025-0.8) .. (8.85+2*0.0868241,1.5*0.866025) .. controls (9.85+2*0.0868241,1.5*0.866025-0.6) and (9.75+2*0.0868241,0.6+1.5*0.866025)
.. (8.75+2*0.0868241,0.2+1.5*0.866025) -- (8.75+2*0.0868241,0.2+1.5*0.866025+0.5) arc (-80:0:0.5 and 0.2) .. controls (8.75+0.0868241+0.5,0.2+1.5*0.866025) and (10.5+2*0.0868241,0.9) .. (11+2*0.0868241,0.7) .. controls (11+2*0.0868241,0.3)
and (9.5+2*0.0868241,0.4) .. (9.5+2*0.0868241,0.2);
\draw [cyan!40!white] (9.5+2*0.0868241,0.2) .. controls (9.2+2*0.0868241,0.2) and (8.85+2*0.0868241,1.5*0.866025-0.8) .. (8.85+2*0.0868241,1.5*0.866025) .. controls (9.85+2*0.0868241,1.5*0.866025-0.6) and (9.75+2*0.0868241,0.6+1.5*0.866025)
.. (8.75+2*0.0868241,0.2+1.5*0.866025) -- (8.75+2*0.0868241,0.2+1.5*0.866025+0.5) arc (-80:0:0.5 and 0.2) .. controls (8.75+0.0868241+0.5,0.2+1.5*0.866025) and (10.5+2*0.0868241,0.9) .. (11+2*0.0868241,0.7);
\draw [thick] (8,0.2) .. controls (8+0.2,0.2+0.8) and (8.75-0.2,0.2+1.5*0.866025-0.8) .. (8.75,0.2+1.5*0.866025) -- (8.75,0.2+1.5*0.866025+0.5) arc (260:-80:0.5 and 0.2) -- (8.75+2*0.0868241,0.2+1.5*0.866025) .. controls (9.75+2*0.0868241,0.6+1.5*0.866025) and 
(9.85+2*0.0868241,1.5*0.866025-0.6) .. (8.85+2*0.0868241,1.5*0.866025) .. controls (8.85+2*0.0868241,1.5*0.866025-0.8) and (9.2+2*0.0868241,0.2) .. (9.5+2*0.0868241,0.2);
\shade [top color=black, bottom color=black!50!white] (8.75+0.0868241-0.5,0.2+1.5*0.866025+0.6+0.196962) arc (180:0:0.5 and 0.2) .. controls (8.75+0.0868241+0.5,0.2+1.5*0.866025+0.6+0.196962+2) and 
(8.75+0.0868241-0.5,0.2+1.5*0.866025+0.6+0.196962+2) .. (8.75+0.0868241-0.5,0.2+1.5*0.866025+0.6+0.196962) -- cycle;
\draw [thick] (8.75+0.0868241+0.5,0.2+1.5*0.866025+0.6+0.196962) arc (0:180:0.5 and 0.2); 
\shade [top color=cyan!30!white, bottom color=cyan!10!white] (8.75+0.0868241+0.5,0.2+1.5*0.866025+0.6+0.196962) arc (0:-80:0.5 and 0.2) .. controls (8.75+2*0.0868241,1.5*0.866025+1.5) and (8.75+2*0.0868241+0.4,0.2+1.5*0.866025+0.8) .. 
(8.75+2*0.0868241+0.1,0.3+1.5*0.866025);
\draw [thick] (8.75+2*0.0868241+0.1,0.3+1.5*0.866025) .. controls (8.75+2*0.0868241+0.4,0.2+1.5*0.866025+0.8) and (8.75+2*0.0868241,1.5*0.866025+1.5) .. (8.75+2*0.0868241,1.5*0.866025+0.8) arc (-80:0:0.5 and 0.2);  
\shade [top color=black!50!white, bottom color=black!10!white]  (8.85+2*0.0868241,1.5*0.866025+1.138) .. controls (8.84+2*0.0868241+0.09,1.5*0.866025+1.14) and (8.75+2*0.0868241+0.1+0.3/1.45,0.3+1.5*0.866025+0.7/1.45) .. 
(8.75+2*0.0868241+0.1,0.3+1.5*0.866025) .. controls (9.85+2*0.0868241,0.7+1.5*0.866025) and (9.95+2*0.0868241,1.5*0.866025-0.5) .. (8.95+2*0.0868241,1.5*0.866025+0.1) -- (8.85+2*0.0868241+0.7,1.5*0.866025+1.338) .. controls 
(8.85+2*0.0868241+0.4,1.5*0.866025+1.238) and (8.85+0.2+2*0.0868241,1.5*0.866025+1.138) .. (8.85+2*0.0868241,1.5*0.866025+1.138) -- cycle;
\draw [thick] (8.75+2*0.0868241,1.5*0.866025+0.8) .. controls (8.75+2*0.0868241,1.5*0.866025+1.5) and (8.75+2*0.0868241+0.4,0.2+1.5*0.866025+0.8) .. (8.75+2*0.0868241+0.1,0.3+1.5*0.866025) .. controls (9.85+2*0.0868241,0.7+1.5*0.866025) and
(9.95+2*0.0868241,1.5*0.866025-0.5) .. (8.95+2*0.0868241,1.5*0.866025+0.1); 
\draw [red,thick] (8.75+2*0.0868241+0.1,0.3+1.5*0.866025) .. controls (9.85+2*0.0868241,0.7+1.5*0.866025) and (9.95+2*0.0868241,1.5*0.866025-0.5) .. (8.95+2*0.0868241,1.5*0.866025+0.1); 
\shade [top color=cyan!20!white, bottom color=cyan!5!white] (8.95+2*0.0868241,1.5*0.866025+0.1) .. controls (8.95+2*0.0868241+0.5,1.5*0.866025+1.6) and (8.75,1.5*0.866025+3) .. (8.75,1.5*0.866025+0.8) arc (260:180:0.5 and 0.2) .. controls 
(8.75+0.0868241-0.7,0.2+1.5*0.866025+0.5+0.196962+3) and
(8.75+0.0868241-0.7+3,0.2+1.5*0.866025+0.5+0.196962+2) .. (8.85+2*0.0868241+0.8,1.5*0.866025) .. controls (8.85+2*0.0868241+0.7,1.5*0.866025-0.08) and (9.41+2*0.0868241,1.5*0.866025-0.23) .. (8.95+2*0.0868241,1.5*0.866025+0.1);
\draw [cyan!40!white] (8.95+2*0.0868241,1.5*0.866025+0.1) .. controls (8.95+2*0.0868241+0.5,1.5*0.866025+1.6) and (8.75,1.5*0.866025+3) .. (8.75,1.5*0.866025+0.8) arc (260:180:0.5 and 0.2) .. controls 
(8.75+0.0868241-0.7,0.2+1.5*0.866025+0.5+0.196962+3) and
(8.75+0.0868241-0.7+3,0.2+1.5*0.866025+0.5+0.196962+2) .. (8.85+2*0.0868241+0.8,1.5*0.866025) .. controls (8.85+2*0.0868241+0.7,1.5*0.866025-0.08) and (9.41+2*0.0868241,1.5*0.866025-0.23) .. (8.95+2*0.0868241,1.5*0.866025+0.1) -- cycle;
\draw [thick] (8.85+2*0.0868241+0.8,1.5*0.866025) .. controls (8.85+2*0.0868241+0.7,1.5*0.866025-0.08) and (9.41+2*0.0868241,1.5*0.866025-0.23) .. (8.95+2*0.0868241,1.5*0.866025+0.1) .. controls (8.95+2*0.0868241+0.5,1.5*0.866025+1.6) and 
(8.75,1.5*0.866025+3) .. (8.75,1.5*0.866025+0.8) arc (260:180:0.5 and 0.2); 
\draw [red,thick] (8.85+2*0.0868241+0.8,1.5*0.866025) .. controls (8.85+2*0.0868241+0.7,1.5*0.866025-0.08) and (9.41+2*0.0868241,1.5*0.866025-0.23) .. (8.95+2*0.0868241,1.5*0.866025+0.1);
\node at (-0.5,3.25) {$\pmb A$};
\node at (0.5,3.25) {$\pmb B$};
\node at (1.45,2.95) {$\pmb C$};
\node at (2.35,2.35) {$\pmb D$};
\node at (2.9,1.45) {$\pmb E$};
\node at (0,2.75) {$\pmb {a_1}$};
\node at (0.9,2.65) {$\pmb {b_1}$};
\node at (-0.9,2.55) {$\pmb {b^{-1}_g}$};
\node at (1.7,2.05) {$\pmb {a^{-1}_1}$};
\node at (2.2,1.45) {$\pmb {b^{-1}_1}$};
\node at (2.65,0.73) {$\pmb {b_2}$};
\node at (4.5,0.2+1.5*0.866025+0.18) {$\pmb A$};
\node at (4.5,0.2+1.5*0.866025+1.3) {$\pmb B$};
\node at (5.15,0.2+1.5*0.866025+0.18) {$\pmb D$};
\node at (5.15,0.2+1.5*0.866025+1.3) {$\pmb C$};
\node at (6.65,0.2+1.5*0.866025+0.18) {$\pmb E$};
\node at (4.8,3.7) {$\pmb {b_1}$};
\node at (5.7,1.2) {$\pmb {b^{-1}_1}$};
\end{tikzpicture}
\caption{Riemann surfaces from the fundamental $4g$-gon.}
\label{Fig2}
\end{center}
\end{figure}
This polygon is simply connected so any closed form restricted to it is exact. In particular, this is true for $\omega_2$, which has the form $\omega_2=d\phi_2$, for $\phi_2$ a well defined function on $\tilde \Sigma$ (but not on $\Sigma$). More precisely
we can fix a point $o$ in $\tilde \Sigma$ (for example the center). Since the polygon is convex, any $x\in \tilde \Sigma$ is connected to $o$ by a segment and we can write
\begin{align}
\phi_2(x)\equiv \int_o^x \omega_2, 
\end{align}
where the integration is intended along the segment. Now, if $\omega_2$ is exact and $\omega_1$ is closed (indeed exact too on the simply connected domain), then also $\omega_1\wedge \omega_2$ is exact.
Indeed, the 1-form $\Omega=-\phi_2 \omega_1$ has the property that $d\Omega=\omega_1 \wedge \omega_2$, and is well defined on $\tilde \Sigma$. Finally, since as sets
\begin{align}
 \Sigma -\bigcup_{j=1}^g (a_j \bigcup b_j)=\tilde \Sigma -\partial \tilde \Sigma,
\end{align}
together with $\bigcup_{j=1}^g (a_j \bigcup b_j)$ and $\partial \tilde \Sigma$ are subsets of vanishing measure of $\Sigma$ and $\tilde \Sigma$ respectively, we have
\begin{align}
 \int_\Sigma \omega_1\wedge \omega_2&=\int_{\Sigma -\bigcup_{j=1}^g (a_j \bigcup b_j)} \omega_1\wedge \omega_2=\int_{\tilde \Sigma-\partial \tilde \Sigma} \omega_1\wedge \omega_2=\int_{\tilde \Sigma} \omega_1\wedge \omega_2
 =\int_{\tilde \Sigma}d\Omega =\int_{\partial\tilde \Sigma}\Omega\cr &=-\int_{\partial\tilde \Sigma} \phi_2 \omega_1 =-\int_{\partial\tilde \Sigma} \omega_1(x) \int_o^x \omega_2,
\end{align}
where the main trick has been to rewrite the integral over a region where the integrand is exact, in order to be able to use the Stokes theorem. Since the boundary $\partial \tilde \Sigma$ is the (oriented) union of the edges of the polygon, we then have
\begin{align}
 \int_\Sigma \omega_1\wedge \omega_2&=-\sum_{j=1}^g \left( \int_{a_j} \omega_1(x) \int_o^x \omega_2+\int_{b_j} \omega_1(x) \int_o^x \omega_2+\int_{a_j^{-1}} \omega_1(x) \int_o^x \omega_2+\int_{b_j^{-1}} \omega_1(x) \int_o^x \omega_2 \right).
\end{align}
Now, consider the integrals
\begin{align}
I^a_j\equiv \int_{a_j} \omega_1(x) \int_o^x \omega_2+\int_{a_j^{-1}} \omega_1(x) \int_o^x \omega_2.
\end{align}
The forms $\omega_1$ and $\omega_2$ are well defined on $\Sigma$, so they take the same values on the identified points of $a_j$ and $a_j^{-1}$. If $y(x)$ is the point on $a_j^{-1}$ which identifies with $x\in a_j$ in gluing $\tilde \Sigma$ to form $\Sigma$,
then, noticing that $a_j^{-1}$ is oriented in the opposite sense of $a_j$, we can write
\begin{align}
I^a_j= \int_{a_j} \omega_1(x) \left(\int_o^x \omega_2-\int_o^{y(x)} \omega_2 \right)=\int_{a_j} \omega_1(x) \left(\int_o^x \omega_2+\int_{y(x)}^o \omega_2 \right).
\end{align}
The sum of the two integrals of $\omega_2$ is just equivalent to the integral along a curve from $y(x)$ to $x$, which is homotopically equivalent to $b_j$ traveled in the opposite orientation (from $a_j^{-1}$ to $a_j$), see Fig. \ref{fig:riemann-bilinear}. 
\begin{figure}[!htbp]
\begin{center}
\begin{tikzpicture}[>=latex,decoration={zigzag,amplitude=.5pt,segment length=2pt}]
\draw [ultra thick,fill=white!70!gray] (-3,0) -- (3,0) -- (2.5,-4) -- (-3.5,-4) -- cycle;
\draw [red, ultra thick] (1.5,-4) -- (-0.25,-2) -- (2,0);
\draw [red,->,ultra thick] (1.5,-4) -- (0.625,-3);
\draw [red,->,ultra thick] (-0.25,-2) -- (0.875,-1);
\draw [red,dashed, ultra thick] (1.5,-4) -- (2,0);
\draw [red,dashed,->, ultra thick] (1.5,-4) -- (1.75,-2);
\draw [violet, <->,ultra thick] (2.2,0.1) .. controls (5,0.8) and (4.5,-4.8) .. (1.7,-4.1);
\draw [->,ultra thick] (-3,0) -- (0.1,0);
\draw [->,ultra thick] (3,0) -- (2.75,-2);
\draw [->,ultra thick] (-3.5,-4) -- (-0.6,-4);
\draw [->,ultra thick] (-3,0) -- (-3.25,-2);
\filldraw [red] (1.5,-4) circle (2pt) (2,0) circle (2pt);
\filldraw (-0.25,-2) circle (2pt);
\node at (0,0.25) {$\pmb {a}$};
\node at (3,-2) {$\pmb {b}$};
\node at (-0.6,-4.3) {$\pmb {a^{-1}}$};
\node at (-3.6,-2) {$\pmb {b^{-1}}$};
\node at (-0.5,-2) {$\pmb {o}$};
\node [red] at (2,0.3) {$\pmb {x}$};
\node [red] at (1.5,-4.3) {$\pmb {y(x)}$};
\node [violet] at (5.5,-2) {identified in $\Sigma$};
\node [violet] at (-2,-2) {\pmb {$\tilde \Sigma$}};
\node [violet] at (0,-1) {\pmb{$\int_o^x \omega_2$}};
\node [violet] at (-0.3,-3) {\pmb{$\int_{y(x)}^o \omega_2$}};
\end{tikzpicture}
\caption{Construction of the integral $I_b$ in case $g=1$.}
\label{fig:riemann-bilinear}
\end{center}
\end{figure}
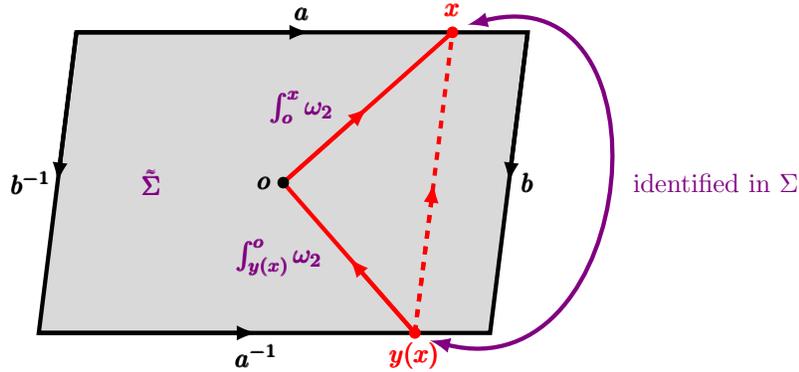
Therefore,
\begin{align}
 \int_o^x \omega_2+\int_{y(x)}^0 \omega_2=-\int_{b_j} \omega_2=-\pi^b_{j}(\omega_2)
\end{align}
is independent from $x$. So, we can write
\begin{align}
 I^a_j=-\pi^b_{j}(\omega_2)\int_{a_j} \omega_1(x)=-\pi^b_{j}(\omega_2) \pi^a_j(\omega_1).
\end{align}
Let us now see what changes for the remaining integrals
\begin{align}
I^b_j\equiv \int_{b_j} \omega_1(x) \int_o^x \omega_2+\int_{b_j^{-1}} \omega_1(x) \int_o^x \omega_2.
\end{align}
Reasoning as above, if $y(x)$ is the point on $b_j^{-1}$ identified with $x$ on $b_j$ when constructing $\Sigma$ from $\tilde \Sigma$, we can write
\begin{align}
I^b_j= \int_{b_j} \omega_1(x) \left(\int_o^x \omega_2+ \int_{y(x)}^o \omega_2\right).
\end{align}
Now, the path $y(x)\to o\to x$ on $\Sigma$ becomes a path homotopic to $a_j^{-1}$ traveled in the opposite direction, so it is homotopic to $a_j$ and, therefore
\begin{align}
 \int_o^x \omega_2+ \int_{y(x)}^o \omega_2=\pi^a_j(\omega_2).
\end{align}
Thus,
\begin{align}
I^b_j= \int_{b_j} \omega_1(x) \pi^a_j(\omega_2)=\pi^a_j(\omega_2)\pi^b_j(\omega_1).
\end{align}
Since
\begin{align}
  \int_\Sigma \omega_1\wedge \omega_2=-\sum_{j=1}^g (I^a_j+I^b_j), 
\end{align}
we get the assertion. 
\end{proof}

This wonderful formula, known as Riemann's bilinear relations, gives (\ref{Legendrepz}) as a very particular case and provides the geometric interpretation we were looking for. 
Indeed it leads to the concept of intersection number between closed forms, or, as we will see in a moment, between cohomology cocycles. 
The transversality of the homology curves is now replaced by the fact that in any given point $\omega_1\wedge \omega_2$ is non degenerate as a bilinear (or multilinear) map on the tangent space
to $\Sigma$. The directions where $\omega_1$ vanishes define locally a curve. Similarly, $\omega_2$ defines locally a second curve. These two curves meet transversally in the given point. This curves are essentially the Poincar\'e duals of the forms and this 
naive description shows its role in providing a geometrical interpretation of the intersection number between forms. When $\Sigma$ is provided by a Riemannian metric $\pmb h$, then, in local coordinates, we can write
\begin{align}
 \omega_1\wedge \omega_2=\omega_{1j} \omega_{2k} dx^j \wedge dx^k=\omega_{1j} \omega_{2k} \frac 1{\sqrt h} \varepsilon^{jk} \sqrt{h} dx^1\wedge dx^2,
\end{align}
where $\varepsilon^{jk}$ is the usual Levi-Civita tensor density, and Einstein's summation convention is adopted. We set 
\begin{align}
 w_2^j\equiv \omega_{2k} \frac 1{\sqrt h} \varepsilon^{jk},
\end{align}
the Hodge dual of $\omega_2$, and identify $\sqrt{h} dx^1\wedge dx^2=dV$ as the volume element on $\Sigma$. Finally, since closed forms are locally exact, in a small enough neighbourhood we can write $\omega_1=d\phi_1$ for a potential $\phi_1$. 
Notice that the Poincar\'e dual of $\omega_1$ is thus locally represented by the level curves $\phi_1=const$. In the given neighbourhood we get
\begin{align}
 \omega_1\wedge \omega_2 =\partial_j \phi_1 w_2^j\ dV= \partial_j (\phi_1 w_2^j)\ dV,
\end{align}
where we used that $\omega_2$ closed implies $\partial_j w_2^j=0$. If as a small neighbourhood we take a small rectangle $Q$ having edges $a,b,c,d$ with $a$ and $c$ along level curves for $\phi_1$ (with values $\kappa_a$ and $\kappa_c$), while 
$b$ and $d$ are orthogonal, then 
\begin{align}
 \int_Q \omega_1\wedge \omega_2= \int_Q \partial_j (\phi_1 w_2^j)\ dV=\kappa_c \int_c n^c_j w_2^j ds +\kappa_a \int_a n^a_j w_2^j ds,
\end{align}
$n^c_j$ and $n^a_j$ being the external normals to $c$ and $a$ respectively. This integral is thus determined by the flux of the vector field $w_2^j$, Hodge dual to $\omega_2$, through the equipotential surfaces (lines) of $\omega_1$. Notice that even in
an extremal case when the rectangle can be extended until wrapping the surface so that\footnote{notice that this cannot be the general case, or at least not assuming $b$ and $d$ remain orthogonal to the equipotential lines} $a=c$, then we would get 
\begin{align}
 \int_Q \omega_1\wedge \omega_2=(\kappa_c-\kappa_a) \int_c n^c_j w_2^j ds,
\end{align}
in general with $\kappa_c-\kappa_a\neq 0$ despite $c$ and $a$ coincide, since $\phi_1$ is not monodromic and takes a different value after wrapping $\Sigma$ along a closed loop. Indeed, in this case $(\kappa_c-\kappa_a)$ is nothing but the period
of $\omega_1$ along the closed path. In general, it is not possible to extend even $c$ to a closed path, and our naive analysis then stops here.  However, it is clear that its correct formulation is given exactly by the Riemann's bilinear formula, which 
generalizes this vision in terms of fluxes (of the Hodge dual of $\omega_2$) through the Poincar\'e dual of $\omega_1$.\\
Now, it is time to be a little bit more precise. If we modify 
$\omega_1$ by an exact form $df$, its periods don't change\footnote{as well as they don't change deforming the path $a_j$ and $b_j$ continuously}. On the other hand, since $\omega_2$ is closed too,
\begin{align}
 (\omega_1+df)\wedge \omega_2=\omega_1\wedge \omega_2 +d(f \omega_2), \label{exact}
\end{align}
and by Stokes's theorem and the fact that $\Sigma$ has no boundary, we see that the exact term does not contribute to the l.h.s. integral. The same holds if we modify $\omega_2$ by an exact form. This means that this formula depends only on the
cohomology classes of $\omega_1$ and $\omega_2$. From (\ref{exact}), we also see that if we consider the cohomology class of $\omega_1\wedge \omega_2$, then, it depends only on the classes $[\omega_1]$ and $[\omega_2]$ and not specifically
on the forms. So, this gives rise to a well defined bilinear antisymmetric map 
\begin{align}
 \cup : H^1(\Sigma)\times H^1(\Sigma) \longrightarrow H^2(\Sigma),\ ([\omega_1],[\omega_2])\mapsto  [\omega_1]\cup[\omega_2]:=[\omega_1\wedge\omega_2],
\end{align}
called the {\it cup product}. The above formula can then be written as
\begin{align}
[\omega_1]\cdot [\omega_2]\equiv \int_\Sigma [\omega_1]\cup [\omega_2]= \sum_{j=1}^g \left[\pi^a_j([\omega_1]) \pi^b_j([\omega_2])-\pi^b_j([\omega_1]) \pi^a_j([\omega_2])\right].
\end{align}
The l.h.s. formula is called the {\it intersection product} of the two classes (or forms). This is because it is in a sense a dual of the intersection formula in homology. Indeed, let us consider the dual basis to the canonical one, that is a set of closed one forms
$\omega^a_j$, $\omega^b_j$, $j=1,\ldots,g$ such that
\begin{align}
 \pi^a_j(\omega^a_k)&=\delta_{jk}, \qquad \pi^a_j(\omega^b_k)=0, \\
 \pi^b_j(\omega^a_k)&=0, \qquad \pi^b_j(\omega^b_k)=\delta_{jk}.
\end{align}
Thus,
\begin{align}
 [\omega^a_j]\cdot [\omega^a_k]&=[\omega^b_j]\cdot [\omega^b_k]=0, \\
 [\omega^a_j]\cdot [\omega^b_k]&=-[\omega^b_k]\cdot [\omega^a_j]=\delta_{jk}, 
\end{align}
which has exactly the same form of the intersection products among the elements of the homology basis. This is our final interpretation of the Legendre's bilinear formula (\ref{Legendrepz}): it gives the intersection product between the first kind and the second kind
differentials. In general, after selecting a basis $b_j$, $j=1,\ldots, 2g$ of $H^1(\Sigma)$, the matrix $\pmb c$ with elements
\begin{align}
c_{jk}= b_j\cdot b_k
\end{align}
is invertible. Any other element $[\omega]$ in $H^1(\Sigma)$ can be written as 
\begin{align}
 [\omega]=\sum_{j=1}^{2g} k^j b_j,
\end{align}
with $k_J\in \mathbb C$.
Therefore,
\begin{align}
 [\omega]\cdot b_k=\sum_{j=1}^{2g} k^j c_{jk}
\end{align}
so that
\begin{align}
k^j =\sum_{k=1}^{2g}   [\omega]\cdot b_k c^{kj},
\end{align}
where $c^{kj}$ are the matrix elements of the inverse matrix $\pmb c^{-1}$. On the other hand, for $u=a,b$, the periods of $[\omega]$ are
\begin{align}
 \pi^u_j([\omega])=\sum_{k=1}^{2g} k^k \pi^u_j(b_k)= \sum_{l=1}^{2g}\sum_{k=1}^{2g}([\omega] \cdot b_l) c^{lk} \pi^u_j(b_k).
\end{align}
Thus, we see that after the intersection product is given, for a path $\gamma$ homologically equal to $\gamma=\sum_{j=1}^g (m_a^j a_j +m_b^j b_j)$, $m_u^j\in \mathbb Z$, our integral is given by
\begin{align}
 \int_\gamma \omega=\sum_{u=a,b} \sum_{j=1}^g \sum_{l=1}^{2g}\sum_{k=1}^{2g} m_u^j ([\omega] \cdot b_l) c^{lk} \pi^u_j(b_k).
\end{align}
So, the only integrals we have really to compute are $\pi^u_j(b_k)$, which we can call the ``master integrals.'' Notice that this expression resembles the decomposition formulas \eqref{Eq_decomposition}. \\
At this point we may wonder how much all this can be generalized to higher dimensions. To get some insight, let us now present some
very general considerations, without any pretense of completeness, but just with the aim to outline the main concepts one would need to deepen referring to the appropriate literature. \\
{\bf Remark:} in literature, Riemann's bilinear relations usually are not exactly (\ref{RBR}), but rather one of its consequences (\cite{dubrovin1984modern}, Th.12.2): if $\omega_j$, $j=1,2$ are two holomorphic forms on a genus $g$ closed oriented Riemann surface $\Sigma$, and $a_j,b_j$, 
$j=1,\ldots,g$ a canonical basis of $H_1(\Sigma)$, then the following relations are true:
\begin{align}
 \sum_{j=1}^g (\pi^a_j(\omega_1) \pi^b_j(\omega_2)-\pi^a_j(\omega_2) \pi^b_j(\omega_1))&=0\\
 {\rm Im} (\sum_{j=1}^g \overline {\pi^b_j(\omega_1)} \pi^a_j(\omega_1) )&>0,
\end{align}
where $\bar z$ means the complex conjugate of $z$.

\subsection{A twisted version}
We want now to consider a situation a little bit more involved. In place of real or complex valued forms over our Riemann surface $\Sigma_g$, let us consider twisted differential forms taking value in a line bundle $\mathcal L$. This is a bundle whose fibres 
are copies of $\mathbb C_*=\mathbb C-\{0\}$. To it we can associate a dual line bundle $\mathcal L^\vee$, which has the following property: if $\psi$ and $\psi^\vee$ are sections of $\mathcal L$ and $\mathcal L^\vee$, respectively, then, 
$\psi(z)\psi^\vee(z)\equiv f(z)$ defines a function $f:\Sigma_g\rightarrow \mathbb C$. In particular, if $u(z)$ is a section of $\mathcal L^\vee$, then $u(z)^{-1}$ is a section of $\mathcal L$. With this in mind, we say that $\phi_L$ is a $\mathcal L$-valued twisted
$k$-form if for any given section $u$ of $\mathcal L^\vee$, then $u \phi_L$ is a $\mathbb C$-valued $k$-form. We will call $\phi_L$ a left $k$-form, and the space of left $k$-forms is
\begin{align}
 \Omega^k(\Sigma_g, \mathcal L) =\Omega^k(\Sigma_g,\mathbb C) \otimes \mathcal L,
\end{align}
and we write
\begin{align}
 \Omega^*_L=\bigoplus_{k=0}^2  \Omega^k(\Sigma_g, \mathcal L).
\end{align}
We say that $\mathcal L$ is flat if it admits a global section. When it happens, it is clear that also $\mathcal L^\vee$ is flat. Assume $u$ is a global section of $\mathcal L^\vee$: we call it {\it the twist}. 
Given a twist, on $\Omega^*_L$ we can put a structure of ring as follows. For $\phi_1,\phi_2\in \Omega^*_L$, we define
\begin{align}
 \phi_1\wedge_u \phi_2:= u \phi_1\wedge \phi_2,
\end{align}
in harmony with the condition $(u\phi_1)\wedge (u\phi_2)= u (\phi_1 \wedge_u \phi_2)$, required by the mapping $\Omega^k(\Sigma_g,\mathcal L)\rightarrow \Omega^k(\Sigma_g,\mathbb C)$, $\phi\mapsto u\phi$. Notice that 
$\Omega^0(\Sigma_g,\mathcal L)$ are the sections of $\mathcal L$. $u$ also allows us to introduce a connection on $\mathcal L$ that induces a covariant derivative 
\begin{align}
 \nabla_\omega: \Omega^k(\Sigma_g,\mathcal L)\longrightarrow \Omega^{k+1}(\Sigma_g,\mathcal L), \quad i.e. \quad d(u\phi_L)=u\nabla_\omega \phi_L.
\end{align}
Concretely,
\begin{align}
 \nabla_\omega \phi \equiv d+\omega\wedge \phi, \qquad \omega=\frac {du}{u}.
\end{align}
It is a flat connection, since it satisfies $\nabla_\omega^2=0$ as one can easily check. This allows us to say that $\phi$ is closed if $\nabla_\omega \phi=0$, and exact if $\phi=\nabla_\omega \psi$ for $\psi$ a left form of lower rank. Notice that
these definitions imply 
\begin{align}
 \nabla_\omega (\phi_L \wedge_u \tilde \phi_L)=(\nabla_\omega \phi_L) \wedge_u \tilde \phi_L+(-1)^k  \phi_L\wedge_u (\nabla_\omega \tilde \phi_L) -\omega \wedge \phi_L\wedge_u \tilde \phi_L,
\end{align}
if $\phi_L$ has degree $k$.\\
In a similar way, we can define the right forms as the twisted forms with value in $\mathcal L^\vee$, simply by exchanging $\mathcal L \leftrightarrow \mathcal L^\vee$, $u\leftrightarrow 1/u$, $\omega \leftrightarrow -\omega$.
A right form will be called $\phi_R$, when necessary to distinguish it from a left form. Of course, a right form is closed if $\nabla_{-\omega} \phi_R=0$, and exact if $\phi_R=\nabla_{-\omega} \psi_R$.\\
A first interesting fact is that the standard wedge product between a left form and a right form satisfies
\begin{align}
 d(\phi_L \wedge \phi_R)&=d\phi_L \wedge \phi_R+(-1)^{k_L} \phi_L \wedge d\phi_R=(\nabla_\omega -\omega\wedge)\phi_L \wedge \phi_R+(-1)^{k_L} \phi_L \wedge (\nabla_{-\omega} +\omega\wedge)\phi_R\cr
 &=\nabla_\omega \phi_L \wedge \phi_R+(-1)^{k_L}\phi_L \wedge \nabla_{-\omega }\phi_R-\omega\wedge \phi_L \wedge \phi_R+\omega\wedge \phi_L \wedge \phi_R \cr
 &=\nabla_\omega \phi_L \wedge \phi_R+(-1)^{k_L} \phi_L \wedge \nabla_{-\omega }\phi_R, \label{connecting formula}
\end{align}
where we used that $\phi_L\wedge \omega=(-1)^{k_L} \omega \wedge \phi_L$. In particular, if $\phi_L$ and $\phi_R$ are covariantly closed, then, $\phi_L \wedge \phi_R$ is closed in the canonical sense. 
Now, suppose $\phi_L$ and $\phi_R$ are 1-forms, well defined over the Riemann surface 
$\Sigma_g$. As in the previous sections, we can cut it to $\tilde \Sigma_g$, the fundamental $4g$-gone, which is simply connected and, indeed, star shaped with respect to an internal point $o$. 
Therefore, for any $z\in \tilde \Sigma_g$ we can take the segment $[o,z]$ and consider the function
\begin{align}
 \psi_L(z)=e^{-\int_o^z \omega(t)} \int_o^z e^{\int_o^t \omega(s)}\phi_L(s).
\end{align}
It is well defined on $\tilde \Sigma_g$ and satisfies
\begin{align}
 \nabla_\omega \psi_L(z)=\phi_L(z).
\end{align}
The formula (\ref{connecting formula}) then shows that
\begin{align}
 d(\psi_L \phi_R)=\phi_L\wedge \phi_R,
\end{align}
so $\phi_L\wedge \phi_R$ is exact on $\tilde \Sigma_g$.
Proceeding as in the proof of (\ref{RBR}), we can write
\begin{align}
 \int_{\Sigma_g} \phi_L\wedge \phi_R&=\int_{\tilde \Sigma_g} \phi_L\wedge \phi_R=\int_{\tilde \Sigma_g} d(\psi_L \phi_R) =\int_{\partial\tilde \Sigma_g}\psi_L \phi_R
 =\int_{\partial\tilde \Sigma_g} e^{-\int_o^z \omega(t)} \int_o^z e^{\int_o^t \omega(s)}\phi_L(t) \phi_R(z),
\end{align}
and following the same exact procedure, we finally get \cite{cho1995}
\begin{align}
 \int_{\Sigma_g} \phi_L\wedge \phi_R= \sum_{j=1}^g &\left[\left(\int_{a_j} e^{\int_o^t \omega(s)}\phi_L(t)\right)\left(\int_{b_j} e^{-\int_o^t \omega(s)}\phi_R(t)\right) \right. \cr &\left. 
 -\left(\int_{b_j} e^{\int_o^t \omega(s)}\phi_L(t)\right)\left(\int_{a_j} e^{-\int_o^t \omega(s)}\phi_R(t)\right)\right],\label{TRPR}
\end{align}
which is the twisted version of (\ref{RBR}).


\subsection{Perversities and thimbles}\label{Perthimbles}
Now, we are to the point! While we apparently related our integrals just to the cohomology of the Riemann surface defined by the polynomial, the situation is more involved in the general case of a Feynman integral. The point is that the section $u$ is
represented by the Baikov polynomial $\mathcal B$ as $u=\mathcal B^\gamma$. We assume $\gamma\in \mathbb R-\mathbb Z$.
This section has zeros (or singular points if $\gamma<0$) in the zeros of $\mathcal B$, so it doesn't define a global section of a line bundle over the whole domain of the
polynomial, but just in the complement of its zeros. Also, if the integration domain goes towards infinity, it must be chosen so that the integral converges (so along paths where $\mathcal B$ goes to $0$ or to $\infty$, according if $\gamma$ is positive or negative).
This looks like to say that the form becomes trivial around infinity, to make sense to the integral. This leads to think that the right (co)homology to be considered is the one of $\mathbb C^n-Z$ relative to infinity, where $Z$ is the zero locus of the polynomial. 
This is indeed the result of Francis Pham, who proved that the right homology is the relative homology $H_n(\mathbb C^n-Z,B)$, $B$ being a relative neighbourhood of infinity: $B=\{x\in \mathbb C^n| |\mathcal B(x)|>N\}$ where $N$ is ``large enough'', 
\cite{MR713258}, \cite{zbMATH03962140}.\\
The point is that one has to manage the homological tools when singularities are present. Let us consider as an illustrative example the integral
\begin{align}
\int_\Gamma \frac {dz}{(z^3+z)^\alpha},
\end{align}
$\alpha\notin \mathbb Z$. The integral has to converge, so $\Gamma$ has to avoid the zeros of the polynomial, but can reach the infinity, if $\alpha>\frac 13$ and the opposite if $\alpha<\frac 13$. The zeros of the polynomials are branch points and we have to 
fix branch cuts to fix the integrand. Moreover, the map $\mathcal B: \mathbb C \rightarrow \mathbb C$ has two critical points, which are the points where the gradient vanishes. Let us briefly see why this could be interesting. If we want to have intuition on the 
cohomology to be used, we can start thinking that it is not expected to depend on the exact value of $\alpha$. Therefore, after rewriting the integral in the form
\begin{align}
\int_\Gamma e^{-\alpha \log (z^3+z)} dz,
\end{align}
we can estimate it for very large $\alpha$. This can be done by means of the saddle point method, so that the integral is dominated by the saddle points of the exponent, which are indeed the critical points of $\mathcal B$. Depending on $\alpha$, one has 
then to follow the steepest descent or ascent lines, in such a way that the integrand goes to zero. For example, if we write $\mathcal{B}(\zeta_+)=|\mathcal{B}(\zeta_+)| e^{i\phi_+}$, where $\zeta_+$ is a critical point, then
\begin{align}
 \mathcal B(\zeta_+)^{-\alpha}=e^{-\alpha\log |\mathcal{B}(\zeta_+)|-i\alpha \phi_+},
\end{align}
we see that the path passing from the saddle point is conveniently chosen so that the phase is constant so to avoid oscillations. Thus, the best choice are the paths having the constant phase, fixed by the critical point. Then, one chooses the direction such to move
from the critical value to infinity or to zero, depending on the case if the path is allowed to reach infinity or zero (in our example, depending if $\alpha>1/3$ or $\alpha<1/3$). Notice that, having fixed the phase, choosing a direction to move away from the critical
point, the absolute value of the integrand moves monotonically toward $0$ or $\infty$. Also, since we have to move along these paths, we have to choose the cuts so that they do not coincide with (part of) the constant phase integration paths. 
Now that we have got the hint, let us go back to the explicit case in some detail. \\
The branch points are $z_0=0,\ z_\pm= \pm i$. The critical points solve $3z^2+1=0$, so that they are $\zeta_\pm=\pm \frac i{\sqrt 3}$. If we coordinatize the target space of $\mathcal B$ with $t$, then the images of the branch points and the critical points in the
target space are three:
\begin{align}
 t_0\equiv \mathcal B(z_i)=0, \qquad \tau_\pm=\mathcal B(\zeta_\pm)=\pm \frac {2i}{3\sqrt 3}.
\end{align}
The phases determined by $\tau_\pm$ are thus $\phi_{\pm}=\pm \frac \pi2$. Therefore, we choose the cuts with phase $0$. Writing $z=x+iy$, this means
\begin{align}
\mathcal B(z)=x^3-3y^2 x+x +i(-y^3+3x^2 y+y), 
\end{align}
so that the cuts are defined by $-y^3+3x^2 y+y=0$, $x^3-3y^2 x+x>0$.
These are
\begin{align}
y&=0, \qquad x>0; \\
y&=\sqrt {3x^2+1}, \qquad x<0;\cr
y&=-\sqrt {3x^2+1}, \qquad x<0,
\end{align}
see figure figure \ref{fig:thimbles}. In a similar way, we can determine the descent/ascent paths from the critical points. 
\begin{figure}[!htbp]
\begin{center}
\begin{tikzpicture}[>=latex,decoration={zigzag,amplitude=.5pt,segment length=2pt}]
\draw [ultra thick,->] (-3,0) -- (3,0); 
\draw [ultra thick,->] (0,-3.5) -- (0,3.5); 
\draw [ultra thick,->] (4,0) -- (10,0);
\draw [ultra thick,->] (7,-3.5) -- (7,3.5);  
\draw [ultra thick,cyan] (7,1.5) -- (7,3.2);  
\draw [ultra thick,cyan,->] (7,1.5) -- (7,2.5);  
\draw [ultra thick,green] (7,-1.5) -- (7,-3.2);  
\draw [ultra thick,green,->] (7,-1.5) -- (7,-2.5); 
\draw [ultra thick,red] (7,1.5) -- (7,0);  
\draw [ultra thick,red,->] (7,1.5) -- (7,0.65);  
\draw [ultra thick,brown!80!black] (7,-1.5) -- (7,0);  
\draw [ultra thick,brown!80!black,->] (7,-1.5) -- (7,-0.65);  
\draw [ultra thick,violet,decorate] (0,2) arc (270:185:1.5);
\draw [ultra thick,violet,decorate] (0,-2) arc (90:175:1.5);
\draw [ultra thick,cyan] (0,1) arc (270:230:4);
\draw [ultra thick,cyan] (0,1) arc (270:310:4);
\draw [ultra thick,cyan,rotate around={-20:(0,5)},->] (0,1) -- (-0.1,1);
\draw [ultra thick,cyan,rotate around={20:(0,5)},->] (0,1) -- (0.1,1);
\draw [ultra thick,green,rotate around={20:(0,-5)},->] (0,-1) -- (-0.1,-1);
\draw [ultra thick,green,rotate around={-20:(0,-5)},->] (0,-1) -- (0.1,-1);
\draw [ultra thick,green] (0,-1) arc (90:130:4);
\draw [ultra thick,green] (0,-1) arc (90:50:4);
\draw [ultra thick,red] (0,1) -- (0,2);
\draw [ultra thick,red] (0,1) -- (0,0);
\draw [ultra thick,red,->] (0,1) -- (0,1.5);
\draw [ultra thick,red,->] (0,1) -- (0,0.5);
\draw [ultra thick,brown!80!black] (0,-1) -- (0,0);
\draw [ultra thick,brown!80!black] (0,-1) -- (0,-2);
\draw [ultra thick,brown!80!black,->] (0,-1) -- (0,-0.5);
\draw [ultra thick,brown!80!black,->] (0,-1) -- (0,-1.5);
\draw [ultra thick,violet,decorate] (0,0) -- (2.7,0);
\draw [ultra thick,decorate] (7,1.5) -- (9.7,1.5);
\draw [ultra thick,decorate] (7,-1.5) -- (9.7,-1.5);
\filldraw (0,0) circle (2pt) (0,1) circle (2pt) (0,2) circle (2pt) (0,-1) circle (2pt) (0,-2) circle (2pt);
\filldraw (7,0) circle (2pt) (7,1.5) circle (2pt) (7,-1.5) circle (2pt);
\draw [->] (7.6,0) arc (0:90:0.6);
\draw [->] (7.7,0) arc (0:-90:0.7);
\node at (-2.5,3.2) {\bf{{$z$-}plane}};
\node at (4.5,3.2) {\bf{{$t$-}plane}};
\node at (6.7,1.5) {$\pmb{\tau_+}$};
\node at (6.7,-1.5) {$\pmb{\tau_-}$};
\node at (6.7,0.25) {$\pmb{t_0}$};
\node at (7.6,0.6) {$\pmb{\phi_+}$};
\node at (7.7,-0.7) {$\pmb{\phi_-}$};
\node at (0.3,2) {$\pmb{z_+}$};
\node at (0.3,-2) {$\pmb{z_-}$};
\node at (-0.3,0.2) {$\pmb{z_0}$};
\node at (-0.25,0.75) {$\pmb{\zeta_+}$};
\node at (-0.25,-0.75) {$\pmb{\zeta_-}$};
\node [red] at (0.3,1.5) {$\pmb {\Gamma_+^d}$};
\node [cyan] at (-1.4,1.5) {$\pmb {\Gamma_+^a}$};
\node [brown!80!black] at (0.3,-1.5) {$\pmb {\Gamma_-^d}$};
\node [green] at (-1.3,-1.55) {$\pmb {\Gamma_-^a}$};
\end{tikzpicture}
\caption{Cuts, ascent paths and descent paths.}
\label{fig:thimbles}
\end{center}
\end{figure}
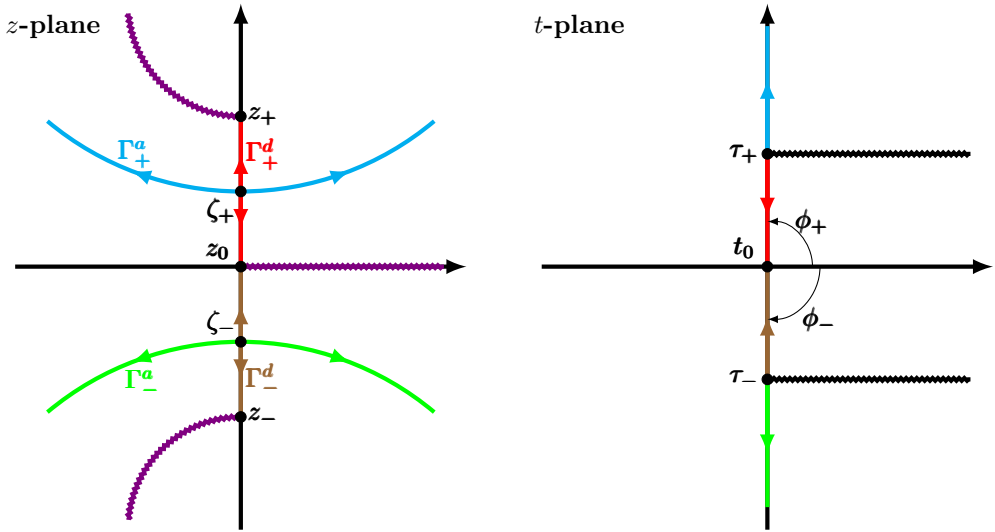
Through $\zeta_\pm$ the paths $\Gamma_\pm$ are defined by the phase $\phi_{\pm}=\pm \frac \pi2$, which correspond to $x^3-3y^2 x+x=0$,
$\pm(-y^3+3x^2 y+y)>0$. These give for $\Gamma_+$ the paths
\begin{align}
 x&=0, \qquad 0<y<1,  & ({\rm descent\ path,}\ \Gamma^d_+);\\
 y&=\sqrt {\frac {x^2+1}3}, & ({\rm ascent\ path,}\ \Gamma^a_+),
\end{align}
for $\Gamma_-$ the paths
\begin{align}
 x&=0, \qquad -1<y<0, & ({\rm descent\ path,}\ \Gamma^d_-);\\
 y&=-\sqrt {\frac {x^2+1}3}, & ({\rm ascent\ path,}\ \Gamma^a_-).
\end{align}
These are curves from infinity to infinity around the cuts (ascent paths), or from a zero to another one (descent paths). They are called {\it Lefschetz thimbles}, for a reason that will become clear soon. Now we want to understand a little better what we have done 
just  inspired by saddle point methods. To this hand, let us consider $\mathcal B$ as a map defining a fibration over $\mathbb C$. The fibre over $t$ is generically given by a set $\mathcal B^{-1}(t)$ of 3 points. The exceptions are the critical points where
the polynomial has double roots, so the fibre is made just by 2 points. For the generic fibre at $t$ (which is zero dimensional) the cohomology is $H^0(\mathcal B^{-1}(t),\mathbb Z)=\mathbb Z^3$, and is invariant when $t$ varies, so we are in trouble
if we move to a critical point. Therefore, to manage this difficulty, it is convenient to restrict this ``fibration'' of cohomology groups to $\mathbb C_\tau\equiv \mathbb C-\{\tau_+,\tau_-\}$ and look at what happens if we move around the critical values. To see this,
notice that if $\lambda:=e^{\frac {2\pi i}3}$ is a primitive third root of $1$, then, 
\begin{align}
 \mathcal B^{-1}(t)=
\begin{pmatrix}
 \sqrt[3]{\frac t2 +\sqrt {\frac {t^2}4 +\frac 1{27}}}+ \sqrt[3]{\frac t2 -\sqrt {\frac {t^2}4 +\frac 1{27}}} \\
 \lambda \sqrt[3]{\frac t2 +\sqrt {\frac {t^2}4 +\frac 1{27}}}+ \lambda^2\sqrt[3]{\frac t2 -\sqrt {\frac {t^2}4 +\frac 1{27}}} \\
 \lambda^2 \sqrt[3]{\frac t2 +\sqrt {\frac {t^2}4 +\frac 1{27}}}+ \lambda\sqrt[3]{\frac t2 -\sqrt {\frac {t^2}4 +\frac 1{27}}} 
\end{pmatrix},
\end{align}
from which we can easily see what happens to its (co)homology when we move on a loop around the critical points: since $\sqrt {\frac {t^2}4 +\frac 1{27}}$ vanishes at $t=\tau_\pm$, it changes sign in moving along the loop, and we see that the upper component
is invariant, while the lowest components are interchanged. So, on $V_t \equiv H^0(\mathcal B^{-1}(t),\mathbb Z)=\mathbb Z^3$, the monodromy action (the action of moving around the points) around $\tau_\pm$ is given by the matrices 
\begin{align}
 M_\pm =
\begin{pmatrix}
 1 & 0 & 0 \\ 0 & 0 & 1 \\ 0 & 1 & 0
\end{pmatrix}.
\end{align}
The badness of the fibration is measured by these (here coincident) monodromy matrices. The point is that $V_t$ has fixed points given by the eigenspaces of $M_\pm$ corresponding to the eigenvalue 1. These define the spaces of invariants
\begin{align}
 V^{M_\pm} \simeq \mathbb Z \
\begin{pmatrix}
 1 \\ 0 \\ 0
\end{pmatrix}
\oplus \mathbb Z 
\begin{pmatrix}
 0 \\ 1 \\ 1
\end{pmatrix}.
\end{align}
Their complement are the spaces of coinvariants 
\begin{align}
 V_{M_\pm} \simeq \mathbb Z 
\begin{pmatrix}
 0 \\ 1 \\ -1
\end{pmatrix}.
\end{align}
On them the monodromies act as minus the identity. The strategy for capturing the right cohomology, is to restrict the fibration of cohomologies, keeping only the coinvariant parts and the information of the action of the monodromies. Notice that 
these are irreducible non trivial representations of $\pi_1 (\mathbb C_\tau)$, the first homotopy group (generated by the loops around the critical values). This ``almost fibration'' with specification of monodromy constrained structures at singular points
substantially describes what is called a perverse sheaf. The family of coinvariant (co)homology elements can be then interpreted as follows. In the homology of $\mathcal B^{-1}(t)$ they can be thought as 
\begin{align}
 &0\cdot \left( \sqrt[3]{\frac t2 +\sqrt {\frac {t^2}4 +\frac 1{27}}}+ \sqrt[3]{\frac t2 -\sqrt {\frac {t^2}4 +\frac 1{27}}} \right) +p\cdot \left( \lambda \sqrt[3]{\frac t2 +\sqrt {\frac {t^2}4 +\frac 1{27}}}+ \lambda^2\sqrt[3]{\frac t2 -\sqrt {\frac {t^2}4 +\frac 1{27}}}\right) \cr
 &-p\cdot \left(\lambda^2 \sqrt[3]{\frac t2 +\sqrt {\frac {t^2}4 +\frac 1{27}}}+ \lambda\sqrt[3]{\frac t2 -\sqrt {\frac {t^2}4 +\frac 1{27}}} \right)
\end{align}
with $p\in \mathbb Z$. The generator is for $p=1$. Then when $t$ varies from a singular point to infinity (or to 0), we see that these points collapse at the critical point, collapsing to zero also as an element of $H_0(\mathcal B^{-1}(t),\mathbb Z)$, since we get the same 
point with total coefficient $1-1=0$.\footnote{If we considered the invariants, the result at the point would be that the points meet with total coefficient $1+1=2$.} These correspond to the thimbles, and, are called vanishing cycles. Notice that
the pair of points are $0$ dimensional cycles. For a polynomial in $n$ variables, one gets that the preimages of $t$ are generically $\mathcal B^-1(t)\simeq TS^{n-1}$, the total space of the tangent bundle of a sphere, which are homologically equivalent
to spheres $S^{n-1}$. The (perverse) cohomology so generated, is represented by union of spheres when $t$ varies, for example, from a critical value to $\infty$ along a line $\gamma$. The result is homotopic to a cylinder $S^{n-1}\times \gamma$, whose face
at the end corresponding to the critical value collapses to a point. This may be thought to have the shape of a thimble.\\
Of course, a rigorous version of this construction in a general setting requires the use of more sophisticated mathematical tools. 

\section{Some general constructions}\label{veryhardmaths}
In this section we want to discuss some relations between Feynman integrals, homologies and cohomologies. Even if we will avoid entering the details, we will illustrate the main ideas on how to deal with singularities by passing to simplicial
(co)homologies, while our original problem remains the one of computing integrals of differential forms. Therefore, we will start by recalling some relations among the main (co)homologies we will need to consider in the successive subsections.
We will not really provide notions but rather only rough ideas and some literature. However, we hope to make this section at least readable.

\subsection{On cohomologies and de Rham Theorem}
In general, for a given space, there can be several associated homologies depending on the structures one is considering. We are of course interested in the de Rham cohomology, which is well defined on smooth manifolds as the quotient of closed forms with
exact forms. But there are other (co)homologies that are of interest for concrete calculations. For the definitions reported below and any further lecture we refer to \cite{dubrovin1984modern},  \cite{bredon1993topology}.
\subsubsection{Singular (co)homology }
One of the most general is {\it singular homology}. A $k$-simplex $\Delta_k$ is the smallest convex subset of $\mathbb R^N$ (N>k) generated by $k+1$ points in 
general position (so it has dimension $k$). If $x_j$, $j=0,\ldots,k$ are the generating points, it is convenient to use the notation $\Delta_k=\{x_0,\ldots,x_k\}$. Then, the $j$-th face of the simplex is the $(k-1)$-simplex 
$\Delta_{k-1}^{j}=\{x_0\ldots, \hat x_j, \ldots, x_n\}$, where the $hat$ means omission of the point. 
If $X$ is a topological space, then a {\it singular $k$-simplex} in $X$ is a pair $(\Delta_k,\phi)$, where $\phi:\Delta_k\rightarrow X$ is a continuous map. A singular $k$-chain $c_k$ is a formal finite combination of singular $k$-simplexes 
$c_k=\sum_\alpha a_\alpha (\Delta_{k,\alpha},\phi_\alpha)$, $a_\alpha\in \mathbb K$, where $\mathbb K$ is an abelian group, for us $\mathbb K=\mathbb Z, \mathbb Q, \mathbb R, \mathbb C$. They form the space $C^{sing}_k(X)$ of singular chains. 
The boundary of a singular simplex $(\Delta_k,\phi)$
is the singular $(k-1)$-chain
\begin{align}
\partial (\Delta_k,\phi)=\sum_{j=0}^k (-1)^j (\Delta_{k-1}^j, \phi|_{\Delta_{k-1}^j}).
\end{align}
This defines by $\mathbb K$-linear extension, a sequence of $\mathbb K$-linear maps 
\begin{align}
 \partial_k: C^{sing}_k(X) \longrightarrow C^{sing}_{k-1}(X), \quad c_k\mapsto \partial c_k.
\end{align}
It follows that $\partial_{k-1}\circ \partial_{k}=0$ and one can define the {\it singular homology groups}
\begin{align}
 H^{sing}_k(X,\mathbb K)=\frac {{\rm ker}(\partial_k)}{{\rm Im}(\partial_{k+1})}.
\end{align}
Similarly, we can define a singular $k$-cochain as a linear map
\begin{align}
 \psi^k: C^{sing}_k(X) \longrightarrow \mathbb K,
\end{align}
that form the $\mathbb K$-linear space $C_{sing}^k(X)$, on which acts the coboundary operator 
\begin{align}
 \delta^k : C_{sing}^k(X) \longrightarrow C_{sing}^{k+1}(X), \quad i.e.\quad\ \delta^k \psi^k (c_{k+1})\equiv \psi^k(\partial_{k+1} c_{k+1}), \ \forall c_{k+1}\in C^{sing}_{k+1}(X).
\end{align}
From these, one defines the {\it singular cohomology groups}
\begin{align}
 H_{sing}^k(X,\mathbb K)=\frac {{\rm ker}(\delta^k)}{{\rm Im}(\delta^{k-1})}.
\end{align}
{\bf Remark:} replacing simplexes with hypercubes, one gets the singular cubic (co)homology. However, this gives a non normalized cohomology \cite{dubrovin1984modern} (also called a generalized cohomology), that is a cohomology that does not satisfy the condition
$H_k(p)=0$ for $k>0$ and $p$ a point. 
\subsubsection{Simplicial (co)homology }
This is less general and is first defined for simplicial complexes. A simplicial complex $\Xi$ is a collection of simplexes such that if a simplex is in the complex, then all its faces are also in, and, moreover, two simplexes in the collection can intersect in one and 
only one common sub-face of a given dimension. It is called finite, if composed by a finite number of simplexes, while it is locally finite if for any point there is a small neighbourhood intersecting just a finite number of simplexes. Given a simplicial complex $\Xi$,
one defines its {\it simplicial $k$-chains} as the formal combinations $c_k=\sum_\alpha a_\alpha \Delta_{k,\alpha}$, $a_\alpha\in \mathbb K$, which define the space $C_k^{simp}(\Xi)$. Defining the boundary
\begin{align}
\partial \Delta_k=\sum_{j=0}^k (-1)^j \Delta_{k-1}^j,
\end{align}
we get linear maps 
\begin{align}
 \partial^s_k: C^{simp}_k(\Xi) \longrightarrow C^{simp}_{k-1}(\Xi), \quad c_k\mapsto \partial c_k,
\end{align}
which satisfy $\partial_{k-1}\circ \partial_{k}=0$ and allow to define the simplicial homology groups
\begin{align}
 H^{simp}_k(\Xi,\mathbb K)=\frac {{\rm ker}(\partial^s_k)}{{\rm Im}(\partial^s_{k+1})}.
\end{align}
Similarly, we can define a simplicial $k$-cochain as a linear map
\begin{align}
 \xi^k: C^{simp}_k(\Xi) \longrightarrow \mathbb K,
\end{align}
that form the $\mathbb K$-linear space $C_{simp}^k(\Xi)$, on which acts the coboundary operator 
\begin{align}
 \delta_s^k : C_{simp}^k(\Xi) \longrightarrow C_{simp}^{k+1}(\Xi), \quad i.e.\quad \ \delta_s^k \xi^k (c_{k+1})\equiv \xi^k(\partial_{k+1} c_{k+1}), \ \forall c_{k+1}\in C^{simp}_{k+1}(\Xi).
\end{align}
From these, one defines the {\it simplicial cohomology groups}
\begin{align}
 H_{simp}^k(\Xi,\mathbb K)=\frac {{\rm ker}(\delta_s^k)}{{\rm Im}(\delta_s^{k-1})}.
\end{align}
This construction can then be extended to the case of a topological space $X$ that admits triangulations. If a given triangulation $T$ can be homotopically deformed to become homeomorphic to a simplicial complex $\Xi_T$, then we can define the
simplicial homology groups of $X$ as
\begin{align}
 H^{simp}_k(X,\mathbb K)=H^{simp}_k(\Xi_T,\mathbb K), \qquad  H_{simp}^k(X,\mathbb K)=H_{simp}^k(\Xi_T,\mathbb K).
\end{align}
These are well defined, since it can be proved that the resulting groups are independent on the triangulation $T$. An important result is that the following isomorphisms (for spaces admitting simplicial homology) hold true:
\begin{align}
 H^{sing}_k(X,\mathbb K)\simeq H^{simp}_k(X,\mathbb K), \qquad  H_{sing}^k(X,\mathbb K)\simeq H_{simp}^k(X,\mathbb K).
\end{align}
For a proof, see \cite{dubrovin1984modern}, \cite{bredon1993topology}.

\subsubsection{de Rham theorem}
Now we can state the following important result (\cite{GeorgesdeRham1931,Weil1952,10.2307/2372421,samelson1967,dubrovin1984modern,bredon1993topology}):

\

{\bf de Rham Theorem.} {\it If $X$ is a smooth manifold and $\mathbb K=\mathbb R, \mathbb C$, then
\begin{align}
 H^{de Rham}_k(X,\mathbb K)&\simeq H^{sing}_k(X,\mathbb K)\simeq H^{simp}_k(X,\mathbb K), \\
 H_{de Rham}^k(X,\mathbb K)&\simeq H_{sing}^k(X,\mathbb K)\simeq H_{simp}^k(X,\mathbb K).
\end{align}
}

An important remark for us is that the same theorem can be extended by weakening the hypothesis and assuming that $X$ admits singularities and more in general a (Withney) stratification or the structure of a stratifold. See \cite{Ewald2004} and
references therein.\\
This result allows us to move to simplicial (co)homology to simplify abstract constructions as we did in section \ref{HardMaths}. But it is also helpful for simplifying certain computations like the one of Betti number and the Euler characteristic.
Here we recall that the Betti numbers of $X$ are the numbers
\begin{align}
 b_k={\rm dim} H_{de Rham}^k(X,\mathbb K),
\end{align}
while the Euler characteristic of $X$ is the number
\begin{align}
 \chi(X)=\sum_{j=0}^N (-1)^j b_j,
\end{align}
where $N={\rm dim} X$.

Finally, we remark that further cohomologies can be introduced, like for example relative (co)homologies $H_k(X,B,\mathbb K)$, cellular (co)homologies, \v Cech and sheaf (see \cite{dubrovin1984modern}), or generalized (co)homologies, like $K$-theories or (co)bordisms. 
We demand any further reading to the standard literature.

\subsection{Cup products and intersections on smooth manifolds}
In the example of closed one forms on a compact closed Riemann surface $\Sigma$, we have seen that the cup product $\cup: H^1(\Sigma)\times H^1(\Sigma)\rightarrow H^2(\Sigma)$ can be interpreted as the intersection product between 1-forms,
in a sense, dual to the intersection product among homological curves (1-cycles), probably more intuitive to be figured out. Now, we want to understand how much this notion can be generalized. To this hand we will refer to the nice lecture in
\cite{Hutchings2011CupPA}. \\
Let $\mathcal M$ be an $m$-dimensional oriented compact smooth manifold. Also, we assume $S_1$ and $S_2$ are two smooth submanifolds of dimensions $m_1$ and $m_2$ respectively. We say that they have transversal intersection if 
$N=S_1\cap S_2\neq \emptyset$, and for any $x\in N$ the union of the embeddings of $\iota_j:T_xS_j \hookrightarrow T_xM$ is the whole $T_xM$. In particular, this implies that $m_1+m_2-m=m_N\geq0$, which is the dimension of the intersection. 
We have to take into account the orientation. We say that the intersection is oriented on a given connected component $\bar N$ of $N$ if it is possible to choose a basis $\pmb v_1,\ldots \pmb v_{m_N} \in T_xN$, $x\in \bar N$, such that:
\begin{itemize}
 \item we can complete it to an oriented basis $\pmb w_1, \ldots, \pmb w_{m_1-m_N}, \pmb v_1,\ldots \pmb v_{m_N} \in T_xS_1$ according to the orientation of $S_1$;
 \item we can complete it to an oriented basis $\pmb v_1,\ldots \pmb v_{m_N}, \pmb z_1, \ldots, \pmb z_{m_2-m_N} \in T_xS_2$ according to the orientation of $S_2$;
 \item $\pmb w_1,\ldots \pmb w_{m_1-m_N}, \pmb v_1,\ldots \pmb v_{m_N}, \pmb z_1, \ldots, \pmb z_{m_2-m_N}$ is an oriented basis for $T_xM$ according to the orientation of $M$.
\end{itemize}
See Fig. \ref{Fig3} for a pictorial representation.
\begin{figure}[!htbp]
\begin{center}
\begin{tikzpicture}[>=latex,decoration={zigzag,amplitude=.5pt,segment length=2pt}]
\shade [top color=green!50!black, middle color=green!30!black] (-5.5,-3) .. controls (-4.5,-2.6) and (-3.75,-2.35) .. (-3.5,-0.9) -- (-0.5,-0.9) .. controls (-0.5/1.5,-0.9-2.9/1.5) and (1-0.5/1.5,-0.9-2.9/1.5) .. (1,-0.9) -- (3,-0.9) .. controls (3.05,-0.9-3.9/2) and (4,-2.5) .. 
(4.5,-3) .. controls (3,-3.4) and (-4,-3.4) .. (-5.5,-3) -- cycle;
\draw [fill=gray!30!white] (-4,0) -- (4,0) -- (5,-2) -- (-6,-2) -- cycle;
\draw [red] (-2,-0.9) ellipse (1.5 and 0.4);
\draw [red] (2,-0.9) ellipse (1 and 0.25);
\shade [ball color=green!] (-3.5,-0.9) .. controls (-3,2) and (-1,2) .. (-0.5,-0.9) arc (0:-180:1.5 and 0.4) -- cycle;
\shade [ball color=green!] (1,-0.9) .. controls (1.5,2) and (2.9,3) .. (3,-0.9) arc (0:-180:1 and 0.25) -- cycle;
\draw [red,ultra thick] (-3.49,-0.9) arc (180:360:1.49 and 0.39);
\draw [red,ultra thick,opacity=0.4] (-3.49,-0.9) arc (180:0:1.49 and 0.39);
\draw [red,ultra thick] (1.01,-0.9) arc (180:360:0.99 and 0.24);
\draw [red,ultra thick,opacity=0.4] (1.01,-0.9) arc (180:0:0.99 and 0.24);
\draw [ultra thick,red,->] (-2,-1.3) -- (-1.9,-1.3);
\draw [ultra thick,red,->] (2,-1.15) -- (2.1,-1.15);
\node [red] at (0.4,-3) {$\pmb{N=S_1\cap S_2}$};
\draw [red,ultra thick,->] (0,-2.7) .. controls (-0.2,-2) and (-1.7,-2.8) .. (-1.7,-1.4);
\draw [red,ultra thick,->] (0.4,-2.7) .. controls (0.6,-2) and (1.7,-2.8) .. (1.7,-1.2);
\node at (-5.2,-1.7) {$\pmb {S_1}$};
\node [green] at (-4.2,-2.7) {$\pmb {S_2}$};
\draw [thick,->] (-4.6,-0.9) -- (-4,-0.9);
\draw [thick,->] (-4.6,-0.9) -- (-4.3,-0.6);
\node at (-4.1,-1.05) {\tiny{\pmb{1}}};
\node [rotate=-33.69] at (-4.5,-0.6) {\tiny{\pmb 2}};
\draw [thick,->,green] (-3.2,-2.8) -- (-2.6,-2.8);
\draw [thick,->,green] (-3.2,-2.8) -- (-3,-2.4);
\node [green] at (-2.75,-2.95) {\tiny{\pmb{1}}};
\node [rotate=-26.5651,green] at (-3.15,-2.45) {\tiny{\pmb 2}};
\draw [thick,cyan,->] (-5,2) -- (-5,2.7);
\draw [thick,cyan,->] (-5,2) -- (-4.3,2);
\draw [thick,cyan,->] (-5,2) -- (-5.5,1.5); 
\node [cyan] at (-4.44,1.85) {\tiny{\pmb{1}}};
\node [cyan] at (-5.5,1.7) {\tiny{\pmb{2}}};
\node [cyan] at (-5.15,2.6) {\tiny{\pmb{3}}};
\node [cyan] at (-3,2.6) {$\pmb M$};
\end{tikzpicture}
\caption{Intersection between two oriented surfaces $S_1$ and $S_2$ in a three dimensional oriented manifold $M$.}
\label{Fig3}
\end{center}
\end{figure}
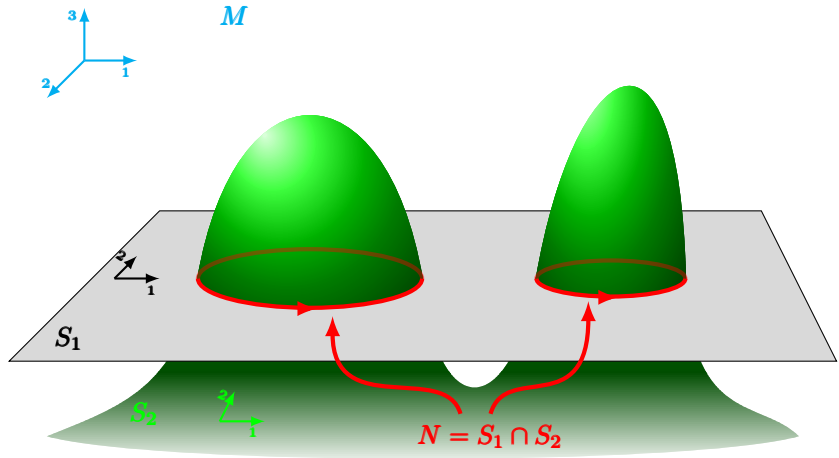
It is clear that it depends only on the connected component and not on the specific point. So, $\pmb v_1,\ldots \pmb v_{m_N}$ defines the orientation of the intersection.
We can consider the homology class $[S_1\cap S_2]$ associated to $N$, in the sense of simplicial homology.\footnote{so, $[N]$ is defined by any given simplicial decomposition of $N$.}
It can be shown that it depends only on the 
homology classes of $S_j$ so it gives a map $H_{m_1} (M,\mathbb Z)\times H_{m_2}(M,\mathbb Z)\rightarrow H_{m_1+m_2-m}(M,\mathbb Z)$, each time $m_1+m_2-m\geq0$.  \\
In particular, when $m_1+m_2=m$, $N$ is just the union of a set of isolated points. This set is finite since $M$ is compact, and defines an element of $H_0(M,\mathbb Z)\simeq \mathbb Z$ (since $M$ is connected). 
Thus we can associate to the intersection product a number representing such element. This can be done as follows. If $\sigma(x)$ is the sign associated to $x\in N$, we set
\begin{align}
S_1\cdot S_2 := \sum_{x\in N} \sigma(x).
\end{align}
It is well defined, since $N$ is finite, and it results to depend on the homology classes only. So, it shows the connection with the intersection product we met in the case of a Riemann surface.\\
However, let us stay general for the moment. There is another interesting map to be considered in simplicial (co)homology in order to go ahead: the {\it cap product}
\begin{align}
 \frown : H^{j}(M,\mathbb Z) \times H_m(M,\mathbb Z) \longrightarrow H_{m-j} (M,\mathbb Z), \quad (\mu,u) \mapsto \mu\frown u,
\end{align}
defined in such the way that for any given $\nu\in H^{m-j}(M,\mathbb Z)$, we must have $\nu(u\frown \mu)\equiv(\nu\cup \mu) (u)$. Then, let $[M]\in H_m(M,\mathbb Z)$ be the class generated my $M$, also called the {\it fundamental class}. Fixing
$u=[M]$, we get a map
\begin{align}
 \star: H^j(M,\mathbb Z)\longrightarrow H_{m-j} (M,\mathbb Z),\quad \mu\mapsto \star \mu:= \mu \frown [M].
\end{align}
Then, this theorem holds (\cite{maxim2019intersection}, Th. 1.1.3): 

\

\begin{Theorem}[Poincar\'e duality]
 If $M$ is a connected, closed and oriented topological manifold, then $\star: H^j(M,\mathbb Z)\longrightarrow H_{m-j} (M,\mathbb Z)$ is an isomorphism for all integers $j$.
\end{Theorem}


\

\noindent The most important consequence of this theorem, at least as what regards the applications we are interested in, is the following theorem (\cite{Hutchings2011CupPA}, Th. 1.1):

\

\begin{Theorem}
Assume $M$ is a compact, connected, closed smooth manifold and let $S_j$ smooth submanifolds of dimension $m_j$, $j=1,2$ that (up to an homological deformation) have transversal intersection. Then, it holds
\begin{align}
 \star [S_1] \cup \star [S_2]=\star[S_1 \cap S_2].
\end{align} 
\end{Theorem}

\

\noindent For a proof, see \cite{Hutchings2011CupPA}. This theorem thus states that cup and cap products are dual under the Poincar\'e map. So, it provides a remarkable generalization of the results discussed for Riemann surfaces. We described it in the case of
simplicial homology, but it clearly works as well if we tensorize with $\mathbb K=\mathbb R,\mathbb C$, and then we can use the isomorphism $H_{simp}^*(M,\mathbb K)\simeq H_{dR}^*(M,\mathbb K)$  to 
get it for the case of the de Rham cohomology (where the cup product is induced by the wedge product), which better matches with the previous section. \\
However, we are not yet happy enough, since we had to make use of smoothness, but, in considering the general integrals we are interested in, we generically have to consider manifolds with singularities, which may be not homologous to smooth manifolds,
so we need to explore further generalizations.

\subsection{Lefschetz theorems, Hodge-Riemann bilinear relations and perversities}\label{Perversion}
We want to go beyond in our tour on intersection theory by looking at finer structures. We will follow mainly \cite{decataldo2009the}, indicating some further lectures explicitly when necessary.\\
First, let us go back to the case of Riemann surfaces. They are complex curves of complex dimension 1, and smooth real surfaces of real dimension 2. They can be embedded in a projective space, for example the elliptic curves can be embedded in
$\mathbb P^2$ as the zero locus of a quartic or cubic polynomial, such as
\begin{align}
y^2=P_4(x), 
\end{align}
in local non homogeneous coordinates. If in the above equation in $\mathbb P^2$ we replace $P_4$ with a polynomial $P_d$ of degree $d=2g$ or $d=2g-1$, we get an hyperelliptic complex curve of genus $g$. They are therefore {\it algebraic varieties},
that are zeros of polynomial functions in some real or complex space. When an algebraic variety can be embedded in a projective space $\mathbb P^N$ for some $N$, then it is called a {\it projective variety}. Also, we can notice that the one forms on
a Riemann surface $\Sigma$, in local complex coordinates $z$ and $\bar z$, have the general form
\begin{align}
 \omega=f(z,\bar z) dz+g(z,\bar z)d\bar z,
\end{align}
with differential
\begin{align}
 d\omega= (\partial_z g-\partial_{\bar z} f) dz\wedge d\bar z. \label{exactform}
\end{align}
Any two form is
\begin{align}
\Omega=\lambda(z,\bar z) dz\wedge d\bar z, 
\end{align}
and it is closed, $d\Omega=0$ for dimensional reasons. It is cohomologically trivial if it has the form (\ref{exactform}). $\Sigma$ is also a K\"ahler manifold, which means the following. It is always possible to put on it a Hermitian metric 
\begin{align}
 G=g_{z\bar z} dz\otimes d\bar z,
\end{align}
that defines at any point $p$ of $\Sigma$ a sesquilinear map
\begin{align}
G_p: T_p\Sigma \times T_p\Sigma \longrightarrow \mathbb C,
\end{align}
such that 
\begin{align}
 G_p(\alpha \xi, \zeta)=G(\xi, \bar \alpha \zeta)
\end{align}
for any $\alpha\in \mathbb C$. Its real part defines a Riemannian metric on $\Sigma$, while its imaginary part defines a two form
\begin{align}
 K=-\frac i2 g_{z\bar z} dz \wedge d\bar z,
\end{align}
which is a closed form, called the K\"ahler form. The same construction can be made in higher dimensions, with the only difference that in general the K\"ahler form is not automatically closed. When it is, then the complex manifold is said to be a K\"ahler manifold.
It is possible to prove that for a K\"ahler manifold $M$, in a local coordinate patch $U$ it is always possible to find a real valued function $\mathcal K:U\rightarrow \mathbb R$, called the {\it K\"ahler potential}, such that 
\begin{align}
g_{z\bar z}=\partial_z\partial_{\bar z} \mathcal K. 
\end{align}
An example are the complex projective manifolds $\mathbb P^N$, that, in non homogeneous coordinates $z_1,\ldots,z_N$, admit the K\"ahler potential
\begin{align}
 \mathcal K_{FS} =\log (1+\sum_{j=1}^N |z_j|^2),
\end{align}
generating the Fubini-Study Hermitian metric
\begin{align}
 G_{FS}=\frac 1{(1+\sum_{j=1}^N |z_j|^2)^2} \left[ (1+\sum_{j=1}^N |z_j|^2) \sum_{k=1}^N dz_k\otimes d\bar z_k - \sum_{j,k=1}^N z_j \bar z_k  dz_j\otimes d\bar z_k\right].
\end{align}
Any projective manifold is a K\"ahler manifold, since it inherits a K\"ahler structure after restriction of the Fubini-Study of the ambient space $\mathbb P^N$ in which it is embedded.\\
The cohomology of a K\"ahler manifold of complex dimension $m$ has a nice structure, called a Hodge structure. For any $k=0, 1, \ldots, 2m$, the cohomology group $H^k(M,\mathbb C)$ admits a direct decomposition
\begin{align}
 H^k(M,\mathbb C)=\bigoplus_{p+q=k} H^{p,q}(M,\mathbb C),
\end{align}
where $p$ and $q$ run in $\{0,\ldots,m\}$, and the elements of $H^{p,q}(M,\mathbb C)$ are represented by closed forms of the type
\begin{align}
 \omega_{p,q}=\sum_{j_1,\ldots,j_p,k_1,\ldots,k_q} \omega_{j_1\ldots j_p k_1\ldots k_q} dz^{j_1} \wedge \cdots \wedge dz^{j_p} \wedge d\bar z^{k_1}\wedge \cdots d\bar z^{k_q}.
 \end{align}
For example, for connected Riemann surfaces of genus $g$ we have
\begin{align}
 H^0(\Sigma,\mathbb C) &\simeq \mathbb C, \\
 H^1(\Sigma,\mathbb C)&=H^{1,0} (\Sigma,\mathbb C)\oplus H^{0,1}(\Sigma,\mathbb C)\simeq \mathbb C^g\oplus \mathbb C^g,\\
 H^2(\Sigma,\mathbb C)&=H^{2,0}(\Sigma,\mathbb C)\oplus H^{1,1}(\Sigma,\mathbb C)\oplus H^{0,2}(\Sigma,\mathbb C)\simeq 0\oplus \mathbb C\oplus 0.
\end{align}
In particular, $H^{1,0}(\Sigma,\mathbb C)$ is generated by $g$ holomorphic forms, which for an hyperelliptic curve in $\mathbb P^2$
\begin{align}
 y^2=P(z),
\end{align}
are
\begin{align}
 \omega_j= \frac {z^j dz}{y}, \quad j=0,\ldots, g-1.
\end{align}
$H^2(\Sigma,\mathbb C)$ is generated by the K\"ahler form, which is also real valued, $K\in H^{1,1}(M,\mathbb C)\cap H^2(M,\mathbb R)$. \\
Now, unfortunately, it is time to become more abstract. Given a cohomology group of any kind, let us think at it as a finitely generated abelian group $H$ over $\mathbb Z$, and, following \cite{decataldo2009the}, let use the notation 
$H_{\mathbb K}=H\otimes_{\mathbb Z} \mathbb K$, $\mathbb K=\mathbb Q, \mathbb R, \mathbb C$. We say that $H$ carries a {\it pure Hodge structure} of weight $k$ if there is a decomposition of abelian groups
\begin{align}
 H_{\mathbb C}=\bigoplus_{p+q=k} H^{p,q}, \quad i.e.\quad \overline{H^{p,q}}=H^{q,p}.
\end{align}
An equivalent formulation is to say that $H$ admits a decreasing Hodge Filtration, that is a family of abelian subgroups $F^p$ of $H_{\mathbb C}$ with the properties
\begin{align}
 H_{\mathbb C}&=F^0\supset F^{1} \supset \cdots \supset F^p \supset F^{p+1}\supset \cdots, \\
 F^p \cap \overline {F^{k+1-p}}&=0, \\
 F^p \oplus \overline {F^{k-p}}&=H_{\mathbb C}.
\end{align}
Indeed, the equivalence is read by the relations
\begin{align}
 H^{p,q}&=F^p\cap \overline{F^q},\\
 F^p&=\bigoplus_{j\geq p} H^{j,k-j}.
\end{align}
On $H_{\mathbb C}$ there is a real action of $S^1\subset \mathbb C_*$ given by $\rho(s) a=s^{p-q} a$ for $a\in H^{p,q}$, and extended by $\mathbb R$ linearity. In particular, $\rho(i)=w$ is called the Weyl map.
A {\it polarization} on the pure Hodge structure is a real bilinear form $Q$ on $H_{\mathbb R}$, which is invariant under the action of $S^1$ and such that 
\begin{align}
 B(a,b):=Q(a,wb)
\end{align}
is symmetric and positive definite. After extending $Q$ over $H_{\mathbb C}$, one gets that $i^{2k+p-q} Q(a,\bar a)>0$ for a non vanishing $a\in H^{p,q}$.\\
Now, consider a non singular algebraic projective variety of dimension $m$, $X\subset \mathbb P^N$ (for instance a projective manifold). Let $H$ be a generic hyperplane in $\mathbb P^N$. Then, $H\cap X$ defines an element of
$H_{n-2}(X,\mathbb Z)$, whose Poincar\'e dual is $\eta\in H^2(X,\mathbb Z)$, see \cite{maxim2019intersection}. We then define the Lefschetz map
\begin{align}
 L: H^k(X,\mathbb Q) \longrightarrow H^{k+2}(X,\mathbb Q), \quad a\mapsto L(a)=\eta \cup a.
\end{align}
The degree $k$ primitive vector spaces are the subspaces $P^k\subseteq H^k(X,\mathbb Q)$ defined by $P^k:={\rm ker} L^{m-k+1}$. Finally, on $H^k(X,\mathbb Q)$ let us define the quadratic form
\begin{align}
 Q^{(k)}_L([\omega],[\omega']):=(-1)^{\frac {k(k+1)}2} \int_X \eta^{m-k}\wedge \omega \wedge \omega'.
\end{align}
Then, $H^k(X,\mathbb Z)$ is a pure Hodge structure of weight $k$ and $P^k$ is a rational hodge structure of weight $k$. Then we have the following important properties (\cite{decataldo2009the}, Th.5.2.1).

\

\begin{Theorem}[Hard Lefschetz Theorem]
 With the above notations, the maps
\begin{align}
L^k:H^{m-k}(X,\mathbb Q)\longrightarrow H^{m+k}(X,\mathbb Q)
\end{align}
are isomorphisms.
\end{Theorem}


\

\

\begin{Theorem}[Primitive Lefschetz decomposition]
 For $0\leq k\leq m$ one has the decomposition
\begin{align}
H^k(X,\mathbb Q)=\bigoplus_{j\geq 0} L^j P^{k-2j}. 
\end{align}
Moreover, for any given $j$, $V^j\equiv L^j P^{k-2j}$ is a pure Hodge substructure of weight $k$, and for any $j\neq j'$, $V^j$ and $V^{j'}$ are $Q^{(k)}_L$-orthogonal.
\end{Theorem}

\

\

\begin{Theorem}[Hodge-Riemann bilinear relations]
 For any $0\leq k\leq m$, $Q^{(k)}_L$ is a polarization for the pure Hodge structure $P^k_{\mathbb R}$. In particular, after complexification,
\begin{align}
 i^{p-q} Q^{(k)}_L ([\omega],[\bar\omega])>0, \qquad \forall [\omega]
\in P^k\cap H^{p,q}(X,\mathbb C),
\end{align}
so it is strictly positive definite.
\end{Theorem}

\

This result seems to give us a very broad generalization of the properties we met for integrals along one dimensional cycles on a Riemann surface. But, again, it is not enough, since we had to require for $X$ to be a nonsingular variety. The simplest
example is $\mathbb C_*=\mathbb C\backslash \{0\}$. Its first cohomology group is obviously one dimensional, and this fact is incompatible with the existence of a pure Hodge structure. Indeed, assuming $H^{2k+1}$ has a pure Hodge structure, then
the relation $H^{p,q}=\overline {H^{q,p}}$ implies
\begin{align}
 {\rm dim} (H^{2k+1})=2\sum_{j=0}^k {\rm dim}(H^{j,2k+1-j})\in 2\mathbb N,
\end{align}
so it must be even dimensional.\\
Since we have to work with general Feynman integrals, we cannot avoid to work with singular varieties, so we need to overcome this difficulty. Luckily, this has been done along the years, passing from Weyl, Serre, Grothendieck, and finally solved by
Deligne \cite{DeligneI}, \cite{DeligneII}, \cite{DeligneIII}, who proved the existence of a {\it mixed Hodge structure} on the $H^k(X,\mathbb Q)$ when $X$ is an arbitrary algebraic complex variety. This means that, given $X$ of complex dimension $m$ and 
$0\leq k\leq 2m$, it is always possible to find
two filtrations:
\begin{itemize}
 \item a decreasing filtration $F^*$, called the {\it Hodge filtration}, such that
 \begin{align}
 H^k(X,\mathbb Q) =F^0 \supseteq F^1 \supseteq \cdots \supseteq F^m \supseteq F^{m+1}=\{0\}; 
 \end{align}
 \item an increasing filtration $W_*$, called the {\it Weight filtration}, such that
 \begin{align}
 \{0\} =W_{-1}\subseteq W_0 \subseteq \cdots \subseteq W_{2k-1} \subseteq W_{2k}=H^k(X,\mathbb Q);
 \end{align}
\end{itemize}
which have the following property. On the complexification of each graded part of $W_*$, $F^*$ induces a decreasing filtration that endows the pieces 
\begin{align}
gr_j^W:=W_j\otimes \mathbb C/ W_{j-1} \otimes \mathbb C 
\end{align}
with a pure Hodge structure of weight $j$. We do not intend to delve further into these concepts, and demand the reader to \cite{DeligneI} for details.\\
At this point, in order to define intersection theory, one has essentially to tackle the problem on how to define intersections of cycles that may contain singularities or pass through a singular region. To get an idea, let us follow the nice exposition in 
\cite{maxim2019intersection}, Chp. 1. Let us consider a two dimensional pinched torus $W$, which looks something like a w\"urstel bent over and glued at the tips just in one point $p$. It is clear that on this surface, a loop around the hole is a non trivial generator for $H_1(W)$
(but pass through the singularity), while a loop $L$ wounding the ``w\"urstel'' shrinks down to a point when moved to the gluing point $p$ (the singularity), Fig. \ref{Fig4}. 
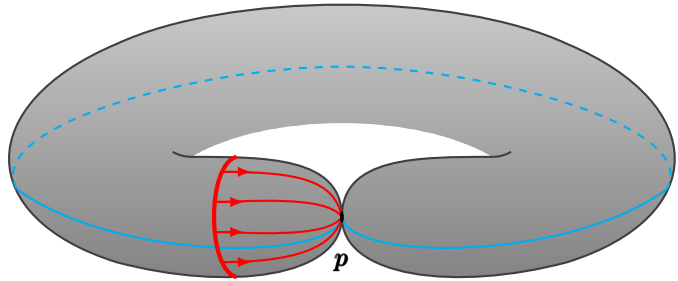
\begin{figure}[!htbp]
\begin{center}
\begin{tikzpicture}[>=latex,decoration={zigzag,amplitude=.5pt,segment length=2pt}]
\shade [top color=white!60!gray, bottom color=gray] (0,0) .. controls (0,0.8) and (-1,0.8) .. (-2,0.8)  .. controls (-1,1.4) and (1,1.4) .. (2,0.8) .. controls (1,0.8) and (0,0.8) .. (0,0) .. controls (0,-0.8) and (1.5,-1) .. (3,-0.6) .. controls (4.9,0) and (4.9,1.5) .. (3,2.3) .. controls (1.5,3) and (-1.5,3) .. (-3,2.3)
.. controls (-4.9,1.5) and (-4.9,0) .. (-3,-0.6) .. controls (-1.5,-1) and (0,-0.8) .. (0,0) -- cycle;
\draw [thick,gray!50!black] (0,0) .. controls (0,-0.8) and (1.5,-1) .. (3,-0.6) .. controls (4.9,0) and (4.9,1.5) .. (3,2.3) .. controls (1.5,3) and (-1.5,3) .. (-3,2.3) .. controls (-4.9,1.5) and (-4.9,0) .. (-3,-0.6) .. controls (-1.5,-1) and (0,-0.8) .. (0,0) -- cycle;
\draw [thick,gray!50!black] (0,0) .. controls (0,0.8) and (-1,0.8) .. (-2,0.8) arc (270:240:0.5);
\draw [thick,gray!50!black] (0,0) .. controls (0,0.8) and (1,0.8) .. (2,0.8) arc (270:300:0.5);
\draw [thick,cyan] (0,0) .. controls (-0.1,-0.6) and (-3,-0.6) .. (-4.32,0.37); 
\draw [thick,cyan] (0,0) .. controls (0.1,-0.6) and (3,-0.6) .. (4.32,0.37); 
\draw [thick,cyan,dashed]  (-4.32,0.37) .. controls (-4.8,1.2) and (-1.9,2) .. (0,2) .. controls (1.9,2) and (4.8,1.2) .. (4.32,0.37); 
\draw [ultra thick,red] (-1.4,0.8) .. controls (-1.8,0.7) and (-1.8,-0.7) .. (-1.4,-0.8); 
\draw [red,thick] (-1.6,0.6) .. controls (-0.6,0.6) and (-0.1,0.5) .. (0,0); \draw [red,thick, ->] (-1.6,0.6) -- (-1.2,0.6);
\draw [red,thick] (-1.7,0.2) .. controls (-1.3,0.2) and (-0.2,0.3) .. (0,0); \draw [red,thick,->] (-1.7,0.2) -- (-1.3,0.2);
\draw [red,thick] (-1.7,-0.2) .. controls (-1.3,-0.2) and (-0.2,-0.3) .. (0,0); \draw [red,thick,->] (-1.7,-0.2) -- (-1.3,-0.2);
\draw [red,thick] (-1.6,-0.6) .. controls (-0.6,-0.6) and (-0.1,-0.4) .. (0,0); \draw [red,thick,->] (-1.6,-0.6) -- (-1.2,-0.6);
\filldraw (0,0) ellipse (0.5pt and 2pt);
\node at (0,-0.6) {$\pmb p$};
\end{tikzpicture}
\caption{A singular torus}
\label{Fig4}
\end{center}
\end{figure}
So, $H^1(W)$ is one dimensional and the pure Hodge structure is lost. A relevant 
observation is that if we think at the singular point $p$ as a zero dimensional representative of $H_0(W)$, it meets the loop $L$ non transversally. This situation repeats for higher dimensional varieties, where singularities can be much worse. A strategy
is to construct more refined (co)homologies but starting from the simplicial ones and restricting the allowed (co)chains. One idea is first to generate a homology starting from a simplicial one where only triangulations giving locally finite chains
 are allowed (so that any point admits a small boundary that intersects only a finite number of cells of the chain). This construction defines the {\it Borel-Moore homology groups} $H^{BM}_k(X,\mathbb Z)$, \cite{maxim2019intersection}. 
 Tensoring with $\mathbb C$ gives the isomorphism
$H_k^{BM}(X,\mathbb C)\simeq H^k_c(X,\mathbb C)$ of the compactly supported cohomology, whose cochains (or closed $k$-forms) have compact support. Assume now that singularities define an increasing filtration of $X$, that is a sequence of subvarieties
\begin{align}
 \emptyset =X_{-1} \subseteq X_0 \subseteq X_1\ldots \subseteq X_{2m-2}=X_{2m-1} \subseteq X_{2m}=X,
\end{align}
where the $X_j$, $j<2m$ are the real $j$-dimensional singular sub loci, called the {\it strata}, of which the largest one has at least real codimension 2, and $X\backslash X_{2m-2}$ is required to be dense in $X$. With this in mind, we can refine further the 
homology by allowing only triangulations giving locally finite $k$-chains whose support intersects all strata transversally. We call $H^{tr}_k(X)$ the obtained homology group. A theorem by McCrory (\cite{maxim2019intersection}, Th.2.3.2, \cite{mccrory1975cone}) then provides the isomorphism 
$H^{tr}_k(X)\simeq H^{2m-k}(X,\mathbb Z)$. Assuming that (in a suitable sense, see \cite{maxim2019intersection}) $X$ is oriented, we can use cap product with the fundamental class $[X]$ to map $H^{2m-k}(X,\mathbb Z)$ into $H_k^{BM}(X,\mathbb Z)\simeq H^{lf}_k(X)$,
the last term being the locally finite simplicial homology. However, since $X$ is singular, the Poincar\'e map is no more an isomorphism. Therefore, $H^{lf}_k(X)$ and $H^{tr}_k(X)$ are not isomorphic, the reason essentially being that, in general, 
the locally finite cycles cannot be deformed so to have transversal intersection with all strata. For example, this is what happened in our example of the ``glued w\"urstel''. Therefore, one needs to do a final refinement, in order to get a correct intersection theory
together with a good Poincar\'e duality. This is obtained by restricting which chains can meet the singular locus, if not transversally. \\
The first important request is that any chain has to meet $X_{2m-2}$ transversally. For algebraic varieties of complex dimension 1, that's all. For example, in the case of our glued w\"urstel this just implies that the loop around the hole is forbidden. 
Thus, $H_1(Z)=0$ and the perverse homology of the singular surface is the same as for the two dimensional sphere. Indeed, the singularity in this case is non normal, in the sense that if we take a small neighbourhood of the singular point $p$ and keep the
point out, the neighbourhood separates in two disconnected components. We can normalize the surface by adding the laking point separately to each one of the two disconnected parts. In other words we unglue the w\"urstel, which becomes an usual one, 
with the topology of $S^2$. These kinds of singularity are thus simply solved by normalization, and for the case of complex dimension two things remain relatively simple.\\
However, for higher dimensional $X$, one has to consider also intersection with higher codimensional singularities. Here is where one introduces a {\it perversity}, \cite{goresky1,goresky2,goresky3,goresky4,goresky5,goresky6,goresky7,goresky8,goresky9}: 

\

\begin{Definition}
 A perversity is a function 
\begin{align}
p:\mathbb Z_{\geq 2}\longrightarrow \mathbb N, 
\end{align}
satisfying the conditions 
\begin{align}
 p(2)=0, \qquad p(j)\leq p(j+1) \leq p(j)+1, \quad j\geq 2.
\end{align}
\end{Definition}

\

In particular, $p_t(j)=j-2$ is called the top perversity, while $p_0(j)=0$ is the zero perversity. Also, $p_{lm}(j)=\lfloor (j-2)/2 \rfloor$, $j>2$, is called the lower-middle perversity, and $p_{um}=p_t-p_{lm}$ is the upper-middle perversity.\\
Given a perversity $p$, an allowed $k$-chain $\sigma$ is a locally finite chain whose support $|\sigma|$ satisfies
\begin{align}
 {\rm dim} (|\sigma| \cap X_{2m-j})&\geq k-j+p(j), \\
 {\rm dim} (|\partial\sigma| \cap X_{2m-j})&\geq k-j+p(j)-1,
\end{align}
for all $j\geq 2$. The so generated homology groups are called the {\it BM intersection homology groups} with perversity $p$ of $X$, $I^{BM}H_k^{(p)} (X,\mathbb Z)$. If in place of locally finite chains one starts with finite chains, then 
one obtains the {\it homology groups} with perversity $p$, $IH_k^{(p)} (X,\mathbb Z)$. \\
With these notions, one has the following version of the Poincar\'e duality (\cite{maxim2019intersection}, Th.2.6.1):

\

\begin{Theorem}[Poincar\'e duality]
 If $X$ is an algebraic manifold of complex dimension $m$, whose singularities define a filtration as above, and $p$ and $\bar p$ are two complementary perversities, which means $p+\bar p=p_t$, then
there is a a non-degenerate bilinear pairing
\begin{align}
\frown : IH_k^{(p)}(X,\mathbb Z) \times I^{BM}H_{2m-k}^{(\bar p)}(X,\mathbb Z) \longrightarrow \mathbb Q.
\end{align}
\end{Theorem}

\

This theorem provides the necessary definition of ($\mathbb Q$-valued) intersection number. A construction that proves homological invariance and independence of choices, as well as calculability, passes through sheaf theory and its generalizations.
Analyzing such a formulation goes too far beyond the aim of our discussion. However, we want just to add some remarks, for who has at least a basic knowledge on sheaf theory (a nice introduction is in \cite{maxim2019intersection}, Chp. 4). In the sheaf theoretical formulation,
a particular role is played by Local Systems, which are locally constant sheaves, that is, any sheaf $\mathcal L$ over $X$, such that its local restriction to a small neighbourhood $U_p$ of any given $p\in X$, $\mathcal L_p\simeq M_p$, a given $A$-module. 
Thus, it comes out to be necessary to work with {\it twisted homologies}, which are homologies taking value in a Local System. This leads to work with complexes of sheaves. A main difficulty in doing these, is that in such framework often diagrams 
commute only up to homotopies and not identically. To overcome this problem, one introduces the category $K(X)$, that is essentially the same of the one of complex of sheaves, up to identifications by homotopy. The problem is that one of the main
computational instruments, t.i. short exact sequences, are lost in passing from the category $C(X)$ of complexes of sheaves to $K(X)$. A remedy consists in finding a way to map short exact sequences in $C(X)$ into {\it distinguished triangles}, from which it is
also possible to construct long exact sequences. This leads to the language of derived categories (\cite{maxim2019intersection}, Chp. 4). This way one arrives to construct the intersection homology groups $IH^k(X)$.\\
In place of embarking into the formidably difficult program of analyzing such constructions further, we pass to illustrate a possible strategy to be applied to the more specific situation of computing integrals of the form (\ref{baikov}), suggested by 
M. Kontsevich in a private communication to P. Mastrolia, T. Damour and S.L. Cacciatori, at IHES on January 2020.

\subsection{From Feynman integrals to intersection theory.}\label{FeynmanToInters}
Let us consider (\ref{baikov}). In general, we assume that $\nu_j=0$ for sake of simplicity, and $\gamma\in \mathbb C\backslash \mathbb Z$.\footnote{Indeed $\gamma=N/2+\epsilon$ because of dimensional regularization.} Consider the map
\begin{align}
 B:\mathbb C_*^M \longrightarrow \mathbb C \supset \mathbb C_*.
\end{align}
This map has critical regions with critical values $t_\alpha\in \mathbb C$, $\alpha=1,\ldots,n$. For each $t$, we can consider the algebraic variety
\begin{align}
X_t=B^{-1}(t) \cap \mathbb C_*^M, 
\end{align}
with $t\in \tilde {\mathbb C} \equiv\mathbb C^*-\{t_\alpha\}_{\alpha=1}^n$. For these $t$, $X_t$ is smooth and admits a compactification $\bar X_t$, with $\bar X_t$ projective, $X_t\subset \bar X_t$ open embedding, and $\bar X_t\backslash X_t$ is a normal 
crossing divisor, which means that at any point of $X_t$ it is described by local equations $\prod_{j=1}^r z_j$. Thus, starting from $H^{M-1}(X_t, \mathbb Z)$, we get a (Deligne) weight filtration, starting from $\bar X_t$ (with $H^{M-1}(\bar X_t)$ of lowest 
weight $0$). When $t$ moves in $\tilde {\mathbb C}$, then we get a variation of mixed Hodge structures on $H^{M-1}(X_t, \mathbb Z)$.\\
Now, these variations of mixed Hodge structures define a perverse sheaf over $\tilde {\mathbb C}$. Luckily, on complex curves they can be understood quite easily, we refer to \cite{Williamson} for a nice presentation. A perverse sheaf is not exactly a sheaf, but
has properties similar to a sheaf, as it can be glued by local data. Also, away from a singular point, it is effectively a sheaf, more precisely a Local System, so on $x$ non singular, we can think the relative stalks as a finitely generated module $V$. Around a 
singular point, one has to add a monodromy rule around the singular point $t_\alpha$:
\begin{align}
 T_\alpha : V\rightarrow V, \qquad \alpha=1,\ldots,n.
\end{align}
This is the effect of our variation of mixed Hodge monodromy. Let $V^{T_\alpha}\subseteq V$ the submodule of invariant elements (eigenelements). The exact sequence $V^{T_\alpha} \hookrightarrow V\rightarrow V_{T_\alpha}$ defines the coinvariant
module $V_{T_\alpha}$. Since the monodromy action depends only on the first homotopy group, we can see monodromies as representing the action of the first homotopy group $\pi_1(\tilde {\mathbb C})$. On $V_{T_\alpha}$, we get an irreducible
representation of $\pi_1(\mathbb C^*-\{t_\alpha\}_{\alpha=1}^n)$ on a variation of a pure Hodge structure. Thus, it is endowed with a polarization via the Hodge-Riemann bilinear relations. This polarization is symmetric or skew-symmetric, according to the
parity of $M-1$, so we have representations
\begin{align}
 \pi_1(\tilde {\mathbb C})\longrightarrow G,
\end{align}
with $G$ $\mathbb Z$-orthogonal or $\mathbb Z$-symplectic, according to the given parity. Any such irreps determines univocally a local system. 
Notice that from the Primitive Lefschetz decomposition one can read out the precise structure of the intersection product induced by the polarization.  \\
At this point, one can also determine a basis for the homology underlying the original integral, using this perverse extension. We can get a twisted sheaf over $C_*$, given by the Local System $\mathcal L$ twisted by the perverse monodromies, 
associated to $t^\gamma$. The original integral can be thought to be defined on a middle dimensional cycles in a complex $m$-dimensional variety $Y$. 
The cohomology $H^1(\mathbb C_*, \mathcal L)$ determines a basis of real $m$-dimensional {\it Lefschetz thimbles} in $Y$, Fig. \ref{Fig5}, which indeed represent the perverse cohomology structure. 
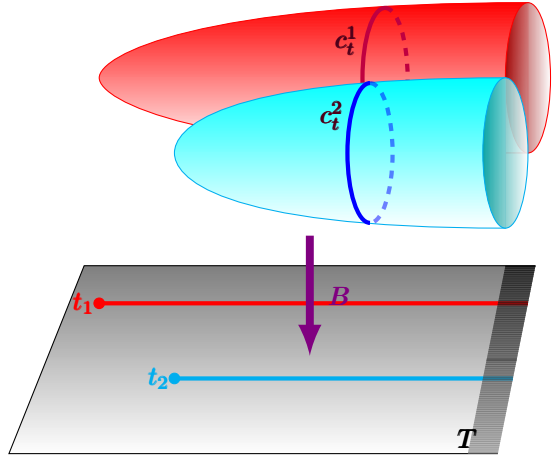
\begin{figure}[!htbp]
\begin{center}
\begin{tikzpicture}[>=latex,decoration={zigzag,amplitude=.5pt,segment length=2pt}]
\draw [top color=gray] (3,0) -- (-3,0) -- (-4,-2.5) -- (2.5,-2.5);
\draw [red,top color=red] (3-0.5*0.5/2.5,3.5) .. controls (-1.5,3.5) and (-2.8,2.8) .. (-2.8,2.5) .. controls (-2.8,5-2.8) and (-1.5,5-3.5) .. (3-0.5*0.5/2.5,5-3.5) -- cycle;
\draw [red,left color=red!60!black] (3-0.5*0.5/2.5,2.5) ellipse (0.3 and 1);
\draw [purple,ultra thick] (1,2.5+0.93) arc (90:270:0.3 and 0.93);
\draw [purple,ultra thick,dashed] (1,2.5+0.93) arc (90:-90:0.3 and 0.93);
\draw [cyan,top color=cyan] (2.5+0.5*0.5/2.5,2.5) .. controls (-1.3,2.5) and (-1.8,1.8) .. (-1.8,1.5) .. controls (-1.8,3-1.8) and (-1.3,3-2.5) .. (2.5+0.5*0.5/2.5,3-2.5) -- cycle;
\draw [cyan,left color=cyan!60!black] (2.5+0.5*0.5/2.5,1.5) ellipse (0.3 and 1);
\draw [blue,ultra thick] (1-0.5/2.5,1.5+0.93) arc (90:270:0.3 and 0.93);
\draw [blue,ultra thick,dashed,opacity=0.5] (1-0.5/2.5,1.5+0.93) arc (90:-90:0.3 and 0.93);
\draw [red,ultra thick] (-2.8,-0.5) -- (3-0.5*0.5/2.5,-0.5);
\draw [cyan,ultra thick] (-1.8,-1.5) -- (2.5+0.5/2.5,-1.5);
\draw [violet,->, line width=3pt] (0,0.4) -- (0,-1.2);
\filldraw [red] (-2.8,-0.5) circle (2pt);
\filldraw [cyan] (-1.8,-1.5) circle (2pt);
\node [violet] at (0.4,-0.4) {$\pmb {B}$};
\node [red] at (-3,-0.5) {$\pmb {t_1}$};
\node [cyan] at (-2,-1.5) {$\pmb {t_2}$};
\node [purple!40!black] at (0.5,3) {$\pmb {c^1_t}$};
\node [purple!40!black] at (0.5-0.5/2.5,2) {$\pmb {c^2_t}$};
\shade [top color= black, bottom color=gray,opacity=0.5] (3,0) -- (2.6,0) -- (2.1,-2.5) -- (2.5,-2.5) -- cycle;
\node at (2.1,-2.3) {$\pmb {T}$};
\end{tikzpicture}
\caption{Thimbles drawn above the plane of $t=-\gamma \log B$.}
\label{Fig5}
\end{center}
\end{figure}
On $\mathbb C$, each critical value $t_\alpha$ of $B$ determines a phase of $t_\alpha^\gamma$. Draw $n$ lines with such constant phases, passing through the $n$ critical values
(so $\gamma \log B$ has constant imaginary part) starting from the critical values of $B$ to $\infty$ or to zero, according Re$(\gamma)<0$, or Re$(\gamma>0)$. To each line it corresponds a family of $(m-1)$-dimensional cycles $c^\alpha_t$, collapsing to a 
point in $t=t_\alpha$ (a collapsing cycle), and going towards a region where $B^{-\gamma}$ assumes larger and larger values. Such cycles are determined by the coinvariant elements in a way similar to the construction we gave in section \ref{Perthimbles}.
Fix a large real value $T\gg {\rm Re}(-\gamma \log t_\alpha)$, for all $\alpha$; then $-\gamma \log t>T$ corresponds to the region $Z\subset Y$ where $B^{-\gamma}|_{Z}>T$. The unions 
\begin{align}
\tau^\alpha =\bigcup_{t | {\scriptsize
\begin{pmatrix}
{\rm Im}(-\gamma \log \frac t{t_\alpha})=0,\\ {\rm Re}(-\gamma \log \frac t{t_\alpha})> 0 
\end{pmatrix}
}} c^\alpha_t,
\end{align}
are called Lefschetz thimbles, and result to be a basis for the middle homology of $Y$ relative to $Z$. We remark that such constructions are quite familiar to physicists even if in a different form. Let us define the potential function
\begin{align}
 W=-\hbar \gamma \log B
\end{align}
so that the integrand becomes
\begin{align}
 e^\frac W\hbar.
\end{align}
One can evaluate the integral for very small $\hbar$ by using the stationary phase method. So, such integral is dominated by the stationary points, where 
\begin{align}
 dW=0=\frac {dB}{B}.
\end{align}
These are the critical points of the polynomial. The valuation of the integral is then dominated by the paths that make the integrand as less oscillating as possible. These are the union of paths such that the phase of the integrand is constant, and that goes to
infinity with Re$(W)>0$, so that the integrand goes to zero. This means ${\rm Re}(-\gamma \log \frac t{t_\alpha})> 0$, and we get the thimbles.

\

Here, we end up our mathematical excursus and go back to the practical constructions of Cho, Matsumoto and others.


\section{Computing Intersection Numbers: state of art and open problems}\label{secinternumb}
In this section we mainly address the practical computation of the coefficients of the Master Integrals, the so-called \textsl{intersection numbers}.\\
Intersection numbers between two $n$-forms $\bra{\phi_L}$ and $\ket{\phi_R}$, which have been introduced abstractly in Section \ref{veryhardmaths}, are defined by Cho and Matsumoto \cite{cho1995} as 
\begin{equation}
\bra{\phi_L}\ket{\phi_R}_\omega:=\frac{1}{(2\pi i)^n}\int_X \iota_\omega(\phi_L)\wedge \phi_R\,,\label{interdef}
\end{equation}
where $X$ is the whole complex space $\mathbb{C}^n$ deprived of the hyperplanes corresponding to the poles of $\omega$ and $\iota_\omega$ is a map which sends $\phi_L$ to an equivalent form with compact support. In our particular case, $\bra{\phi_L}$ and $\ket{\phi_R}$ are \textsl{twisted} cocycles, which implies that the intersection number in Eq.~\eqref{interdef} is not an integer in general \cite{aomoto2011theory}.\\
Notice that if we omitted $\iota_\omega$ in Eq.~\eqref{interdef}, then the intersection number would vanish, as both $\phi_L$ and $\phi_R$ are holomorphic in the domain of integration $X$. To understand why this is true, we focus on a 1D description (there is only a single variable $z$) in order to have a simple notation, but there is no difference in considering a multidimensional example. We consider $z$ to be complex, hence $X$ is $\mathbb{C}$ without some points corresponding to poles. A certain function $f(x,y)$ can always be split into its real and imaginary part: $f(x,y)=u+iv$. Introducing a change of variables
\begin{equation}
\begin{cases}
z=x+iy\\
\bar{z}=x-iy
\end{cases}\label{change}
\end{equation}
allows to obtain a certain $\tilde{f}(z,\bar{z})$. If $f$ is holomorphic, in order to satisfy the Cauchy-Riemann conditions, after the change of variables $\tilde{f}$ depends only on $z$: $\tilde{f}(z,\bar{z})=f(z)$.
In general a form $\phi$ can be decomposed as \begin{equation}
\phi=\hat{\phi}_1dz+\hat{\phi}_2d\bar{z}\,.\label{formgen}
\end{equation}
When we wedge two forms as in Eq.~\eqref{formgen}, only the mixed terms $dz\wedge d\bar{z}$ survive; but if the forms involved are holomorphic, the wedge product is null.
\

Strictly focusing on the twisted cohomology framework, we present in this section a summary on what has been obtained for the computation of intersection numbers so far; in Section \ref{newideas}, we give a review of still open problems plus some ideas on how to tackle them.

\subsection{Univariate case}\label{logforms}

In this section we show how the intersection number can be evaluated exactly in the case when integrals are defined over only one complex variable. We will refer to the simple derivation that was presented in \cite{matsumoto1998}.
Let us consider a $1$-form $\phi_L$ having poles at some points $z_i \in \mathbb{C}$. In order to compute Eq.~\eqref{interdef},  $\iota_\omega(\phi_L)$ must be constructed explicitly. The key point lies in defining circular regions around each $z_i$ point: to fix  ideas, we call $V_i$ and $U_i$ two discs centered in $z_i$ such that $V_i\subset U_i$; these discs are defined such that $U_i\cap U_j=\emptyset$ for $i\neq j$. We then introduce for each $i$:
\begin{itemize}
    \item[1:] a holomorphic function $\psi_i$ such that
    \begin{equation}
        \nabla_\omega\psi_i=\phi_L\,\,\,\text{on}\,\,\,U_i\setminus \left\lbrace z_i\right\rbrace;\label{psi_i}
    \end{equation}
    \item[2:] a function $h_i$ such that
    \begin{equation}
h_i=
\begin{cases}
1\,\,\,\text{on}\,\,\,V_i\,;\\
0\leq h_i \leq 1\,\,\,\text{smooth interpolation on}\,\,\,U_i\setminus V_i\,;\\
0\,\,\,\text{out of}\,\,\,U_i\,.\label{h_i}
\end{cases}
\end{equation}
\end{itemize}
Then $\iota_\omega(\phi_L)$ can be written as
\begin{equation}
\iota_\omega(\phi_L)=\phi_L-\sum_i \nabla_\omega(h_i\psi_i)=\phi_L-\sum_i(dh_i\psi_i+h_i\nabla_\omega\psi_i)\,.\label{compact}
\end{equation}
Notice that in Eq.~\eqref{compact}, as $h_i=0$ out of $U_i$, we are actually not modifying $\phi_L$ in that region. Inside $V_i$, $h_i=1$, meaning that the whole $\phi_L$ is subtracted and $\iota_\omega(\phi_L)=0$ in the innermost region around the singular points; on the other hand, in the outer ring $U_i\setminus V_i$, $\phi_L$ is subtracted smoothly. Finally, $-dh_i\psi_i$ is an extra term existing only in the $U_i\setminus V_i$ ring. Notice that $\phi_L$ and $\iota_\omega(\phi_L)$ lie in the same cohomology class, as they are identical up to a covariant derivative of a function.\\
What is left now is to find the explicit form of $\psi_i$ obeying Eq.~\eqref{psi_i}.
\begin{Lemma}\label{Lemmapsi} Unique existence of $\psi_i$.\\
$\exists !\,\, \psi_i$ such that $\psi_i$  is holomorphic on $U_i\setminus\left\lbrace z_i \right\rbrace$ and $\nabla_\omega\psi_i=\phi_L$ on $U_i\setminus \left\lbrace z_i\right\rbrace$. 
\end{Lemma}
\begin{proof}
Consider a 1-form $\phi_L$ having a pole of order $N$ at $z=z_i$; $\omega$ is a 1-form sharing a pole of $\phi_L$, $\psi_i$ is required to be holomorphic: hence in terms of the local coordinate variable $z$ near $z_i$ (up to a change of coordinates, we can consider $z_i=0$ without loss of generality) it is possible to write
\begin{equation}
\phi_L=\sum_{m=-N}^\infty b_mz^m dz\,,\,\,\,\,\,\,\,\,\,\,\,\,\,\,\,\,\,\,\,\,\,\omega=\sum_{q=-1}^\infty a_qz^qdz\,,\,\,\,\,\,\,\,\,\,\,\,\,\,\,\,\,\,\,\,\,\,\psi_i     =\sum_{m=-N+1}^\infty c_mz^m.\label{localmente}
\end{equation}
Notice how in \eqref{localmente} $b_{-1}=\text{Res}_{z=z_i}(\phi_L)$ and $a_{-1}=\text{Res}_{z=z_i}(\omega)=\alpha_i$. Then $\nabla_\omega\psi_i$ becomes
\begin{align}
(d+\omega\wedge)\sum_{m=-N+1}^\infty c_mz^m=d\left(\sum_{m=-N+1}^\infty c_mz^m\right)+\sum_{m=-N+1}^\infty\sum_{q=-1}^\infty a_qc_mz^{m+q}dz=\cr=\sum_{n=-N+1}^\infty\left((n+1)c_{n+1}+\sum_{q=-1}^{n}a_qc_{n-q}\right)z^{n}dz\,,
\end{align}
in which we defined $n=m+q$: hence the sum $\sum_{m=-N+1}^\infty\sum_{q=-1}^\infty$ must satisfy the condition $q=n-m$, whose bigger value is obtained when $m$ is as small as possible (which is $m=-N+1$). This implies that the sum over $q$ goes up to $n+N-1$. We are thus lead to the identification
\begin{equation}
(n+1)c_{n+1}+\sum_{q=-1}^{n+N-1}a_q c_{n-q} =b_{n} \xRightarrow{n=-1} \sum_{q=-1}^{N-2} a_{q}c_{-q-1} = b_{-1} .\nonumber\qedhere
\end{equation}
\end{proof}

The intersection number \eqref{interdef} can be rewritten as
\begin{equation}
\bra{\phi_L}\ket{\phi_R}_\omega=\frac{1}{2\pi i}\int_X\left[\phi_L-\sum_{i}(dh_i)\psi_i-\sum_ih_i\nabla_\omega\psi_i\right]\wedge\phi_R=-\frac{1}{2\pi i}\sum_i\int_{U_i\setminus V_i} dh_i\psi_i\wedge \phi_R,\label{calcoli}
\end{equation}
where the second equality can be obtained by recalling that the first term vanishes because $\phi_L\wedge\phi_R=0$; the second term survives where $dh_i\neq0$, i.e. in $U_i\setminus V_i$; the third term is again proportional to $\phi_L\wedge\phi_R$ and vanishes. Notice that it holds:
\begin{equation}
dh_i\psi_i\wedge \phi_R=d(h_i\psi_i\phi_R).
\end{equation}
This is because the extra terms $h_id\psi_i\wedge\phi_R$ and $h_i\psi_id\phi_R$ vanish, as $d\psi_i$ and $\phi_R$ are both holomorphic and as $\phi_R$ is a closed form. By the Stokes theorem, it is possible to rewrite Eq.~\eqref{calcoli} as
\begin{align}
\bra{\phi_L}\ket{\phi_R}_\omega=-\frac{1}{2\pi i}\sum_i\int_{U_i\setminus V_i}d(h_i\psi_i\phi_R)=-\frac{1}{2\pi i}\sum_i\int_{\partial (U_i\setminus V_i)}h_i\psi_i\phi_R=\frac{1}{2\pi i}\sum_i\int_{\partial V_i}\psi_i\phi_R\,.\label{eqclosedpath}
\end{align}
Because $\partial V_i$ is a closed path, we can always rewrite Eq.~\eqref{eqclosedpath} as a sum of residues:
\begin{equation}\label{int_num_univar}
 \bra{\phi_L}\ket{\phi_R}_\omega=\sum_{i=1}^k\text{Res}_{z=z_i}(\psi_i \phi_R)\,.
\end{equation}

\subsection{Logarithmic $n$-forms}

Evaluation of intersection numbers is particularly simple for logarithmic forms. In the univariate case, one can check from Lemma \ref{Lemmapsi} that if $\phi_L$ has a simple pole, than the function $\psi_i$ takes the form
\begin{equation}
\psi_i=\frac{\text{Res}_{z=z_i}(\phi_L)}{\text{Res}_{z=z_i}(\omega)}+\mathcal{O}(z-z_i)=\frac{\text{Res}_{z=z_i}(\phi_L)}{\alpha_i}+\mathcal{O}(z-z_i)\,.
\end{equation}
In this case, Eq. \eqref{int_num_univar} leads to the following formula for the evaluation of the intersection number:
\begin{equation}
\bra{\phi_L}\ket{\phi_R}_\omega=\sum_i\frac{\text{Res}_{z=z_i}\phi_L\text{Res}_{z=z_i}\phi_R}{\alpha_i}.\label{intersimple}
\end{equation}

Eq.~\eqref{intersimple} is valid for logarithmic $1$-forms, but it can be generalized to logarithmic $n$-forms. We give only the final result for this case. Let us consider the twisted one-form $\omega$ in the following form: \begin{equation}
\omega=\sum_{i=0}^{s+1}\alpha_i d\log f_i\,,
\end{equation}
where the $f_i$ are linear functions in the given variables (hence each $f_i=0$ defines a certain hyperplane where the twist $\omega$ shows a logarithmic singularity), while $\sum_i\alpha_i=0$. The cocycles $\phi_{L,R}$ assume instead the form
\begin{equation}
\phi_{I}=d\log\left(\frac{f_{i_0}}{f_{i_1}}\right)\wedge\cdots\wedge d\log\left(\frac{f_{i_{n-1}}}{f_{i_n}}\right)\label{philog}
\end{equation}
with $I=(i_o,\cdots,i_n)$; the indexes are ordered such that $0\leq i_0<\cdots<i_n\leq s+1$. In other words, $\phi_{I}$ belongs to the $n$-th twisted cohomology group  $H^n(X,\nabla_\omega)$, where $X=\mathbb{CP}^n\setminus\left\lbrace\cup_i{f_i=0}\right\rbrace$.\\
With these hypotheses, a recent work \cite{Mizera:2017rqa} has shown that the intersection number between two logarithmic $n$-forms can also be written as
\begin{equation}
\bra{\phi_L}\ket{\phi_R}_\omega=\frac{1}{(-2\pi i)^n}\oint_{\wedge_{a=1}^n\left\lbrace\right|\omega_a|=\epsilon\rbrace}\frac{\phi_L\hat{\phi}_R}{\prod_{a=1}^n\omega_a}, \label{usoquestaformula}
\end{equation}
where $\hat{\phi}$ is the component of $\phi=\hat{\phi}\,dz_1\wedge \cdots \wedge dz_n$ , $a$ enumerates the dimensions involved (the number of variables) and $\omega_a$ represent the components of the twist $\omega$ along each $dz_a$.

\subsection{Multivariate case: recursive method}\label{recursive}
The main aim now is to be able to compute any intersection number, independently of the order of the poles appearing and the number of dimensions involved. A first attempt follows a recursive approach in the number of variables (we refer to \cite{Weinzierl:2020xyy,Frellesvig:2020qot}). In order to compute an intersection number in $z_1,\cdots,z_n$ variables, the problem can be divided into two steps: calculating the intersection number in $z_1,\cdots,z_{n-1}$ variables and then a generalized intersection number depending only on $z_n$, which will be introduced later. In the 1D case, $\phi_{L,R}$ and $\omega$ are 1-forms, which we can regard as some sort of \textsl{scalars}. Taking into account the existence of $\nu$ independent equivalence classes, it is possible to build \textsl{vectorial} ($\phi_{L,R,j})$ and \textsl{tensorial} ($\Omega_{ij}$) objects, with ``gauge'' transformations written as
\begin{equation}
\hat{\phi}_{L,j}^\prime=\hat{\phi}_{L,j}+\partial_z\xi_j+\xi_i\Omega_{ij}\,,\,\,\,\,\,\,\,\,\,\,\,\,\,\,\,\,\,\,\,\,\hat{\phi}_{R,j}^\prime=\hat{\phi}_{R,j}+\partial_z\xi_j-\Omega_{ji}\xi_i\,.
\end{equation}
The next key idea is to study, for each $i$, the twisted cohomology group associated to the fibration
\begin{equation}
\nabla_\omega=d+\sum_{j=1}^n\omega_jdz_j=\sum_{j=1}^idz_j\left(\frac{\partial}{\partial z_j}+\omega_j\right)+\sum_{j=i+1}^ndz_j\left(\frac{\partial}{\partial z_j}+\omega_j\right)\label{fibration} \,,
\end{equation}
and define
\begin{equation}
    \omega^{(i)}:=\sum_{j=1}^i\omega_jdz_j \,.
\end{equation}

In Eq.~\eqref{fibration} $z_{i+1},\cdots,z_n$ are treated as fixed parameters. Starting from $i=n-1$, the goal is to express a cohomology class $\bra{\phi_L^{(n)}}\in H_\omega^n$ using a basis of $H_\omega^{n-1}$, which is supposed to be known: such basis is formed by $\nu_{n-1}$ elements $\bra{e_j^{(n-1)}}$, with dual basis $\ket{d_j^{(n-1)}}$ such that $\bra{e_j^{(n-1)}}\ket{d_k^{(n-1)}}=\delta_{jk}$. Inserting an identity of the form $I=\sum_{j=1}^{\nu_{n-1}}\ket{d_j^{(n-1)}}\bra{e_j^{(n-1)}}$ leads to:
\begin{equation}
\bra{\phi_L^{(n)}}=\sum_{j=1}^{\nu_{n-1}}\bra{\phi_{L,j}^{(n)}}\wedge\bra{e_j^{(n-1)}},\qquad \ket{\phi_R^{(n)}}=\sum_{j=1}^{\nu_{n-1}}\ket{d_j^{(n-1)}}\wedge\ket{\phi_{R,j}^{(n)}},\label{nton-1}
\end{equation}
where
\begin{equation}
\bra{\phi_{L,j}^{(n)}}=\bra{\phi_L^{(n)}}\ket{d_j^{(n-1)}},\qquad\ket{\phi_{R,j}^{(n)}}=\bra{e_j^{(n-1)}}\ket{\phi_R^{(n)}}.\label{classcoeff}
\end{equation}
Notice how the coefficients defined in Eq.~\eqref{classcoeff} are not unique objects but classes themselves. For example, the coefficient $\bra{\phi_{L,j}^{(n)}}$ is invariant with respect to the choice within the cohomology class $\ket{d_j^{(n-1)}}$ of $H_\omega^{*\,\,n-1}=H_{-\omega}^{n-1}$: a change $d_j^{n-1}\to d_j^{n-1}+\nabla_{-\omega^{(n-1)}}\psi$ with $\psi$ being a 2-form would make no difference. This equivalence class is smaller than $\bra{\phi_L^{(n)}}$, which is composed of all the transformations $\phi_L^{(n)}\to \phi_L^{(n)}+\nabla_\omega\xi$ with $\xi$ being a generic$(n-1)$-form: thus, the coefficients in Eq.~\eqref{classcoeff} represent new equivalence classes overall.\\
We introduce a new matrix $\Omega^{(n)}$ with dimension $\nu_{n-1}\times\nu_{n-1}$
\begin{equation}
\Omega_{ij}^{(n)}:=\bra{(\partial_{z_n}+\omega_n)e_i^{(n-1)}}\ket{d_j^{(n-1)}},\label{matriciozzo}
\end{equation}
which implies
\begin{equation}
    \sum_i\Omega_{ij}^{(n)}\bra{e_j^{(n-1)}}=\bra{(\partial_{z_n}+\omega_n)e_i^{(n-1)}}\,.
\end{equation}
The matrix in Eq.~\eqref{matriciozzo} allows to write the gauge transformations of the coefficients in the following form (the derivation can be found in \cite{Weinzierl:2020xyy}):
\begin{equation}
\hat{\phi}_{L,j}^{(n)}\to\hat{\phi}_{L,j}^{(n)}+g_i\left(\cev{\partial_{z_n}}\delta_{ij}+\Omega_{ij}^{(n)}\right)\,,\,\,\,\,\,\,\,\,\,\,\,\,\,\,\,\,\,\,\,\,\hat{\phi}_{R,j}^{(n)}\to\hat{\phi}_{R,j}^{(n)}+\left(\delta_{ij}\veclungo{\partial_{z_n}}-\Omega_{ji}^{(n)}\right)g_i,\label{gaugecoeff}
\end{equation}
where the $g_i$ are arbitrary functions. With a shorter notation Eq.~\eqref{gaugecoeff} can be rewritten as 
\begin{equation}
\hat{\phi}_j\to \hat{\phi}_j+(\delta_{jk}\partial_{z_n}+\Omega_{jk})g_k\, ,\label{gaugecoeffshort}
\end{equation}
where the $L$ and $R$ cases are restored taking $\Omega\equiv(\Omega^{(n)})^{\intercal}$ and $\Omega\equiv-\Omega^{(n)}$ respectively. In a number of practical cases, the matrix $\Omega^{(n)}$ can often be chosen to have only simple poles by the application of the \textsl{Moser algorithm}, whose details we do not present in this work but can be found in \cite{Moser1959/60,Lee:2014ioa}. 

\

This property, along with Eq.~\eqref{gaugecoeffshort}, allows to find coefficients $\hat{\phi}_j$ which are more convenient than the original ones, namely which show only simple poles in $z_n$. We start looking for the transformation which lowers the order of a pole appearing in $\hat{\phi}_j$. We outline the procedure when a pole is located at a finite point.
To fix ideas, we call $q$ the irreducible polynomial appearing at the denominator of $\hat{\phi}_j$, raised to a certain power $o$. The correct transformation is given by the choice
\begin{equation}
g_j(z_n)=\frac{1}{q^{o-1}}\sum_{k=0}^{\text{deg}(q)-1}c_{jk}z_n^k\label{f_j}\, ,
\end{equation}
where the coefficients $c_{jk}$ have yet to be determined. With the introduction of Eq.~\eqref{f_j}, Eq.~\eqref{gaugecoeffshort} assumes the form
\begin{align}
\hat{\phi}_j&\to\hat{\phi}_j+\delta_{jk}\left[\partial_{z_n}\left(\frac{1}{q^{o-1}}\right)\sum_{s=0}^{\text{deg}(q)-1}c_{ks}z_n^s+\frac{1}{q^{o-1}}\sum_{s=0}^{\text{deg}(q)-1}c_{ks}\partial_{z_n}z_n^s\right]+{\frac{1}{q^{o-1}}\sum_{s=0}^{\text{deg}(q)-1}\Omega_{jk}c_{ks}z_n^s}\nonumber\\
&=\underbrace{\hat{\phi}_j-(o-1)\frac{\partial_{z_n}q}{q^o}\sum_{s=0}^{\text{deg}(q)-1}c_{js}z_n^s}_{1}+\underbrace{\frac{1}{q^{o-1}}\sum_{s=0}^{\text{deg}(q)-1}sc_{js}z_n^{s-1}}_{2}+\underbrace{\frac{1}{q^{o-1}}\sum_{s=0}^{\text{deg}(q)-1}\Omega_{jk}c_{ks}z_n^s}_3\,.\label{123}
\end{align}
Notice how the term labeled $1$ in Eq.~\eqref{123} is the one involving the power $q^0$ in the denominator: these terms (which are of the type $z_n^k/q^o$ with $k\in\left\lbrace0,\cdots,\text{deg}(q)-1\right\rbrace$) can be eliminated by solving $\text{deg}(q)$ equations. The solution of this system determines $c_{jk}$ for a fixed $j$. Considering all the $j$ values, all the appropriate $c_{jk}$ coefficients are determined by solving $\nu_{n-1}\cdot\text{deg}(q)$ equations. A similar procedure exists when the pole is not located at a finite point but at infinity.

\

Notice that, thanks to Eq.~\eqref{nton-1}, it is possible to write the intersection number as 
\begin{equation}
\bra{\phi_L}\ket{\phi_R}_\omega=\sum_{z_0\in S_n}\sum_{j=1}^{\nu_{n-1}}\text{Res}_{z_n=z_0}\left(\hat{\psi}_{L,j}^{(n)}\bra{e_j^{(n-1)}}\ket{\phi_R}\right)\, ,\label{almostthere}
\end{equation}
where $S_n$ is the set comprehensive of all the singular points of $\Omega^{(n)}$ in $z_n$, while $\hat{\psi}_{L,j}^{(n)}$ is a function which locally solves
\begin{equation}
\partial_{z_n} \hat{\psi}_{L,j}^{(n)}+\hat{\psi}_{L,i}^{(n)}\Omega_{ij}^{(n)}=\hat{\phi}_{L,j}^{(n)}.\label{altrapsi_i}
\end{equation}
Notice how this situation closely resembles the one outlined by Eq.~\eqref{intersimple} and Eq.~\eqref{psi_i}. Similarly to the result given by Lemma \ref{Lemmapsi}, the solution of Eq.~\eqref{altrapsi_i} is given by
\begin{equation}
\hat{\psi}_{L,j}^{(n)}=\sum_i\text{Res}_{z_n=z_0}\hat{\phi}_{L,i}^{(n)}\left(\text{Res}_{z_n=z_0}\Omega^{(n)}\right)_{ij}^{-1}+\mathcal{O}(z_n-z_0).\label{sol}
\end{equation}
Notice that only the first term in Eq.~\eqref{sol} contributes to the residue in Eq.~\eqref{almostthere}; hence, recalling that $\hat{\phi}_{L,i}^{(n)}$ and $\Omega^{(n)}$ have at most simple poles in $z_n$, it is possible to substitute
\begin{equation}
\hat{\psi}_{L,j}^{(n)}\to \sum_i \hat{\phi}_{L,i}^{(n)}\left(\Omega^{(n)}\right)_{ij}^{-1}.
\end{equation}
Remembering that $\text{adj}\,\,\Omega^{(n)}=\text{det} (\Omega^{(n)})\cdot(\Omega^{(n)})^{-1}=(\Omega^{(n)})^{-1}P/Q$ (we write $\det\Omega^{(n)}=P/Q$) we can finally obtain
\begin{align}
\bra{\phi_L}\ket{\phi_R}_\omega&=\sum_{z_0\in S_n}\sum_{i,j=1}^{\nu_{n-1}}\text{Res}_{z_n=z_0}\left(\frac{Q}{P}\hat{\phi}_{L,i}^{(n)}\left(\text{adj}\Omega^{(n)}\right)_{ij}\hat{\phi}_{R,j}^{(n)}\right)\nonumber\\
&=\frac{1}{2\pi i} \int_{\gamma}dz_n\sum_{i,j=1}^{\nu_{n-1}}\frac{Q\hat{\phi}_{L,i}^{(n)}\left(\text{adj}\Omega^{(n)}\right)_{ij}\hat{\phi}_{R,j}^{(n)}}{P}\,,\label{prewow}
\end{align}
where by $\gamma$ we mean the contour of small counterclockwise circles around the $z_0\in S_n$ points.\\
Let's take a deeper look at Eq.~\eqref{prewow}.
In the univariate case, the twist is a certain $\omega=\omega_1dz_1$ with $\omega_1=P/Q$ , where $P$ and $Q$ are certain polynomials in the only variable $z_1$. Calling $C_1=\left\lbrace z_1\in\mathbb{C}: P(z_1)=0\right\rbrace$ and $I_1=<P>$ the ideal generated from $P$, then $C$ is the algebraic variety generated from $I_1$: $C=V(I_1)$. If $\phi_{L,R}$ have only simple poles, then we can express the intersection number between such forms as a global residue:
\begin{equation}
\bra{\phi_L}\ket{\phi_R}_\omega=-\text{Res}_{<P>}(Q\hat{\phi}_L\hat{\phi}_R)\,.\label{weinzierl1D}
\end{equation}
The notation in Eq.~\eqref{weinzierl1D} is as follows \cite{alggeom}: given a certain form
\begin{equation}
\tilde{\omega}=\frac{h(z)dz_1\wedge\cdots\wedge dz_n}{f_1(z)\cdots f_n(z)},\label{tildomega}
\end{equation}
in which the variables are $z=(z_1,\cdots,z_n)$ and $f\equiv f_1(z),\cdots,f_n(z)$ are $n$ holomorphic functions in the neighbourhood of the closure $\overline{U}$ of the ball $U=\left\lbrace ||z-\xi||<\epsilon\right\rbrace$ which have a unique isolated zero given by the point $\xi$ ($f^{-1}(0)=\xi)$, then the residue of $\tilde{\omega}$ at the point $\xi$ is defined as
\begin{equation}
\text{Res}_{\left\lbrace f_1,\cdots,f_n\right\rbrace,\xi}(\tilde{\omega})=\left(\frac{1}{2\pi i}\right)^n\oint_\Gamma \frac{h(z)dz_1\wedge\cdots\wedge dz_n}{f_1(z)\cdots f_n(z)}.\label{resdef}
\end{equation}
In Eq.~\eqref{resdef} the contour $\Gamma$ is given by $\Gamma=\left\lbrace z: |f_i(z)=\epsilon_i|\right\rbrace$ with orientation given by $d(\text{arg}f_1)\wedge\cdots\wedge d(\text{arg}f_n)\geq 0$. The subscript $<P>$ in Eq.~\eqref{weinzierl1D} refers to the fact that the residue in Eq.~\eqref{resdef} does not depend strictly on the $f$ functions in the denominator, but more in general to the ideal they generate $<f_1,\cdots,f_n>$, as we will see with more detail later in section \ref{newideas}. Notice that Eq.~\eqref{weinzierl1D} is equivalent to Eq.~\eqref{intersimple}. With this new formalism, Eq.~\eqref{prewow} can be similarly rewritten as
\begin{equation}
\bra{\phi_L}\ket{\phi_R}_\omega=-\text{Res}_{<P>}(Q\hat{\phi}_{L,i}(\text{adj}\,\,\Omega^{(n)})_{ij}\hat{\phi}_{R,j})\,,\label{wow}
\end{equation}
in which the contour given by the many counterclockwise circles of Eq.~\eqref{prewow} is now deformed into a single clockwise contour (hence the minus sign) partially encircling a singular point, going towards infinity, coming back from infinity, partially encircling the next singular point, etc. until a closed path is formed. Eq.~\eqref{wow} generalizes Eq.~\eqref{weinzierl1D} in many dimensions. Notice how in Eq.~\eqref{matriciozzo} a generalized twist $\Omega^{(n)}$ was built, with properties very similar to the ones of the usual twist $\omega$. Writing its determinant  $\det\Omega^{(n)}=P/Q$, one can build the set $C_n=\left\lbrace z_n\in \mathbb{C}:P(z_n)=0 \right\rbrace$ (set of the points where $\Omega^{(n)}$ doesn't have full rank) and $I=<P>$, the ideal generated by $P$: much similarly to what happened when writing the univariate Eq.~\eqref{weinzierl1D}.

\

In conclusion, the procedure described in this section outlines an algorithm which, given as an input two cohomology classes $\phi_{L,R}$, computes their intersection number recursively: starting from $i=n-1$, it expands $\phi_{L,R}$ on the basis of $H_{\omega,-\omega}^{(n-1)}$ and  computes the matrix $\Omega^{(n)}$. By using transformations of the form \eqref{gaugecoeffshort}, it transforms the coefficient vectors $\phi_{L,R,j}$ in equivalent coefficients but with only simple poles in $z_n$ and therefore can finally compute the residue \eqref{wow}. The procedure is then repeated for all $i$, until one has only an intersection number depending on the last variable.
\subsection{Moving onwards: an open problem}\label{newideas}
Recursive approaches to the computation of intersection numbers are known, but the time needed for the computation can be extremely long because of the many steps composing the algorithm. A new question arises: is there any way to compute \eqref{interdef} without having to introduce some type of recursivity (whether in the number of variables, or in the reduction of the order of the poles appearing in the twisted cocycles)? The fact that a recursive solution of the problem exists suggests there might be some underlying property of intersection numbers which could lead to a more fundamental method of computation. Being able to understand this could make the computation faster, but it would also imply getting a deeper understanding of the problem: in the end, it would be possible to understand what Feynman integrals represent - both in an algebraic and geometric perspective, not only by a straightforward computation problem point of view.

\

As a first note, we observe that in every case computation can be carried out, intersection numbers are always written in terms of a residue of the form Eq.~\eqref{resdef}: we find it instructive to focus on Eq.~\eqref{resdef} and highlight some of its properties. 
\begin{itemize}
\item[1:] \textbf{Non degenerate case}\\
If $f=f_1,\cdots, f_n$ is non degenerate (i.e. its Jacobian evaluated in $0$ is $J_f(0)\neq 0$), then Eq.~\eqref{resdef} can be evaluated by the introduction of a change of variables $w:=f(z)$. Using the usual Cauchy formula leads to
\begin{align}
\left(\frac{1}{2\pi i}\right)^n\int_{\Gamma}\tilde{\omega} &=\left(\frac{1}{2\pi i}\right)^n\int_{|w_i|=\epsilon_i}h(f^{-1}(w)) \frac{dw_1\cdots dw_n}{J_f(w)}\frac{1}{w_1}\cdots\frac{1}{w_n}=\frac{h(f^{-1}(0))}{J_f(0)}\,.
\end{align}
\

\item[2:] \textbf{$h\in \text{I}(f)$ ideal generated from the $f_i$}\\
In this case the residue is 0. If, for example, $h(z)=g(z)f_1(z)$, then $\tilde{\omega}$ is holomorphic in a bigger set, which we call $U_1:=U\setminus (D_2+\cdots D_n)$.\\
Then the contour $\Gamma_1=\left\lbrace z\,\,\text{such that}\,\,|f_1(z_1)|\leq\epsilon_1,|f_i(z_i)|=\epsilon_i \,\,\text{for}\,\,i\neq 1\right\rbrace$, which is an element of $U_1$, has a boundary $\delta \Gamma_1=\pm \Gamma$: by the Stokes theorem, the residue of $\tilde{\omega}$ along $\Gamma$ is then 0.
\end{itemize}
\
While the general degenerate case is hard to compute,in \cite{Sogaard:2013fpa} is suggested that if each of the $z_i$ appearing in $\tilde{\omega}$ depends only on a single variable $z_i$, it is possible to factorize Eq.~\eqref{resdef} as
\begin{equation}
\text{Res}_{(0)}(\tilde{\omega})=\left(\frac{1}{2\pi i}\right)^n\oint\frac{dz_1}{f_1(z_1)}\cdots\oint\frac{dz_n}{f_n(z_n)}h(z).\label{factorize}
\end{equation}
How can this be achieved? Algebraic geometry \cite{alggeom} shows that actually Eq.~\eqref{resdef} can be reinterpreted in terms of sheaf cohomology, which allows to obtain an important theorem:
\begin{Theorem}\label{teotrasf}\textbf{Transformation Theorem}\\
Let $f=\left\lbrace f_1,\cdots, f_n \right\rbrace$ and $g=\left\lbrace g_1,\cdots, g_n \right\rbrace$ be holomorphic maps $f,g:\bar{U}\to \mathbb{C}^n$ such that $f^{-1}(0)=g^{-1}(0)=0$ and $g_i(z)=\sum_j A_{ij}(z)f_j(z)$, with $A_{ij}$ being a holomorphic matrix. Then
\begin{equation}
\textup{Res}_{(0)}\left(\frac{h(z)dz_1\wedge\cdots\wedge dz_n}{f_1(z)\cdots f_n(z)}\right)=\textup{Res}_{(0)}\left(\frac{h(z)dz_1\wedge\cdots\wedge dz_n}{g_1(z)\cdots g_n(z)}\det A\right)\,.
\end{equation}
\end{Theorem}
Given $n$ functions $g_i(z)=g_i(z_i)$ obeying the hypotheses of Theorem \ref{teotrasf}, then the residue in Eq.~\eqref{resdef} would take the form in Eq.~\eqref{factorize}. Calling $R=\mathbb{K}(z_1,\cdots,z_n)$ the ring over a field $\mathbb{K}$ with $n$ variables $z_1,\cdots,z_n$, we notice that the ideal $I=<f_1,\cdots,f_n>\subset R$ generated by the $f_1, \cdots, f_n$, defined as $I=\left\lbrace\sum_i h_i f_i\,\,\vert \,\,h_i\in R\right\rbrace$ implies that the ideal $J$ generated by the $g_i$ polynomials is a subset of $I$: 
\begin{equation}
J=\left\lbrace\sum_i \tilde{h}_i g_i=\sum_{ij}\tilde{h}_iA_{ij}(z)f_j=\sum_i\left(\sum_j\tilde{h}_j A_{ji}(z)\right)f_i\,\,\vert\,\, \tilde{h}_i\in R\right\rbrace\subset I.
\end{equation}
This observation implies that theorem \ref{teotrasf} can be reinterpreted in the following way: the residue does not depend on the specific generators of the ideal $I$: it is possible to use the generators of an ideal which is a subset of the original one. As suggested in \cite{Sogaard:2013fpa}, there is a way to generate such an ideal: the exploitation of Gr\"obner bases, which can be regarded as a set of polynomials with special properties. In Appendix \ref{appendixGrobner} we give a brief introduction on the topic (for more details, see for example \cite{grobnerbook}). \\
Going back to the problem of the computation of the residue and following the idea from \cite{Sogaard:2013fpa}, we recall that the aim is to find some new $g_i(z)=A_{ij}(z)f_j(z)$ which are more convenient than the original $f_i$ polynomials. 
In order to find $g_i(z)=g_i(z_i)$ for a fixed $i$, we define a Lexicographic order in which the $i$-th variable $z_i$ is the ``least important variable'': 
\begin{equation}
z_{i+1}\succ z_{i+2}\succ\cdots\succ z_n\succ z_1\succ z_2\succ\cdots\succ z_i\,.
\end{equation}
Starting from the original $f_1,\cdots,f_n$ polynomials and constructing a Gr\"obner basis while keeping this order in mind, one fact is bound to happen: one polynomial in the basis must depend only on the $z_i$ variable. In fact, the obtained Gr\"obner basis can be reduced with respect to itself (i.e. each polynomial in the basis can be divided with respect to the other ones): because of the chosen order, this process eliminates progressively the variables $z_1,\cdots,z_{i-1},z_{i+1},z_n$, leaving at least one polynomial depending only on the last variable $z_i$. This polynomial is $g_i(z)=g_i(z_i)$. Notice that, by the Buchberger's algorithm, the many elements of the basis are constructed as a linear combination of the former polynomials: all of them are of the form requested by the Transformation Theorem hypotheses. Repeating this process for all the $n$ possible Lexicographic orders, at the end a set $\{g_1(z_1),\cdots,g_n(z_n)\}$ is created. The computation of Gr\"obner bases is achieved by the software \textsl{Macaulay2} \cite{macaulay2}: here we give a brief example of how this procedure works.
\begin{Example}\label{ex1}
As an example, we study the residue for $\tilde{\omega}=\frac{z_2 dz_1\wedge dz_2}{z_1^2(z_2-z_1)}$.

\

Here we use $h(z)=z_2$, $f_1(z)=z_1^2, f_2(z)=z_2-z_1$. The Jacobian is degenerate in $(0,0)$, hence we compute the residue of $\tilde{\omega}$ around $(0,0)$.\\
First, we look for $g_1(z_1)$: in this case there is no need to start creating a Gr\"obner basis associated to the ordering $z_2\succ z_1$: $f_1(z)$ already depends only on $z_1$, hence we simply consider $g_1(z_1)=f_1(z)=z_1^2$.\\
We now switch to the problem of finding $g_2(z_2)$, using the order $z_1\succ z_2$: it is clear that it is not possible to obtain a function depending only on $z_2$ dividing each $f_i$ by the other remaining polynomial, so we build the Gr\"obner basis. Notice how this suggests that $\left\lbrace f_1, f_2\right\rbrace$ cannot form a complete Gr\"obner basis for the chosen order. 
\begin{equation}
S(f_1,f_2)=\text{lcm}(z_1^2,z_1)\left(\frac{f_1}{z_1^2}+\frac{f_2}{z_1}\right)=z_1z_2.
\end{equation}
Dividing $S(f_1,f_2)$ with respect to the set $\left\lbrace f_1,f_2\right\rbrace$ with the goal to obtain a remainder which is ``smaller'' in the sense of the specified order leads to
\begin{equation}
S(f_1,f_2)=z_1z_2\xrightarrow{+z_2f_2}z_2^2\,,
\end{equation}
which is irreducible. We then add $S(f_1,f_2)$ (or better yet, its remainder, which has already went through the division process) to the original set: we obtain $\left\lbrace f_1, f_2, f_3 \right\rbrace$ with $f_3=z_2^2$. Without checking if this new set is a Gr\"obner basis or not (which it actually is), we identify $g_2(z_2)=f_3(z)=z_2^2$.
\begin{gather}
\begin{pmatrix} g_1 \\ g_2 \end{pmatrix}
=
\begin{pmatrix} A_{11}& A_{12}\\ A_{21}& A_{22} \end{pmatrix}
\begin{pmatrix} f_1 \\ f_2 \end{pmatrix}\,\,\,\text{with}\,\,\,
\begin{pmatrix} A_{11} & A_{12}\\ A_{21} & A_{22} \end{pmatrix}=
\begin{pmatrix} 1 & 0\\ 1 & z_1+z_2 \end{pmatrix}\Rightarrow \det A=z_1+z_2
\end{gather}
By Theorem \ref{teotrasf}, the residue can be written as
\begin{align}
\text{Res}_{\left\lbrace f_1,f_2\right\rbrace(0,0)}\tilde{\omega}&=\text{Res}_{\left\lbrace g_1,g_2\right\rbrace(0,0)}\left(\frac{z_2(z_1+z_2)dz_1\wedge dz_2}{z_1^2 z_2^2}\right)=\left(\frac{1}{2\pi i}\right)^2\oint\frac{dz_1}{z_1^2}\oint dz_2\frac{z_2(z_2+z_1)}{z_2^2}=1.
\end{align}
\end{Example}

We wonder if it is always possible to write an intersection number in terms of a residue of the kind in Eq.~\eqref{resdef}, not only when simple poles are involved. Regarding the emerging importance of the global residue theorem in the computation of intersection numbers, we mention a recent work \cite{Mizera:2019vvs}. It is argued that intersection numbers between forms can be computed by an expansion in the parameter $\gamma$ appearing in the Baikov polynomial $B^\gamma$, both for big and small $\gamma$: it appears that every term in the expansion can be written as a residue. We briefly focus on this result in order to give an intuitive idea of the reasons behind it without entering into too much detail.\\
We start by considering the manifold $M=\mathbb{CP}^m\setminus \cup_{i=1^k}H_i$, where $H_i$ represent some hyperplanes defined by certain equations $f_i=0$. With these functions it is possible to build the holomorphic function
\begin{equation}
    W=\sum_{i=1^k}\alpha_i \log{f_i}\,,\label{W}
\end{equation}
which only has logarithmic singularities along the $H_i$ hyperplanes. Eq.~\eqref{W} also defines a holomorphic 1-form $dW$, which in the language introduced in Section \ref{secbaikov} corresponds to $d\log(B)$. Similarly, it is possible to construct the operator 
\begin{equation}
    \nabla_{dW}=d+\gamma dW\wedge\,,\label{Deltadw}
\end{equation}
which corresponds to our $d+\omega\wedge$ (remember that $\omega=d\log(B^\gamma)=\gamma\,d\log(B)$). We consider the space $\Omega^k(M)$ of the k-forms on $M$ and construct two sequences:

\begin{equation}
  0 \xrightarrow{dW\wedge} \Omega^1(M) \xrightarrow{dW\wedge} \cdots \xrightarrow{dW\wedge} \Omega^m(M) \xrightarrow{dW\wedge} 0 \label{1sequence}
\end{equation}
and 
\begin{equation}
   0 \xrightarrow{\nabla_{dW}} \Omega^1(M) \xrightarrow{\nabla_{dW}} \cdots \xrightarrow{\nabla_{dW}} \Omega^m(M) \xrightarrow{\nabla_{dW}} 0\,. \label{2sequence}
\end{equation}
Notice how both sequences are well defined, as $dW\wedge dW=0$ and $\nabla_{dW}^2=0$. The associated cohomology groups $H^k(M,dW\wedge)$ and $H^k(M,\nabla_{dW})$ are respectively defined as
\begin{equation}
    H^k(M,dW\wedge)\coloneqq\{\phi_k\in \Omega^k(M) | dW\wedge\phi_k=0\}/\{\phi_k\in \Omega^k(M) | \phi_k=dW\wedge\phi_{k-1}, \text{with}\,\phi_{k-1}\in \Omega^{k-1}(M)\}\label{om1}
\end{equation}
and
\begin{equation}
    H^k(M,\nabla_{dW})\coloneqq\{\phi_k\in \Omega^k(M) | \nabla_{dW}\phi_k=0\}/\{\phi_k\in \Omega^k(M) | \phi_k=\nabla_{dW}\phi_{k-1}, \text{with}\,\phi_{k-1}\in \Omega^{k-1}(M)\}\label{om2}
\end{equation}

All cohomology groups are trivial for each $k\neq m$: only $H^m(M,dW\wedge)$ and $H^m(M,\nabla_{dW})$ contain useful information. We define two pairings, the first between $H^m(M,dW\wedge)$ and itself and the second between $H^m(M,\nabla_{dW})$ and $H^m(M,\nabla_{-dW})$. In the first case, taking $\phi_{\pm}\in H^m(M,\nabla_{dW})$ it is possible to define \cite{Mizera:2017rqa}
\begin{equation}
    (\phi_{-}|\phi_{+})_{dW,0}=\text{Res}_{dW=0}\left(\frac{\hat{\phi}_{-}\hat{\phi}_{+}d^m z}{\partial_1 W\cdots\partial_m W}\right)\,,\label{pair1}
\end{equation}
where the $\,\hat{}\,$ notation strips the forms from their $d^m z$ part and where $\text{Res}_{dW=0}$ represents a sum of residues around each critical point given by $\partial_i W=0$, for $i=1,\cdots, m$.\\
In the second case, taking $\phi_{\pm}\in H^m(M,\nabla_{\pm dW})$ allows to define
\begin{equation}
    \bra{\phi_{-}}\ket{\phi_{+}}_{dW}=\left(\frac{\gamma}{2\pi i}\right)^m\int_M\phi_{-}\wedge\phi_{+}^c\,,\label{pair2}
\end{equation}
where $\phi_{+}^c$ represents the compact version of $\phi_{+}$ belonging to the same equivalence class as $\phi_{+}$, similarly to the usual definitions from Cho and Matsumoto \cite{cho1995} (in Eq.~\eqref{interdef} we take the compact version of $\phi_{-}$, but it is just a matter of convention): Eq.~\eqref{pair2} represents an intersection number. It is argued that, like ~\eqref{om1} looks like a limit for big $\gamma$ of ~\eqref{om2}, then the pairings in Eq.~\eqref{pair1} and Eq.~\eqref{pair2} must be linked. It is found that Eq.~\eqref{pair2} can be expanded as follows:
\begin{equation}
   \bra{\phi_{-}}\ket{\phi_{+}}_{dW}=\sum_{k=0}^\infty \gamma^{-k} (\phi_{-}|\phi_{+})_{dW,k}\,
\end{equation}
where each $(\phi_{-}|\phi_{+})_{dW,k}$ is a \textit{higher residue pairing} \cite{saitohigher} and has the form of a $\text{Res}_{dW=0}$ of certain functions which only contain $\hat{\phi}_{-}$, $\hat{\phi}_{+}$ in the numerator and the derivatives $\partial_i W$ both in numerator and denominator.

\
Notice that, as it can be seen in Example \ref{ex1}, any residue of the kind we met can be calculated independently of the order of the poles appearing. If - as suggested by \cite{Mizera:2019vvs} - intersection numbers were always found to be expressed only in terms of residues, then a procedure involving Gr\"obner bases could be exploited. The biggest difference between the residue in Eq.~\eqref{resdef} and the integral in Eq.~\eqref{I} is the type of cohomology involved. Consider a $(n-1)$-form $\tilde{\xi}$ of the same type as in Eq.~\eqref{tildomega}:
\begin{equation}
\tilde{\xi}=\frac{\tilde{h}(z)dz_1\wedge\cdots\wedge dz_{n-1}}{\tilde{f}_1(z)\cdots \tilde{f}_{n-1}(z)}.\label{tildexi}
\end{equation}
The form $\tilde{\xi}$ is holomorphic everywhere in $U=\left\lbrace z\in \mathbb{C}^{n-1}\,\,\text{such that}\,\,\vert\vert z||<\epsilon\right\rbrace$ except from the set $\tilde{D}$, defined as $\tilde{D}=\tilde{D}_1+\cdots+\tilde{D}_{n-1}$, where each $\tilde{D}_i=(\tilde{f}_i)$ is the divisor of $\tilde{f}_i$ (namely, it is the formal sum of the points sent to 0 by applying $\tilde{f}_i$). This implies that its external derivative is null: $d\tilde{\xi}=0$, so that $\int_\Gamma d\tilde{\xi}=0$, where $\Gamma$ is the contour appearing in Eq.~\eqref{resdef}. As the addition of a term of the type $d\tilde{\xi}$ gives no contribution to Eq.~\eqref{resdef}, the residue of $\tilde{\omega}$ depends on the de Rham cohomology class $[\tilde{\omega} ]\in \text{H}^n_\text{dR}(U\setminus D)$, where $D=D_1+\cdots+D_n$, with $D_i=(f_i)$. On the other side, the integral \eqref{I} similarly depends on the twisted cohomology class of $\phi$. This difference is due to the presence of $u=B^\gamma$ inside the integral: as we have seen earlier, the effect of $u$ is the introduction of some cuts in the complex space: we wonder if, in a sense, $u$ can be considered as part of the geometry and  can be ``transferred inside the sign of integral''. The integral $I$ on the cut complex space can also be thought as an integral on a certain Riemann surface where the cuts have been glued together. We wonder if, in this sense, it is possible to rewrite the integral \eqref{I} as multi-variable residue with a form similar to Eq.~\eqref{resdef}, or at least as a sum of terms of this kind as suggested in \cite{Mizera:2019vvs}. In this case, a similar approach could be followed: while Theorem \ref{teotrasf} emerges from the interpretation of the residue by means of sheaf cohomology  - in which the cohomology of interest is the de Rham cohomology - another similar theorem for the twisted cohomology group should exist, telling how one can adapt the function inside the integral while obtaining the same residue. As a first guess, it shouldn't differ too much from the original theorem. Once having derived how the integrand can be modified to obtain a more pleasing integral, the focus could remain the same: having functions in the denominator which depend only from a single variable, so that the computation of the integral can finally be performed.

\

Of course, while we stress that these observations only outline a possible way to tackle the problem of computing an intersection number as a whole, we highlight how certainly the path has to lie in looking at Eq.~\eqref{interdef} not only with straightforward computation in mind, but with the objective to understand a deeper structure: the path certainly lies in the abstract procedure outlined in Section \ref{veryhardmaths}. In this section, we introduced many of the known practical tools: some have already been exploited for the actual evaluation of intersection numbers (Sections \ref{logforms}-\ref{recursive}), while in Section \ref{newideas} we introduced some concepts and definitions which may become part of the practical instruments needed to pursue, in the end, the computation of the Feynman integral. 

\section{An explicit example of Feynman integral}
We want to conclude this review by showing how to explicitly apply intersection theory to the computation of Feynman integrals. The present example is due to Manoj Mandal and Federico Gasparotto.
Let us consider the double box diagram for massless particles:
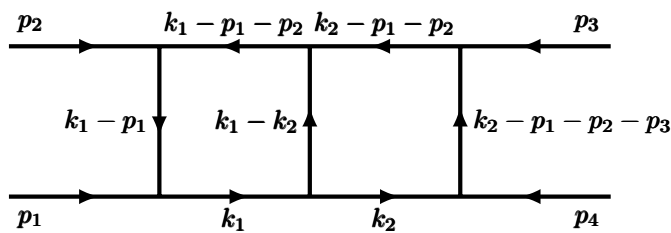
\begin{figure}[!h]
\begin{center}
\begin{tikzpicture}[>=latex,decoration={zigzag,amplitude=.5pt,segment length=2pt}]
\draw [ultra thick] (-4,1) -- (4,1) (-4,-1) -- (4,-1) (-2,1) -- (-2,-1) (0,1) -- (0,-1) (2,1) -- (2,-1);
\draw [ultra thick,->] (-3,1) -- (-2.8,1); 
\draw [ultra thick,->] (-3,-1) -- (-2.8,-1);
\draw [ultra thick,->] (3,1) -- (2.8,1);
\draw [ultra thick,->] (3,-1) -- (2.8,-1);
\draw [ultra thick,->] (-1,-1) -- (-0.8,-1);
\draw [ultra thick,->] (1,-1) -- (1.2,-1);
\draw [ultra thick,->] (-2,0) -- (-2,-0.2);
\draw [ultra thick,->] (0,0) -- (0,0.2);
\draw [ultra thick,->] (2,0) -- (2,0.2);
\draw [ultra thick,->] (-1,1) -- (-1.2,1);
\draw [ultra thick,->] (1,1) -- (0.8,1);
\node at (-3.7,-1.3) {$\pmb {p_1}$};
\node at (-3.7,1.3) {$\pmb {p_2}$};
\node at (3.7,1.3) {$\pmb {p_3}$};
\node at (3.7,-1.3) {$\pmb {p_4}$};
\node at (-1,-1.3) {$\pmb {k_1}$};
\node at (1,-1.3) {$\pmb {k_2}$};
\node at (-1,1.3) {$\pmb {k_1-p_1-p_2}$};
\node at (1,1.3) {$\pmb {k_2-p_1-p_2}$};
\node at (-2.7,0) {$\pmb {k_1-p_1}$};
\node at (-0.7,0) {$\pmb {k_1-k_2}$};
\node at (3.5,0) {$\pmb {k_2-p_1-p_2-p_3}$};
\end{tikzpicture}
\caption{Double box diagram.}
\label{FigLast}
\end{center}
\end{figure}

After reducing to the maximal cut and choosing a basis of Master Integrals, we will show how to use intersection theory for determining the differential equations for the MIs as functions of a Mandelstam's variable.
Of course, momentum conservation requires $p_4=-p_1-p_2-p_3$, so we have $E=3$ independent external momenta, while $L=2$, so we need
\begin{align}
 M=EL+\frac {L(L+1)}2=9
\end{align}
Baikov's variables. Seven are provided by the true denominators 
\begin{align}
 D_1&=k_1^2, \qquad\ D_2=(k_1-p_1)^2, \qquad\ D_3=(k_1-p_1-p_2)^2, \qquad D_4=(k_1-k_2)^2, \cr
 D_5&=(k_2-p_1-p_2)^2, \qquad\ D_6=(k_2-p_1-p_2-p_3)^2, \qquad\ D_7=k_2^2,
\end{align}
while the remaining two are given by the fake denominators
\begin{align}
 D_8&=(k_2-p_1)^2, \qquad\ D_9=(k_1-p_1-p_2-p_3)^2.
\end{align}
Using the Mandelstam's variables\footnote{Notice also that the massless condition and momentum conservation imply $s+t=-2p_1\cdot p_3$}
\begin{align}
 s=2p_1\cdot p_2, \qquad t=2p_2\cdot p_3,
\end{align}
the Baikov polynomials $B$ in the variables $z_j=D_j$, $j=1,\ldots,9$, is given by
{\small
\begin{align}
 \frac{16B}{(-st(s+t))^{\frac {4-d}2}}=& st^2 z_1 z_3+s^2t^2 z_4-st^2 z_1 z_4-st^2 z_2z_4 -st^2z_3z_4 
 +(s^2 t +st^2) z_4^2 -st^2 z_1 z_5 +t^2 z_1^2 z_5 +st z_1 z_2 z_5\cr 
 &-t^2 z_1 z_3 z_-st^2 z_4 z_5 -2st z_1 z_4 z_5-t^2 z_1 z_4 z_5+st z_2 z_4 z_5 +t^2 z_3z_4z_5 +t^2 z_1 z_5^2 +t z_1^2 z_5^2 \cr
 &-t z_1 z_2 z_5^2 -s^2 t z_2 z_6 +st z_1 z_2 z_6 +s^2 z_2^2 z_6 -2st z_1 z_3 z_6 +st z_2 z_3 z_6 -s^2 t z_4 z_6 +st z_1 z_4 z_6 \cr
 &-s^2 z_2 z_4 z_6 -2st z_2 z_4 z_6 +st z_3 z_4 z_6 +st z_1 z_5 z_6 -t z_1^2 z_5 z_6 +st z_2 z_5 z_6 +(s+t) z_1 z_2 z_5 z_6 \cr 
 & -s z_2^2 z_5 z_6 +t z_1 z_3 z_5 z_6 -t z_2 z_3 z_5 z_6 +s^2 z_2 z_6^2 -s z_1 z_2 z_6^2 +sz_2^2 z_6^2+s z_1 z_3 z_6^2\cr 
 &-s z_2 z_3 z_6^2 -st^2 z_3 z_7-t^2 z_1 z_3 z_7 +st z_2 z_3 z_7 +t^2 z_3^2 z_7 -st^2 z_4 z_7 +t^2 z_1z_4z_7 +st z_2 z_4 z_7 \cr
 &-2st z_3 z_4 z_7 -t^2 z_3 z_4 z_7 +st^2 z_5 z_7 -t^2 z_1 z_5 z_7 -2st z_2 z_5 z_7 +t z_1 z_2 z_5 z_7 +s z_2^2 z_5 z_7 \cr 
 &-t^2 z_3 z_5 z_7 -2t z_1 z_3 z_5 z_7 +t z_2 z_3 z_5 z_7 +st z_2 z_6 z_7 -t z_1 z_2 z_6 z_7 -s z_2^2 z_6 z_7 +st z_3 z_6 z_7 \cr
 &+t z_1 z_3 z_6 z_7 +(s+t) z_2 z_3 z_6 z_7-t z_3^2 z_6 z_7 +t^2 z_3 z_7^2 -t z_2 z_3 z_7^2 +t z_3^2 z_7^2-2st z_1 z_3 z_8 \cr
 &-s^2 t z_4 z_8 +st z_1 z_4 z_8 +s^2 z_2 z_4 z_8 +st z_3 z_4 z_8 +st z_1 z_5 z_8 -t z_1^2 z_5z_8 -s z_1 z_2 z_5 z_8 \cr
 &+t z_1 z_3 z_5 z_8 +s^2 t z_6 z_8 -2st z_1 z_6 z_8 +t z_1^2 z_6 z_8 -s^2 z_2 z_6 z_8 +s z_1 z_2 z_6 z_8 -2st z_3 z_6 z_8 \cr
 &-2(s+t) z_1 z_3 z_6 z_8 -s z_2 z_3 z_6 z_8 +t z_3^2 z_6 z_8 +st z_3 z_7 z_8 +t z_1 z_3 z_7 z_8 -s z_2 z_3 z_7 z_8 \cr
 &-t z_3^2 z_7 z_8 +s z_1 z_3 z_8^2 +s^2 t z_2 z_9 -s^2 t z_4 z_9 +st z_1 z_5 z_9 -2st z_2 z_5 z_9 +st z_4 z_5 z_9 \cr
 &-t z_1 z_5^2 z_9 +t z_2 z_5^2 z_9 -s^2 z_2 z_6 z_9 +s^2 z_4 z_6 z_9 -s z_1 z_5 z_6 z_9 +s z_2 z_5 z_6 z_9 -2st z_2 z_7 z_9 \cr
 &+st z_3 z_7 z_9 +st z_4 z_7 z_9 -2st z_5 z_7 z_9 +t z_1 z_5 z_7 z_9 -2(s+t) z_2 z_5 z_7 z_9+t z_3 z_5 z_7 z_9 +s z_2 z_6 z_7 z_9 \cr
 &-s z_3 z_6 z_7 z_9 +t z_2 z_72 z_9 -t z_3 z_7^2 z_9 -s^2 t z_8 z_9 +st z_1 z_8 z_9 -s^2 z_2 z_8 z_9 +st z_3 z_8 z_9 \cr
 &-s^2 z_4 z_8 z_9 -2st z_4 z_8 z_9 +st z_5 z_8 z_9 +s z_1 z_5 z_8 z_9 +t z_1 z_5 z_8 z_9 +s z_2 z_5 z_8 z_9 -t z_3 z_5 z_8 z_9 \cr
 &-s^2 z_6 z_8 z_9 +s z_1 z_6 z_8 z_9 -2s z_2 z_6 z_8 z_9 +s z_3 z_6 z_8 z_9 +st z_7 z_8 z_9 -t z_1 z_7 z_8 z_9 +s z_2 z_7 z_8z_9 \cr 
 &+(s+t) z_3 z_7 z_8 z_9 +s^2 z_8^2 z_9 -s z_1 z_8^2 z_9 -s z_3 z_8^2 z_9 +s z_5 z_7 z_9^2 +s^2 z_8 z_9^2 -s z_5 z_8 z_9^2 \cr
 &-s z_7 z_8 z_9^2 +s z_8^2 z_9^2.
\end{align}
}
However, we will not work with the complete diagram but on the maximal cut, which correspond to set $z_j=0$, $j=1,\ldots,7$ so that only the denominators $D_8$ and $D_9$ survive. Using $w_1=z_8$ and $w_2=z_9$ we then get for the
maximal cut
\begin{align}
 B(w_1,w_2)=\frac {(-st(s+t))^{\frac {4-d}2}}{16} \left[ 
 -s^2 t w_1 w_2 + s^2 w_1^2 w_2 + s^2 w_1 w_2^2 + s w_1^2 w_2^2
 \right].
\end{align}
The twisting section is thus 
\begin{align}
 u=B^{\gamma}
\end{align}
with 
\begin{align}
 \gamma=\frac {d-E-L-1}2=\frac {d-6}2.
\end{align}
The associated connection 1-form is
\begin{align}
 \omega=\gamma\frac {-s t + 2sw_1 + s w_2 + 2w_1 w_2}{-s t w_1 + s w_1^2  + s w_1 w_2 + w_1^2 w_2}dw_1+\gamma\frac{-s t + s w_1 + 2s w_2 + 2w_1 w_2}{-s t w_2 + s w_1 w_2 + s w_2^2 + w_1 w_2^2} dw_2.
\end{align}
The equation $\omega=0$ has two solutions
\begin{align}
 w_1=w_2=w_\pm\equiv-\frac 34 s\pm \sqrt {\frac 9{16} s^2+\frac 12 st}, \label{glizeri}
\end{align}
so there are two MIs. Since we have to compute bivariate intersection numbers, we consider $w_1$ as internal integration variable. The number of corresponding MIs is given by the solutions of $\omega_1=0$ for $w_1$, where 
$\omega=\omega_1 dw_1+\omega_2 dw_2$. Therefore, we need just one MI for the internal integration.\\
As MIs we choose $I_{[1,1,1,1,1,1,1,0,0]}$ and $I_{[1,1,1,1,1,1,1,-1,0]}$, which correspond to the forms 
\begin{align}
 \bra{e_1}\equiv \bra {1} , \qquad \bra{e_2}\equiv \bra {w_1},
\end{align}
respectively. Considering them as functions of the external parameter $s$, they have to satisfy the differential equation system
\begin{align}
 \bra{\partial_s e_j+(\partial_s \log u) e_j}=\sum_{k=1}^2O_{jk} \bra {e_k}
\end{align}
where the matrix $\pmb O$ is obtained as follows (see \cite{Frellesvig:2019kgj}, Section 3.2). Let $\phi_{sj}:=\partial_s e_j+(\partial_s \log u) e_j$, $F_{jk}=\bra {\phi_{sj}} d_k\rangle$ and $C_{jk}=\bra {e_j} d_k\rangle$ define the matrices $\pmb F$ and 
$\pmb C$. Then,
\begin{align}
 \pmb O=\pmb F \pmb C^{-1}.
\end{align}
We will then show how to compute $\pmb F$ and $\pmb C$, being very explicit for $C_{11}$ and quoting the results for all the other components. Notice that in our specific case $\partial_s e_j=0$ and 
\begin{align}
 \pmb \phi_s=
\frac {\gamma}s \begin{pmatrix}
  \frac {2 s (t - w_1 -w_2) - w_1 w_2}{ s (t - w_1 -w_2) - w_1 w_2} \\
  \frac {w_1 (2 s (t - w_1 - w_2) - w_1w_2)}{s (t - w_1 -w_2) - w_1w_2)}
\end{pmatrix}.
\end{align}
We choose $d_j=e_j$. 

\subsection{Computation of $\pmb C$}
Let us start computing $\bra 1 1\rangle$ very explicitly. According to the general construction, we need to start with the internal integration, therefore we choose an internal basis (for $w_1$ only). Recalling that the dimension is 1, we choose
\begin{align}
 \bra {\varepsilon} =\bra {1/w_1} 
\end{align}
and $\delta=\varepsilon$ for the dual. The critical point of $d_{w_1} \log u$ is
\begin{align}
 w_1=\bar w\equiv\frac s2 \frac {t-w_2}{s+w_2}.
\end{align}
Since $\varepsilon$ is a $d\log$-form, the reduced intersection matrix $\pmb C_{red}\equiv\bra \varepsilon \delta \rangle$ can be computed with the simple method of univariate intersection numbers for $d\log$-forms.
From (\ref{usoquestaformula}) we get
\begin{align}
 \pmb C_{red}=-{\rm Res}_{w_1=\bar w} \left(\frac 1{w_1^2} \frac 1{\omega_1}\right)=-\frac 1{\bar w^2} \frac 1{\partial_{w_1} \omega_1 |_{w_1=\bar w}}=\frac 1{2\gamma}.
\end{align}
To proceed iteratively, as we have seen for the recursive calculation of the multivariate intersection numbers, we have now to determine the new connection (\ref{matriciozzo}), which is now\footnote{Notice that we are not using an orthonormal frame, so
the inverse matrix $\pmb C_{red}^{-1}$ appears.}
\begin{align}
 \Omega^{(2)}= \bra {\partial_{w_2} \varepsilon+\omega_2 \varepsilon}  \delta\rangle \pmb C_{red}^{-1}.
\end{align}
Using again (\ref{usoquestaformula}), we get
\begin{align}
 \Omega^{(2)}=-2\gamma {\rm Res}_{w_1=\bar w} \left(\frac {\omega_2}{w_1^2} \frac 1{\omega_1}\right)=\frac {\gamma}{w_2}-\frac {\gamma}{w_2+s}+\frac {2\gamma}{w_2-t}.
\end{align}
Notice that $\Omega^{(2)}$ has only simple poles and its zeros are given by (\ref{glizeri}).\\
Next, we need to project $\bra {\phi_L}\equiv \bra {e_1}=\bra 1$ to $\bra \varepsilon$:
\begin{align}
 \bra{\phi_L^{(2)}}:= \bra {\phi_L} \delta \rangle \pmb C_{red}^{-1}.
\end{align}
This computation is not so direct as the previous ones, since $\phi_L$ is not of $d\log$-type, it has a double pole at infinity: for $w_1 \to 1/w_1$ we have $\bra 1 \to \bra{-1/w_1^2}$. However, we can use the cohomological property of the intersection
numbers and shift $\phi_L$ by a $\nabla_{\omega_1}$-exact form to get a representative in the same class with only simple poles. For example
\begin{align}
 \tilde \phi_L=\phi_L-\nabla_{\omega_1} \frac {w_1}{1+2\gamma}=\frac {s\gamma}{1+2\gamma} \frac {w_2-t}{-st+sw_1+sw_2+w_1w_2}
\end{align}
has two simple poles, one at $w_1=\infty$ and the other one at 
\begin{align*}
 w_1=s\frac {t-w_2}{s+w_1}.
\end{align*}
Since cohomologous formes give equal intersection numbers, we can write
\begin{align}
 \bra{\phi_L^{(2)}}=\bra {\tilde \phi_L} \delta \rangle \pmb C_{red}^{-1}
\end{align}
and use once again (\ref{usoquestaformula}) to get
\begin{align}
  \bra{\phi_L^{(2)}}=-2\gamma {\rm Res}_{w_1=\bar w} \left(\frac {\tilde \phi_L}{w_1} \frac 1{\omega_1}\right)=\frac {s\gamma}{1+2\gamma}\frac {t - w_2}{s + w_2}.
\end{align}
In a similar way, we have to project $\ket {\phi_R}\equiv \ket {e_1}=\ket 1$ on $\ket \delta$:
\begin{align}
 \ket {\phi_R^{(2)}}:= \pmb C_{red}^{-1} \langle \varepsilon \ket {\phi_R}.
\end{align}
As before, $\phi_R$ has a double pole at infinity so we replace it with the representative\footnote{notice that we have changed the connection to $-\omega_1$}
\begin{align}
 \tilde \phi_R=\phi_L-\nabla_{-\omega_1} \frac {w_1}{1-2\gamma}=\frac {s\gamma}{(2\gamma-1)} \frac {w_2-t}{-st+sw_1+sw_2+w_1w_2},
\end{align}
to get 
\begin{align}
 \ket{\phi_R^{(2)}}=-2\gamma{\rm Res}_{w_1=\bar w} \left(\frac {\varepsilon \tilde \phi_R}{\omega_1}\right)=-\frac {s\gamma}{1-2\gamma}\frac {t - w_2}{s + w_2}.
\end{align}
Before considering the final step of the computation of $C_{11}$, which involves the intersection between $\phi_L^{(2)}$ and $\phi_R^{(2)}$ twisted with $\Omega^{(2)}$, we notice that the lasts two have again double pole at infinity. Once again, we 
change the cohomology representatives in order to work with forms having only simple poles. Since
\begin{align}
 \phi_L^{(2)}=-\frac {s\gamma}{1+2\gamma} +\frac {\gamma s(s+t)}{(1+2\gamma)(s+2w_2)},
\end{align}
we conveniently define
\begin{align}
 \tilde \phi_L^{(2)}=\phi_L^{(2)}-\nabla_{\Omega^{(2)}} \left( -\frac  {s\gamma}{1+2\gamma} \frac {w_2}{1+2\gamma} \right)=\frac {2st\gamma^2}{(1+2\gamma)^2 (w_2-t)} +\frac {\gamma s(s+t)+3s^2\gamma^2+2\gamma^2 st}{(1+2\gamma)^2 (s+w_2)}.
\end{align}
Similarly we set
\begin{align}
 \tilde \phi_R^{(2)}=\phi_R^{(2)}-\nabla_{-\Omega^{(2)}} \left( \frac  {s\gamma}{1-2\gamma} \frac {w_2}{1-2\gamma} \right)=\frac {2st\gamma^2}{(1-2\gamma)^2 (w_2-t)} +\frac {-\gamma s(s+t)+3s^2\gamma^2+2\gamma^2 st}{(1-2\gamma)^2 (s+w_2)}.
\end{align}
We can finally use (\ref{prewow}) to compute the final step\footnote{once again taking into account of the non othonormality through $\pmb C_{red}$}
\begin{align}
\pmb C_{11}=&\bra 1 1\rangle=-{\rm Res}_{w_2=w_+} \left( \frac {\tilde \phi_L^{(2)} \tilde \phi_R^{(2)}}{\Omega^{(2)}}  \pmb C_{red}\right)-{\rm Res}_{w_2=w_-} \left( \frac {\tilde \phi_L^{(2)} \tilde \phi_R^{(2)}}{\Omega^{(2)}}\pmb C_{red}\right)\cr
=& s^2 \frac {3 (1-8\gamma^2)(s+t)^2-3\gamma^2s^2+4\gamma^2 t^2}{4(1-4\gamma^2)^2}.
\end{align}
In the same way we can compute the remaining matrix elements of which we just quote the final results
\begin{align}
\pmb C_{12}=&\bra 1 w_1\rangle= \frac {s^2}{8(1-\gamma)(1-4\gamma^2)^2} \left[ 2t^3\gamma (1-4\gamma^2)-2st^2 (3 - 7\gamma - 22 \gamma^2 + 44 \gamma^3)\right. \cr
&\left. -3s^2t (4 - 7\gamma - 33 \gamma^2 + 54 \gamma^3)
-s^3(6-9\gamma-54\gamma^2+81\gamma^3)\right], 
\end{align}
\begin{align}
\pmb C_{21}=&\bra {w_1} 1\rangle= -\frac {s^2}{8(1+\gamma)(1-4\gamma^2)^2} \left[ 2t^3\gamma (1-4\gamma^2)+2st^2 (3 +7\gamma - 22 \gamma^2 - 44 \gamma^3)\right. \cr
&\left. +3s^2t (4 + 7\gamma - 33 \gamma^2 - 54 \gamma^3) 
+s^3(6+9\gamma-54\gamma^2-81\gamma^3)\right]=\pmb C_{12}(\gamma\to -\gamma), \\
\pmb C_{22}=&\bra {w_1} w_1\rangle= 
\frac {s^3}{16 (1-4\gamma^2)^2 (1-\gamma^2)} \left[ -16 t^3 \gamma^2 (1-4 \gamma^2) + 12 st^2 (1 - 13 \gamma^2 + 30 \gamma^4) \right. \cr &\left.+ 12 s^2 t (2 - 23 \gamma^2 + 45 \gamma^4) + 3 s^3 (4 - 45 \gamma^2 + 81\gamma^4) \right].
\end{align}

\subsection{Computation of $\pmb F$}
This is computed exactly in the same way, now choosing $\phi_L=\phi_{sj}$ and $\phi_R=d_k$. Again, we will not repeat all passages here, but limit ourselves to quote the final results. Of course, the interested reader is invited to reproduce all 
the details carefully. We get
{\small
\begin{align}
 \pmb F_{11}=&\bra {\phi_{s1}} 1 \rangle=\frac {s}{4(1-4\gamma^2)^2} \left[ 3 s^2(2+5\gamma-18 \gamma^2-45 \gamma^3)+t^2(3+11\gamma-20\gamma^2-68\gamma^3) \right. \cr
 &\left.+3st(3+9\gamma-24\gamma^2-68\gamma^3)\right],\\
 \pmb F_{12}=&\bra {\phi_{s1}} w_1\rangle= \frac {s}{8(1-\gamma)(1-4\gamma^2)^2} \left[ 2t^3 \gamma(1+3\gamma-4\gamma^2-12\gamma^3)- 3s^3(4 + 4\gamma - 51\gamma^2 - 36\gamma^3 + 135\gamma^4) \right. \cr
&\left. -2st^2(3 + 3\gamma - 50\gamma^2 - 26\gamma^3 + 160\gamma^4)-3s^2t(6 + 7\gamma- 82\gamma^2 - 57\gamma^3 + 234\gamma^4) \right],\\
 \pmb F_{21}=&\bra {\phi_{s2}} 1\rangle= \frac {s}{8(1+\gamma)(1-4\gamma^2)^2} \left[ 2 t^3\gamma(-1 - 3\gamma+ 4\gamma^2 + 12\gamma^3)+ 4st^2(-3 - 12\gamma + 8\gamma^2 + 79\gamma^3 + 80\gamma^4) \right. \cr
&\left.+3 s^3(-6 - 19\gamma+ 39\gamma^2 + 171\gamma^3 + 135\gamma^4)+ 3s^2t(-10 - 35\gamma+ 50\gamma^2 + 273\gamma^3 + 234\gamma^4) \right],\\
\pmb F_{22}=&\bra {\phi_{s2}} w-1\rangle= \frac {s^2}{16(1-\gamma)(1-4\gamma^2)^2} \left[ 4t^3\gamma(-1 - 6\gamma- 9\gamma^2 + 24\gamma^3 + 52\gamma^4)\right. \cr
&+ 12 st^2(2 + 3\gamma- 26\gamma^2 - 45\gamma^3 + 60\gamma^4 + 114\gamma^5)+ 6s^2t(10 + 17\gamma- 115\gamma^2 - 197\gamma^3 + 225\gamma^4 + 396\gamma^5)\cr 
&\left.+ 3s^3(12 + 20\gamma- 135\gamma^2 - 225\gamma^3 + 243\gamma^4 + 405\gamma^5)\right].
\end{align}
}

\subsection{Computation of $\pmb O$}
With these elements we can finally compute the matrix $\pmb O=\pmb F \pmb C^{-1}$ defining the differential equation for the MIs:
\begin{align}
 \pmb O=\frac 1{s(s+t)}
\begin{pmatrix}
 -(2s + (d-6)t) & 4-d \\ (4-d) \frac t2 & (-16 + 3 d)\frac s2 + (-6 + d)t
\end{pmatrix}.
\end{align}

\

Of course we could proceed working out the present example, but our scope, illustrating how to apply intersection theory to Feynman calculations, is already reached, so we leave further computation to the interested readers. 
Of course, exploiting all these calculations by hand is lengthy and requires a lot of time, but the huge advantage of these methods is that they can be implemented on a computer making it very fast. 

\section{Final comments}

The strategy we followed here is rising growing interest at the time we are writing the present review. The ability of recognizing the correct cohomology underlying the Feynman Integrals and the corresponding intersection product would allow both to gain deeper understanding of such integrals and to systematize their computation. In particular, after determining the cohomology space, one is left with the choice of a convenient basis, i.e. the Master Integrals, with respect to which one can project the given amplitude, through the intersection product.

\

 In Section~\ref{Perversion} we argued the necessity of introducing the vector space $IH_M^{(p)}$ of twisted cohomology with perversity, in order to correctly deal with singular forms; in Section \ref{secnumberMI} we showed different ways in which the dimension $\nu$ of such space can be computed and interpreted. Such interpretations mostly rely on the Poincar\'{e} duality between cycles and cocycles, which in the case of singular varieties is restored by the introduction of perverse sheaves.  While in principle any set of $\nu$ independent forms in $IH_M^{(p)}$ could serve as a basis, the determination of a preferred one suitable for practical computation is still an open problem. We remark that a natural basis of cycles is represented by the Lefschetz thimbles, as we described at the end of Section~\ref{FeynmanToInters}.  This allows in principle to define a basis also for the cohomology, thanks to Poincar\'e duality. Here our excursus ends: the practical realization of such duality has yet to be performed in a form apt to get a systematic approach in the problem of the calculation of Feynman Integrals.


\section*{Acknowledgements}
First of all, we are indebted with professor Pierpaolo Mastrolia and the group of Padua (Manoj Mandal, Federico Gasparotto, Luca Mattiazzi, Vsevolod Chestnov, Hjalte Frellesvig \& Co.) for introducing us to the topic of Intersection theory methods in the Feynman integrals landscape. We also thank Pierpaolo for helping us in improving considerably our manuscript and adding content to it. We also thank Maxim Kontsevich for his lesson related to Section \ref{FeynmanToInters}, and Thibault Damour for helpful discussions.  In particular M.C. is grateful to Roberta Merlo for sharing her deep knowledge about spheres with handles. 
We are particular indebted with Manoj Mandal and Federico Gasparotto for providing us with the example of Sec. 8, including all the detailed calculations shown there. 

\newpage
\appendix
\section{Baikov Representation}\label{AppBaikov}
\numberwithin{equation}{section}
For the sake of completeness, in this Appendix we provide a short derivation of the Baikov formula \eqref{baikov}. The idea is to rewrite the integral using the independent scalar products between momenta as integration variables. They are
\begin{equation}
M=LE+L(L+1)/2\,,
\end{equation}  
where the first term is the number of scalar products $q_i\cdot p_j$ between one loop and one external momentum, while the second represents the scalar products $q_i\cdot q_j$ between
loop momenta.\\
In total, one has $m=E+L$ total independent momenta. It is useful to introduce the complete vector $k$ of all $m$ momenta:
\begin{equation}
k=(\underbrace{q_1,\cdots q_L}_{K_i=q_i\,\,\,i\leq L}\,\,\,\, , \,\,\,\underbrace{p_1,\cdots p_E}_{K_i=p_{i-L}\,\,\,i> L}).
\end{equation}
 To perform this change of variables, we start by decomposing $q_1$ as
 \begin{align}
     q_1 = q_{1\parallel} + q_{1\perp}\,,
 \end{align}
 where $q_{1\parallel}$ represents the projection of $q_1$ onto the space generated by all the other $m-1=E+L-1$ momenta $\left\lbrace q_2,\cdots q_L,p_1,\cdots,p_E\right\rbrace$, while $q_{1\perp}$ represents the orthogonal component with respect to said space. This is done for every variable: for each $i$, the corresponding $q_i$ is projected onto the space generated by the momenta which come next in the vector $k$:  $\left\lbrace q_{i+1},\cdots q_L,p_1,\cdots,p_E\right\rbrace$. In the last step, one projects the last $q_L$ along the space of the $E$ external momenta $\left\lbrace p_1,\cdots,p_E\right\rbrace$.
This decomposition leads to
\begin{equation}
d^Dq_1\cdots d^Dq_L=(d^{E+L-1}q_{1\parallel}\,\,d^{D-E-L+1}q_{1\perp})\,\,\cdots\,\,(d^Eq_{L\parallel}\,\,d^{D-E}q_{L\perp})\,.\label{decomp}
\end{equation}

We introduce the Gram matrix
\begin{equation}
G(k)=
\begin{pmatrix} s_{11}\,\,\,\cdots\,\,\, s_{1m}\\ \vdots\,\,\,\,\,\,\,\,\,\,\,\,\,\,\,\,\,\,\,\,\,\,\,\,\,\,\,\,\,\,\,\vdots \\ s_{m1} \,\,\,\cdots\,\,\, s_{mm} \end{pmatrix}\label{gmatrix}\,,
\end{equation}
in which the entries are $s_{ij}=k_ik_j$. The square root of the determinant (which we will adress as $G$ instead of $\det (G)$ for brevity) of the matrix \eqref{gmatrix} represents the volume of the parallelotope generated by the elements of $k$. With this interpretation, one can write
\begin{itemize}
\item[1:] Parallel component $q_{i\parallel}$.
\begin{align}
d^{E+L-i}q_{i\parallel}=\frac{ds_{i,i+1}\cdot ds_{i,i+2}\cdots ds_{i,E+L}}{G^{1/2}(q_{i+1},\cdots,q_L,p_1,\cdots,p_E)}=\prod_{j=i+1}^{E+L}\frac{ds_{ij}}{G^{1/2}(q_{i+1},\cdots,q_L,p_1,\cdots,p_E)}\,.\label{par}
\end{align}
In the numerator of Eq.~\eqref{par} we perform the scalar product of $q_i$ along the space generated by the vectors which come next (starting from $q_{i+1}$): this allows to find the projections of $q_i$ along such vectors. The denominator is the necessary normalization which allows to get the correct dimension.
\item[2:] Perpendicular component  $q_{i\perp}$.\\
Introducing polar coordinates and separating the angular part from the radial part, we get
\begin{align}
d^{D-E-L+i}q_{i\perp}&=\Omega_{D-E-L+i-1}|q_{i\perp}|^{D-E-L+i-1}d|q_{i\perp}|\cr
&=\frac{\Omega_{D-E-L+i-1}}{2}|q_{i\perp}|^{D-E-L+i-2}d|q_{i\perp}|^2\,,
\end{align}
where $d|q_{i\perp}|^2$ is $ds_{ii}$. Notice that $|q_{i\perp}|$ can be seen as the height of a parallelotope with base $q_{i+1},\cdots,p_E$, so it can be computed as the whole volume divided by the volume of its base: hence, we write
\begin{equation}
d^{D-E-L+i}q_{i\perp}=\frac{\Omega_{D-E-L+i-1}}{2}\left(\frac{G(q_i,\cdots q_L,p_1,\cdots,p_E)}{G(q_{i+1},\cdots q_L,p_1,\cdots,p_E)}\right)^\frac{D-E-L-2+i}{2}ds_{ii}\,.\label{perp}
\end{equation}
\end{itemize}
Using Eq.~\eqref{par} and Eq.~\eqref{perp} and recalling
\begin{equation}
\frac{\Omega_{D-E-L+i-1}}{2}=\frac{\pi^\frac{D-E-L+i}{2}}{\Gamma(\frac{D-E-L+i}{2})}\,,\label{sphere}
\end{equation}
it is possible to obtain
\begin{align}
\frac{d^Dq_i}{\pi^D/2}&=\frac{\pi^\frac{-E-L+i}{2}}{\Gamma(\frac{D-E-L+i}{2})}\int_{\Gamma_i} G^{-1/2}(q_{i+1},\cdots,p_E)\prod_{j=i}^{E+L}d(q_i k_j)\left(\frac{G(q_i, \cdots,p_E)}{G(q_{i+1},\cdots p_E)}\right)^\frac{D-E-L-2+i}{2}\,,
\end{align}
where $\Gamma_i$ is the contour determined by $|q_{i\perp}|^2>0$. The whole Feynman integral \eqref{feynman} becomes
\begin{align}
I_{\nu_1,\cdots,\nu_N}&=\frac{\pi^{\frac{L-M}{2}}}{\prod_{l=1}^L\Gamma(\frac{D-E-L+l}{2})}\int_{\Gamma}\prod_{i=1}^L G^{-1/2}(q_{i+1},\cdots,p_E)\left(\frac{G(q_i, \cdots,p_E)}{G(q_{i+1},\cdots p_E)}\right)^\frac{D-E-L-2+i}{2}\times\cr
&\times\frac{\prod_{j=i}^{E+L}d(q_ik_j)}{\prod_{a=1}^{N}D_a^{\nu_a}}\,.\label{delirio}
\end{align}
By explicit computation of the product over $i$ in Eq.~\eqref{delirio}, only a few terms survive and Eq.~\eqref{delirio} becomes
\begin{align}
I_{\nu_1,\cdots,\nu_N}&=\frac{\pi^{\frac{L-M}{2}}G(p_1,\cdots,p_E)^{-\frac{L}{2}}}{\prod_{i=1}^L\Gamma(\frac{D-E-L+i}{2})}\int_\Gamma \left(\frac{G(q_1, \cdots,p_E)}{G(p_1,\cdots p_E)}\right)^\frac{D-E-L-1}{2}\frac{\prod_{i=1}^L\prod_{j=1}^{E+L}d(q_ik_j)}{\prod_{a=1}^{N}D_a^{\nu_a}}\nonumber\\
&\equiv C \int_\Gamma \textsl{B}^\gamma\frac{\prod_{i=1}^L\prod_{j=1}^{E+L}d(q_ik_j)}{\prod_{a=1}^{N}D_a^{\nu_a}}\,.\label{notyet}
\end{align}
In the second line of Eq.~\eqref{notyet}, $\textsl{B}$ is the \textsl{Baikov polynomial} (similar to the Jacobian associated to the change of variables) and is raised to the power $\gamma=(D-E-L-1)/2$, while $C$ is an overall constant. Eq.~\eqref{notyet} can be further simplified by noticing that in general the propagators $D_a$ in the denominator usually depend linearly in the scalar products between momenta, hence one can perform a change of variables $z_a=D_a$. Notice however that the number of denominators $N$ and the the number of independent momenta scalar products $M$ are different in general: for this reason it is sufficient to add $N-M$ fake denominators $D_a$, each one raised to a certain power $\nu_a$. The original integral can be recovered by putting $\nu_a=0$ for $a=N+1,\cdots,M$. With this procedure one readily obtains the final form \eqref{baikov}, where the constant $K=C/\det A$, and $A$ is the Jacobian matrix between the scalar products and the propagators $D_a=A_a^{ij}q_ik_j+m_a^2$.


\section{An introduction on Gr\"obner bases}\label{appendixGrobner}

In this Appendix we introduce some basic tools, along with some notations and definitions, necessary to understand how to work with Gr\"obner bases. We start by calling $R=\mathbb{K}(z_1,\cdots,z_n)$ the ring over a field $\mathbb{K}$ with $n$ variables $z_1,\cdots,z_n$, with $I\subset R$ being an ideal in the ring $R$.
\begin{Definition} \textbf{Combination of polynomials}\\
Given a set of polynomials, an expression consisting of the sum of polynomial multiples of the elements in the set is called a combination.
\end{Definition}
For example, $(3z_1+{z_2}^2)(z_1+z_3)+z_1({z_1}^2+5z_2^2z_3^2)$ is a combination of the polynomials $\{ z_1+z_3,z_1^2+5z_2^2z_3^2\} $.
\begin{Definition} \textbf{Power product}\\
A power product is a polynomial which can be obtained only by multiplication of the variables.
\end{Definition}
For example, $z_1z_2^2$ is a power product.
\begin{Definition} \textbf{Term order}\\
A term order on the monomials of a ring $R$ is an order $\prec$ with the following properties:
\begin{itemize}
\item[1:] $M\prec N\Leftrightarrow MP\preccurlyeq NP$ 
\item[2:] $M\preccurlyeq MP$
\end{itemize}
\end{Definition}
For example, the following order is called a \textsl{Lexicographic order} on the variables $z_1,z_2$: 
\begin{equation}
1\prec z_1 \prec z_1^2 \prec z_1^3 \prec \cdots \prec z_2 \prec z_1z_2 \prec z_1^2z_2\prec\cdots z_2^2\prec z_1z_2^2\prec z_1^2z_2^2\prec\cdots
\end{equation}
\begin{Definition} \textbf{Initial term (or leading power product)}\\
Given a ring $R$ with a term order $\prec$, the initial term $in_\succ(f)$ of a polynomial $f\in R$ (or leading power product $LLP_\succ(f))$ is the biggest monomial in $f$ with respect to the given order, together with its coefficient.
\end{Definition}
For example, if $z_1\prec z_2$ and $f=3z_1z_2+2z_2^2$, then $in_\succ(f)=2z_2^2$.
\begin{Definition} \textbf{Initial Ideal}\\
Given an ideal $I\subset R$, $in_{\succ}(I)$ is the monomial ideal generated by the initial terms $\left\lbrace in_{\succ}(f): f\in I \right\rbrace$.
\end{Definition}
\begin{Definition} \textbf{Gr\"obner basis}\label{defgrob}\\
Given an ideal $I\subset R$ with a term order $\prec$ and a set of polynomials $G=\left\lbrace g_1,\cdots,g_s\right\rbrace \subset I$, if $in_{\succ}(I)$ is generated by $\left\lbrace in_{\succ}(g_1),\cdots,in_{\succ}(g_s)\right\rbrace$ then $G$ is called a Gr\"obner basis with respect to the order.
\end{Definition}
Definition \ref{defgrob} implies that, given a $f\in I$ with $f\neq 0$, $in_{\succ}(f)$ must be divisible by at least one of the $in_{\succ}(g_i)$ in the basis.

\

We also introduce some concepts which will be useful in our context. First, we introduce the algorithm of \textsl{division of a polynomial by a given set of polynomials}. This concept arises from the following question: given an ideal $I\subset R$, when is a $h\in R$ also $h\in I$? Suppose we have a Gr\"obner basis for $I$: if $h\in I$, then its leading power product must be divisible by one of the leading power products appearing in the basis (the one belonging to a certain polynomial $g_i$, to fix ideas). The subtraction of an appropriate multiple of $g_i$ from $h$ leads to a new $\tilde{h}$, where the old leading term has been cancelled. If $h\in I$, then also $\tilde{h}\in I$: the process can be reiterated until $0$ is obtained (hence $h\in I$ is verified). If the iterated process leads to a polynomial $\neq 0$ which cannot be further reduced, then $h\notin I$. The result of this procedure of division of a polynomial $f$ by a set of polynomials $G$ is called the \textsl{remainder} $R_G(f)$. Notice that this procedure must be finite, as the definition of a term order $\prec$ implies a descending path in the variables which has to end (this is not true for a generic set of polynomials which is not a Gr\"obner basis for $I$, as one should check infinite types of combinations of polynomials in principle).\\
Notice that the concept of remainder can be used to determine if two polynomials $f,g$ in $R$ are equal modulo $I$ (i.e. $f\sim g$ if $f+$ combination of generators of $I=g)$: in general one should check infinite combinations of the generators, but with a Gr\"obner basis for $I$:
\begin{itemize}
\item $G$ is a Gr\"obner basis for $I\Leftrightarrow R_G(f)=0\,\,\forall\,\,f\in I$;
\item If $G$ is a Gr\"obner basis for $I$, then $f,g\in R$ are equal modulo $I\Leftrightarrow R_G(f)=R_G(g)$.
\end{itemize}
With these tools, it is possible to construct a practical way which tells us how to build Gr\"obner bases.
\begin{Definition}\textbf{S-polynomial}\\
Given two polynomials $f,g$, the S-polynomial between $f$ and $g$ is defined as:
\begin{equation}
S(f,g)=\text{lcm}(\tilde{in}_{\succ}(f),\tilde{in}_{\succ}(g))\left(\frac{f}{in_\succ(f)}-\frac{g}{in_\succ(g)}\right)\,,
\end{equation}
where the notation $\tilde{in}_{\succ}(f)$ means $in_{\succ}(f)$ deprived of its coefficient.
\end{Definition}
Notice that the S-polynomial of two polynomials $f,g$ is a combination of $f$ and $g$: because of the previous observations, if a set $G=\left\lbrace g_1,\cdots,g_s\right\rbrace$ is Gr\"obner basis, then $\forall\,i\neq j$ $R_G(S(g_i,g_j))=0$. This requirement is 
the fundament of \textsl{Buchberger's Algorithm} \cite{Buchberger1965}, which, starting from a certain set of polynomials $\left\lbrace f_1,\cdots,f_t\right\rbrace$ which generate an ideal $I$, returns a Gr\"obner basis for the same ideal, fixed a certain order. 
The algorithm works as follows: given $F=\left\lbrace f_1,\cdots,f_t\right\rbrace$, it computes $S(f_1,f_2)$; if $R_F(S(f_1,f_2))=0$, it proceeds to the next couple of elements ($f_1$ and $f_3$, for instance); if $R_F(S(f_1,f_2))\neq 0$, $F$ is not a Gr\"obner 
basis: it then adds $S(f_1,f_2)$ - or better yet, its remainder $R_F(S(f_1,f_2))$ - as a new element of $F$ and starts again. Notice that the addition of the new term to $F$ ensures that in the next iteration $R_F(S(f_1,f_2)= 0$. At the end of the process, $F$ 
forms a Gr\"obner basis.

\bibliographystyle{unsrt}


\begin{thebibliography}{100}

\bibitem{Veltman:1994wz}
M.~J.~G. Veltman.
\newblock {\em {Diagrammatica: The Path to Feynman rules}}, volume~4.
\newblock Cambridge University Press, 5 2012.

\bibitem{tHooft:1973wag}
Gerard 't~Hooft and M.~J.~G. Veltman.
\newblock {DIAGRAMMAR}.
\newblock {\em NATO Sci. Ser. B}, 4:177--322, 1974.

\bibitem{Bern:1994cg}
Zvi Bern, Lance~J. Dixon, David~C. Dunbar, and David~A. Kosower.
\newblock {Fusing gauge theory tree amplitudes into loop amplitudes}.
\newblock {\em Nucl. Phys.}, B435:59--101, 1995.

\bibitem{Bern:1994zx}
Zvi Bern, Lance~J. Dixon, David~C. Dunbar, and David~A. Kosower.
\newblock {One loop n point gauge theory amplitudes, unitarity and collinear
  limits}.
\newblock {\em Nucl. Phys.}, B425:217--260, 1994.

\bibitem{Britto:2004nc}
Ruth Britto, Freddy Cachazo, and Bo~Feng.
\newblock {Generalized unitarity and one-loop amplitudes in N=4
  super-Yang-Mills}.
\newblock {\em Nucl. Phys. B}, 725:275--305, 2005.

\bibitem{Britto:2005ha}
Ruth Britto, Evgeny Buchbinder, Freddy Cachazo, and Bo~Feng.
\newblock {One-loop amplitudes of gluons in SQCD}.
\newblock {\em Phys. Rev. D}, 72:065012, 2005.

\bibitem{Cachazo:2013gna}
Freddy Cachazo, Song He, and Ellis~Ye Yuan.
\newblock {Scattering equations and Kawai-Lewellen-Tye orthogonality}.
\newblock {\em Phys. Rev.}, D90(6):065001, 2014.

\bibitem{Bern:2019prr}
Zvi Bern, John~Joseph Carrasco, Marco Chiodaroli, Henrik Johansson, and Radu
  Roiban.
\newblock {The Duality Between Color and Kinematics and its Applications}.
\newblock 9 2019.

\bibitem{Ossola:2006us}
Giovanni Ossola, Costas~G. Papadopoulos, and Roberto Pittau.
\newblock {Reducing full one-loop amplitudes to scalar integrals at the
  integrand level}.
\newblock {\em Nucl. Phys. B}, 763:147--169, 2007.

\bibitem{Ellis:2007br}
R.Keith Ellis, W.T. Giele, and Z.~Kunszt.
\newblock {A Numerical Unitarity Formalism for Evaluating One-Loop Amplitudes}.
\newblock {\em JHEP}, 03:003, 2008.

\bibitem{Ellis:2008ir}
R.Keith Ellis, Walter~T. Giele, Zoltan Kunszt, and Kirill Melnikov.
\newblock {Masses, fermions and generalized $D$-dimensional unitarity}.
\newblock {\em Nucl. Phys. B}, 822:270--282, 2009.

\bibitem{Mastrolia:2012bu}
Pierpaolo Mastrolia, Edoardo Mirabella, and Tiziano Peraro.
\newblock {Integrand reduction of one-loop scattering amplitudes through
  Laurent series expansion}.
\newblock {\em JHEP}, 06:095, 2012.
\newblock [Erratum: JHEP 11, 128 (2012)].

\bibitem{Zhang:2012ce}
Yang Zhang.
\newblock {Integrand-Level Reduction of Loop Amplitudes by Computational
  Algebraic Geometry Methods}.
\newblock {\em JHEP}, 09:042, 2012.

\bibitem{Mastrolia:2012an}
Pierpaolo Mastrolia, Edoardo Mirabella, Giovanni Ossola, and Tiziano Peraro.
\newblock {Scattering Amplitudes from Multivariate Polynomial Division}.
\newblock {\em Phys. Lett. B}, 718:173--177, 2012.

\bibitem{Mastrolia:2011pr}
Pierpaolo Mastrolia and Giovanni Ossola.
\newblock {On the Integrand-Reduction Method for Two-Loop Scattering
  Amplitudes}.
\newblock {\em JHEP}, 11:014, 2011.

\bibitem{Badger:2013gxa}
Simon Badger, Hjalte Frellesvig, and Yang Zhang.
\newblock {A Two-Loop Five-Gluon Helicity Amplitude in QCD}.
\newblock {\em JHEP}, 12:045, 2013.

\bibitem{Chetyrkin:1981qh}
K.~G. Chetyrkin and F.~V. Tkachov.
\newblock {Integration by Parts: The Algorithm to Calculate beta Functions in 4
  Loops}.
\newblock {\em Nucl. Phys.}, B192:159--204, 1981.

\bibitem{Laporta:1996mq}
S.~Laporta and E.~Remiddi.
\newblock {The Analytical value of the electron (g-2) at order alpha**3 in
  QED}.
\newblock {\em Phys. Lett.}, B379:283--291, 1996.

\bibitem{Barucchi:1973zm}
G.~Barucchi and G.~Ponzano.
\newblock {Differential equations for one-loop generalized feynman integrals}.
\newblock {\em J. Math. Phys.}, 14:396--401, 1973.

\bibitem{KOTIKOV1991158}
A.V. Kotikov.
\newblock {Differential equations method. New technique for massive Feynman
  diagram calculation}.
\newblock {\em Physics Letters B}, 254(1):158 -- 164, 1991.

\bibitem{Bern:1993kr}
Zvi Bern, Lance~J. Dixon, and David~A. Kosower.
\newblock {Dimensionally regulated pentagon integrals}.
\newblock {\em Nucl. Phys.}, B412:751--816, 1994.

\bibitem{Remiddi:1997ny}
Ettore Remiddi.
\newblock {Differential equations for Feynman graph amplitudes}.
\newblock {\em Nuovo Cim.}, A110:1435--1452, 1997.

\bibitem{Gehrmann:1999as}
T.~Gehrmann and E.~Remiddi.
\newblock {Differential equations for two loop four point functions}.
\newblock {\em Nucl. Phys.}, B580:485--518, 2000.

\bibitem{Henn:2013pwa}
Johannes~M. Henn.
\newblock {Multiloop integrals in dimensional regularization made simple}.
\newblock {\em Phys. Rev. Lett.}, 110:251601, 2013.

\bibitem{Argeri:2014qva}
Mario Argeri, Stefano Di~Vita, Pierpaolo Mastrolia, Edoardo Mirabella, Johannes
  Schlenk, Ulrich Schubert, and Lorenzo Tancredi.
\newblock {Magnus and Dyson Series for Master Integrals}.
\newblock {\em JHEP}, 03:082, 2014.

\bibitem{Adams:2017tga}
Luise Adams, Ekta Chaubey, and Stefan Weinzierl.
\newblock {Simplifying Differential Equations for Multiscale Feynman Integrals
  beyond Multiple Polylogarithms}.
\newblock {\em Phys. Rev. Lett.}, 118(14):141602, 2017.

\bibitem{Laporta:2001dd}
S.~Laporta.
\newblock {High precision calculation of multiloop Feynman integrals by
  difference equations}.
\newblock {\em Int. J. Mod. Phys.}, A15:5087--5159, 2000.

\bibitem{Laporta:2003jz}
S.~Laporta.
\newblock {Calculation of Feynman integrals by difference equations}.
\newblock {\em Acta Phys. Polon.}, B34:5323--5334, 2003.

\bibitem{Tarasov:1996br}
O.~V. Tarasov.
\newblock {Connection between Feynman integrals having different values of the
  space-time dimension}.
\newblock {\em Phys. Rev.}, D54:6479--6490, 1996.

\bibitem{Lee:2009dh}
R.~N. Lee.
\newblock {Space-time dimensionality D as complex variable: Calculating loop
  integrals using dimensional recurrence relation and analytical properties
  with respect to D}.
\newblock {\em Nucl. Phys.}, B830:474--492, 2010.

\bibitem{Argeri:2007up}
Mario Argeri and Pierpaolo Mastrolia.
\newblock {Feynman Diagrams and Differential Equations}.
\newblock {\em Int. J. Mod. Phys. A}, 22:4375--4436, 2007.

\bibitem{Henn:2014qga}
Johannes~M. Henn.
\newblock {Lectures on differential equations for Feynman integrals}.
\newblock {\em J. Phys.}, A48:153001, 2015.

\bibitem{Kalmykov:2020cqz}
Mikhail Kalmykov, Vladimir Bytev, Bernd~A. Kniehl, Sven-Olaf Moch, Bennie F.~L.
  Ward, and Scott~A. Yost.
\newblock {Hypergeometric Functions and Feynman Diagrams}.
\newblock In {\em {Antidifferentiation and the Calculation of Feynman
  Amplitudes}}, 12 2020.

\bibitem{Chen:1977oja}
Kuo-Tsai Chen.
\newblock {Iterated path integrals}.
\newblock {\em Bull. Am. Math. Soc.}, 83:831--879, 1977.

\bibitem{Goncharov:1998kja}
Alexander~B. Goncharov.
\newblock {Multiple polylogarithms, cyclotomy and modular complexes}.
\newblock {\em Math. Res. Lett.}, 5:497--516, 1998.

\bibitem{Remiddi:1999ew}
E.~Remiddi and J.~A.~M. Vermaseren.
\newblock {Harmonic polylogarithms}.
\newblock {\em Int. J. Mod. Phys.}, A15:725--754, 2000.

\bibitem{Broadhurst:1996kc}
David~J. Broadhurst and D.~Kreimer.
\newblock {Association of multiple zeta values with positive knots via Feynman
  diagrams up to 9 loops}.
\newblock {\em Phys. Lett. B}, 393:403--412, 1997.

\bibitem{Broadhurst:2000em}
David~J. Broadhurst and D.~Kreimer.
\newblock {Towards cohomology of renormalization: Bigrading the combinatorial
  Hopf algebra of rooted trees}.
\newblock {\em Commun. Math. Phys.}, 215:217--236, 2000.

\bibitem{Bloch:2005bh}
Spencer Bloch, Helene Esnault, and Dirk Kreimer.
\newblock {On Motives associated to graph polynomials}.
\newblock {\em Commun. Math. Phys.}, 267:181--225, 2006.

\bibitem{Bogner:2007mn}
Christian Bogner and Stefan Weinzierl.
\newblock {Periods and Feynman integrals}.
\newblock {\em J. Math. Phys.}, 50:042302, 2009.

\bibitem{Brown:2009ta}
Francis C.~S. Brown.
\newblock {On the periods of some Feynman integrals}.
\newblock 10 2009.

\bibitem{Marcolli:2009zy}
Matilde Marcolli.
\newblock {Feynman integrals and motives}.
\newblock 7 2009.

\bibitem{Arkani-Hamed:2016byb}
Nima Arkani-Hamed, Jacob~L. Bourjaily, Freddy Cachazo, Alexander~B. Goncharov,
  Alexander Postnikov, and Jaroslav Trnka.
\newblock {\em {Grassmannian Geometry of Scattering Amplitudes}}.
\newblock Cambridge University Press, 4 2016.

\bibitem{Arkani-Hamed:2017tmz}
Nima Arkani-Hamed, Yuntao Bai, and Thomas Lam.
\newblock {Positive Geometries and Canonical Forms}.
\newblock {\em JHEP}, 11:039, 2017.

\bibitem{Mizera:2017cqs}
Sebastian Mizera.
\newblock {Combinatorics and Topology of Kawai-Lewellen-Tye Relations}.
\newblock {\em JHEP}, 08:097, 2017.

\bibitem{Mizera:2017rqa}
Sebastian Mizera.
\newblock {Scattering Amplitudes from Intersection Theory}.
\newblock {\em Phys. Rev. Lett.}, 120(14):141602, 2018.

\bibitem{Broedel:2018qkq}
Johannes Broedel, Claude Duhr, Falko Dulat, Brenda Penante, and Lorenzo
  Tancredi.
\newblock {Elliptic Feynman integrals and pure functions}.
\newblock {\em JHEP}, 01:023, 2019.

\bibitem{Henn:2020lye}
Johannes Henn, Bernhard Mistlberger, Vladimir~A. Smirnov, and Pascal Wasser.
\newblock {Constructing d-log integrands and computing master integrals for
  three-loop four-particle scattering}.
\newblock {\em JHEP}, 04:167, 2020.

\bibitem{Sturmfels:2020mpv}
Bernd Sturmfels and Simon Telen.
\newblock {Likelihood Equations and Scattering Amplitudes}.
\newblock 12 2020.

\bibitem{hwa1966homology}
R.C. Hwa and V.L. Teplitz.
\newblock {\em Homology and Feynman Integrals}.
\newblock Mathematical Physics Monograph Series. W. A. Benjamin, 1966.

\bibitem{Pham:1965zz}
Frederic Pham.
\newblock {INTRODUCTION TO THE TOPOLOGICAL STUDY OF LANDAU SINGULARITIES}.
\newblock 10 1965.

\bibitem{Lefschetz:1975ta}
S.~Lefschetz.
\newblock {\em {Applications of Algebraic Topology. Graphs and Networks. the
  Picard-Lefschetz Theory and Feynman Integrals}}.
\newblock 1975.

\bibitem{Lee:2013hzt}
Roman~N. Lee and Andrei~A. Pomeransky.
\newblock {Critical points and number of master integrals}.
\newblock {\em JHEP}, 11:165, 2013.

\bibitem{Broadhurst:2016hbq}
David Broadhurst and Anton Mellit.
\newblock {Perturbative quantum field theory informs algebraic geometry}.
\newblock {\em PoS}, LL2016:079, 2016.

\bibitem{Broadhurst:2016myo}
David Broadhurst.
\newblock {Feynman integrals, L-series and Kloosterman moments}.
\newblock {\em Commun. Num. Theor. Phys.}, 10:527--569, 2016.

\bibitem{Broadhurst:2018tey}
David Broadhurst and David~P. Roberts.
\newblock {Quadratic relations between Feynman integrals}.
\newblock {\em PoS}, LL2018:053, 2018.

\bibitem{Abreu:2019wzk}
Samuel Abreu, Ruth Britto, Claude Duhr, Einan Gardi, and James Matthew.
\newblock {From positive geometries to a coaction on hypergeometric functions}.
\newblock {\em JHEP}, 02:122, 2020.

\bibitem{Abreu:2019xep}
Samuel Abreu, Ruth Britto, Claude Duhr, Einan Gardi, and James Matthew.
\newblock {Generalized hypergeometric functions and intersection theory for
  Feynman integrals}.
\newblock {\em PoS}, (RACOR2019):067, 2019.

\bibitem{Mastrolia:2018uzb}
Pierpaolo Mastrolia and Sebastian Mizera.
\newblock {Feynman Integrals and Intersection Theory}.
\newblock {\em JHEP}, 02:139, 2019.

\bibitem{Mizera:2019gea}
Sebastian Mizera.
\newblock {Aspects of Scattering Amplitudes and Moduli Space Localization}.
\newblock 2019.

\bibitem{Frellesvig:2019kgj}
Hjalte Frellesvig, Federico Gasparotto, Stefano Laporta, Manoj~K. Mandal,
  Pierpaolo Mastrolia, Luca Mattiazzi, and Sebastian Mizera.
\newblock {Decomposition of Feynman Integrals on the Maximal Cut by
  Intersection Numbers}.
\newblock {\em JHEP}, 05:153, 2019.

\bibitem{Mizera:2019vvs}
Sebastian Mizera and Andrzej Pokraka.
\newblock {From Infinity to Four Dimensions: Higher Residue Pairings and
  Feynman Integrals}.
\newblock 2019.

\bibitem{Frellesvig:2019uqt}
Hjalte Frellesvig, Federico Gasparotto, Manoj~K. Mandal, Pierpaolo Mastrolia,
  Luca Mattiazzi, and Sebastian Mizera.
\newblock {Vector Space of Feynman Integrals and Multivariate Intersection
  Numbers}.
\newblock {\em Phys. Rev. Lett.}, 123(20):201602, 2019.

\bibitem{Mizera:2020wdt}
Sebastian Mizera.
\newblock {Status of Intersection Theory and Feynman Integrals}.
\newblock 2 2020.

\bibitem{Frellesvig:2020qot}
Hjalte Frellesvig, Federico Gasparotto, Stefano Laporta, Manoj~K. Mandal,
  Pierpaolo Mastrolia, Luca Mattiazzi, and Sebastian Mizera.
\newblock {Decomposition of Feynman Integrals by Multivariate Intersection
  Numbers}.
\newblock 8 2020.

\bibitem{Weinzierl:2020xyy}
Stefan Weinzierl.
\newblock {On the computation of intersection numbers for twisted cocycles}.
\newblock 2020.

\bibitem{Kaderli:2019dny}
André Kaderli.
\newblock {A note on the Drinfeld associator for genus-zero superstring
  amplitudes in twisted de Rham theory}.
\newblock 12 2019.

\bibitem{Kalyanapuram:2020vil}
Nikhil Kalyanapuram and Raghav~G. Jha.
\newblock {Positive Geometries for all Scalar Theories from Twisted
  Intersection Theory}.
\newblock {\em Phys. Rev. Res.}, 2(3):033119, 2020.

\bibitem{Weinzierl:2020nhw}
Stefan Weinzierl.
\newblock {Correlation functions on the lattice and twisted cocycles}.
\newblock {\em Phys. Lett. B}, 805:135449, 2020.

\bibitem{fresn2020quadratic}
Javier Fresán, Claude Sabbah, and Jeng-Daw Yu.
\newblock Quadratic relations between periods of connections.
\newblock 2020.

\bibitem{fresn2020quadratic2}
Javier Fresán, Claude Sabbah, and Jeng-Daw Yu.
\newblock Quadratic relations between bessel moments.
\newblock 2020.

\bibitem{Chen:2020uyk}
Jiaqi Chen, Xiaofeng Xu, and Li~Lin Yang.
\newblock {Constructing Canonical Feynman Integrals with Intersection Theory}.
\newblock 8 2020.

\bibitem{Britto:2021prf}
Ruth Britto, Sebastian Mizera, Carlos Rodriguez, and Oliver Schlotterer.
\newblock {Coaction and double-copy properties of configuration-space integrals
  at genus zero}.
\newblock 2 2021.

\bibitem{cho1995}
Koji Cho and Keiji Matsumoto.
\newblock {Intersection theory for twisted cohomologies and twisted Riemann's
  period relations I}.
\newblock {\em Nagoya Math. J.}, 139:67--86, 1995.

\bibitem{matsumoto1994}
Keiji Matsumoto.
\newblock {Quadratic Identities for Hypergeometric Series of Type $(k,l)$}.
\newblock {\em Kyushu Journal of Mathematics}, 48(2):335--345, 1994.

\bibitem{matsumoto1998}
Keiji Matsumoto.
\newblock Intersection numbers for logarithmic $k$-forms.
\newblock {\em Osaka J. Math.}, 35(4):873--893, 1998.

\bibitem{OST2003}
Katsuyoshi Ohara, Yuichi Sugiki, and Nobuki Takayama.
\newblock {Quadratic Relations for Generalized Hypergeometric Functions $_p
  F_{p-1}$}.
\newblock {\em Funkcialaj Ekvacioj}, 46(2):213--251, 2003.

\bibitem{doi:10.1142/S0129167X13500948}
Yoshiaki Goto.
\newblock {Twisted Cycles and Twisted Period Relations for Lauricella's
  Hypergeometric Function $F_C$}.
\newblock {\em International Journal of Mathematics}, 24(12):1350094, 2013.

\bibitem{goto2015}
Yoshiaki Goto and Keiji Matsumoto.
\newblock {The monodromy representation and twisted period relations for
  Appell’s hypergeometric function $F_{4}$}.
\newblock {\em Nagoya Math. J.}, 217:61--94, 03 2015.

\bibitem{goto2015b}
Yoshiaki Goto.
\newblock {Twisted period relations for Lauricella's hypergeometric functions
  $F_{A}$}.
\newblock {\em Osaka J. Math.}, 52(3):861--879, 07 2015.

\bibitem{Yoshiaki-GOTO2015203}
Yoshiaki Goto.
\newblock {Intersection Numbers and Twisted Period Relations for the
  Generalized Hypergeometric Function $_{m+1}F_{m}$}.
\newblock {\em Kyushu Journal of Mathematics}, 69(1):203--217, 2015.

\bibitem{matsubaraheo2019algorithm}
Saiei-Jaeyeong Matsubara-Heo and Nobuki Takayama.
\newblock An algorithm of computing cohomology intersection number of
  hypergeometric integrals.
\newblock 2019.

\bibitem{Smirnov:2010hn}
A.~V. Smirnov and A.~V. Petukhov.
\newblock {The Number of Master Integrals is Finite}.
\newblock {\em Lett. Math. Phys.}, 97:37--44, 2011.

\bibitem{Aluffi2009FeynmanMA}
Paolo Aluffi and Matilde Marcolli.
\newblock Feynman motives and deletion-contraction relations.
\newblock {\em arXiv: Mathematical Physics}, 2009.

\bibitem{Aluffi:2011}
Paolo Aluffi.
\newblock {\em GENERALIZED EULER CHARACTERISTICS, GRAPH HYPERSURFACES, AND
  FEYNMAN PERIODS}, pages 95--136.

\bibitem{Bitoun:2017nre}
Thomas Bitoun, Christian Bogner, Rene~Pascal Klausen, and Erik Panzer.
\newblock {Feynman integral relations from parametric annihilators}.
\newblock {\em Lett. Math. Phys.}, 109(3):497--564, 2019.

\bibitem{Bitoun:2018afx}
Thomas Bitoun, Christian Bogner, Ren\'e~Pascal Klausen, and Erik Panzer.
\newblock {The number of master integrals as Euler characteristic}.
\newblock {\em PoS}, LL2018:065, 2018.

\bibitem{Zhou:2017vhw}
Yajun Zhou.
\newblock {Wick rotations, Eichler integrals, and multi-loop Feynman diagrams}.
\newblock {\em Commun. Num. Theor. Phys.}, 12:127--192, 2018.

\bibitem{Zhou:2017jnm}
Yajun Zhou.
\newblock {Wro\'nskian factorizations and Broadhurst--Mellit determinant
  formulae}.
\newblock {\em Commun. Num. Theor. Phys.}, 12:355--407, 2018.

\bibitem{Lee:2018jsw}
Roman~N. Lee.
\newblock {Symmetric $\epsilon$- and $(\epsilon+1/2)$-forms and quadratic
  constraints in ''elliptic'' sectors}.
\newblock {\em JHEP}, 10:176, 2018.

\bibitem{Baikov:1996iu}
P.~A. Baikov.
\newblock {Explicit solutions of the multiloop integral recurrence relations
  and its application}.
\newblock {\em Nucl. Instrum. Meth.}, A389:347--349, 1997.

\bibitem{Grozin:2011mt}
A.~G. Grozin.
\newblock {Integration by parts: An Introduction}.
\newblock {\em Int. J. Mod. Phys.}, A26:2807--2854, 2011.

\bibitem{Bogner:2010kv}
Christian Bogner and Stefan Weinzierl.
\newblock {Feynman graph polynomials}.
\newblock {\em Int. J. Mod. Phys. A}, 25:2585--2618, 2010.

\bibitem{Mastrolia:2016dhn}
Pierpaolo Mastrolia, Tiziano Peraro, and Amedeo Primo.
\newblock {Adaptive Integrand Decomposition in parallel and orthogonal space}.
\newblock {\em JHEP}, 08:164, 2016.

\bibitem{Lee:2012te}
Roman~N. Lee and Vladimir~A. Smirnov.
\newblock {The Dimensional Recurrence and Analyticity Method for Multicomponent
  Master Integrals: Using Unitarity Cuts to Construct Homogeneous Solutions}.
\newblock {\em JHEP}, 12:104, 2012.

\bibitem{Harley:2017qut}
Mark Harley, Francesco Moriello, and Robert~M. Schabinger.
\newblock {Baikov-Lee Representations Of Cut Feynman Integrals}.
\newblock {\em JHEP}, 06:049, 2017.

\bibitem{matsubaraheo2019euler}
Saiei-Jaeyeong Matsubara-Heo.
\newblock Euler and laplace integral representations of gkz hypergeometric
  functions.
\newblock 2019.

\bibitem{goto2020homology}
Yoshiaki Goto and Saiei-Jaeyeong Matsubara-Heo.
\newblock Homology and cohomology intersection numbers of gkz systems.
\newblock 2020.

\bibitem{Acres:2021sss}
Kevin Acres and David Broadhurst.
\newblock {Empirical determinations of Feynman integrals using integer relation
  algorithms}.
\newblock 3 2021.

\bibitem{SPHM_1982___8_A1_0}
Peter Hilton and Jean Pedersen.
\newblock Descartes, euler, poincar\'e, p\'olya and polyhedra.
\newblock {\em S\'eminaire de Philosophie et Math\'ematiques}, (8):1--17, 1982.

\bibitem{milnor2016morse}
J.~Milnor.
\newblock {\em Morse Theory. (AM-51)}.
\newblock Number v. 51 in Annals of Mathematics Studies. Princeton University
  Press, 2016.

\bibitem{dubrovin1984modern}
B.~A. Dubrovin, A.T. Fomenko, and S.P. Novikov.
\newblock {\em Modern geometry--methods and applications: Part 3: Introduction
  to Homology Theory}.
\newblock Springer-Verlag, New York, 1984.

\bibitem{madsen1997from}
I.~H. Madsen and J.~Tornehaven.
\newblock {\em From calculus to cohomology : de Rham cohomology and
  characteristic classes}.
\newblock Cambridge University Press, Cambridge New York, 1997.

\bibitem{alma991027076759703276}
Solomon Lefschetz.
\newblock {\em L'analysis situs et la géométrie algébrique, par S.
  Lefschetz.}
\newblock Collection de monographies sur la théorie des fonctions, pub. sous
  la direction de m. Émile Borel. Gauthier-Villars et cie, Paris, 1924.

\bibitem{hodge1989the}
W.~V.~D. Hodge.
\newblock {\em The theory and applications of harmonic integrals}.
\newblock Cambridge University Press, Cambridge New York, 1989.

\bibitem{Hodge1955IntegralsOT}
W.~Hodge and M.~Atiyah.
\newblock Integrals of the second kind on an algebraic variety.
\newblock {\em Annals of Mathematics}, 62:56, 1955.

\bibitem{MR0268189}
A.~Grothendieck.
\newblock Standard conjectures on algebraic cycles.
\newblock In {\em Algebraic {G}eometry ({I}nternat. {C}olloq., {T}ata {I}nst.
  {F}und. {R}es., {B}ombay, 1968)}, pages 193--199. Oxford Univ. Press, London,
  1969.

\bibitem{whitney1957geometric}
Hassler Whitney.
\newblock {\em Geometric integration theory}.
\newblock Princeton University Press, Princeton, 1957.

\bibitem{Hancock:1958:EI}
Harris Hancock.
\newblock {\em Elliptic Integrals}.
\newblock Dover Publications Inc., New York, 1958.
\newblock Unaltered reprint of original edition published by Wiley, New York,
  1917.

\bibitem{shafar}
Igor~R Shafarevich.
\newblock {\em {Basic algebraic geometry; 3rd ed.}}
\newblock Springer, Berlin, 2013.

\bibitem{chandrasekharan1985elliptic}
Komaravolu Chandrasekharan.
\newblock {\em Elliptic Functions}.
\newblock Springer Berlin Heidelberg, Berlin, Heidelberg, 1985.

\bibitem{MR713258}
Fr\'{e}d\'{e}ric Pham.
\newblock Vanishing homologies and the {$n$} variable saddlepoint method.
\newblock In {\em Singularities, {P}art 2 ({A}rcata, {C}alif., 1981)},
  volume~40 of {\em Proc. Sympos. Pure Math.}, pages 319--333. Amer. Math.
  Soc., Providence, RI, 1983.

\bibitem{zbMATH03962140}
Fr\'ed\'eric {Pham}.
\newblock {La descente des cols par les onglets de Lefschetz, avec vues sur
  Gauss- Manin. (Steepest descent along Lefschetz thimbles, with views on
  Gauss- Manin)}.
\newblock In {\em {Syst\`emes diff\'erentiels et singularit\'es, C.I.R.M.
  (Luminy, France), 27 juin - 9 juillet 1983 (Colloque)}}. 1985.

\bibitem{bredon1993topology}
Glen Bredon.
\newblock {\em Topology and Geometry}.
\newblock Springer New York, New York, NY, 1993.

\bibitem{GeorgesdeRham1931}
Georges de~Rham.
\newblock Sur l'analysis situs des variétés à n dimensions.
\newblock {\em Journal de Mathématiques Pures et Appliquées}, 10:115--200,
  1931.

\bibitem{Weil1952}
André Weil.
\newblock Sur les théorèmes de de rham.
\newblock {\em Commentarii mathematici Helvetici}, 26:119--145, 1952.

\bibitem{10.2307/2372421}
J.~Schwartz.
\newblock De rham's theorem for arbitrary spaces.
\newblock {\em American Journal of Mathematics}, 77(1):29--44, 1955.

\bibitem{samelson1967}
H.~Samelson.
\newblock {\em {On the de Rham's Theorem}}.
\newblock Topology vol. 6. Pergamon Press, 1967.

\bibitem{Ewald2004}
Ewald Christian-Oliver.
\newblock A de rham isomorphism in singular cohomology and stokes theorem for
  stratifolds.
\newblock {\em International Journal of Geometric Methods in Modern Physics},
  2:63--81, 11 2004.

\bibitem{Hutchings2011CupPA}
M.~Hutchings.
\newblock Cup product and intersections.
\newblock \url{https://math.berkeley.edu/~hutching/teach/215b-2011/cup.pdf},
  2011.

\bibitem{maxim2019intersection}
L.~G. Maxim.
\newblock {\em Intersection Homology \& Perverse Sheaves}.
\newblock Springer, Cham, 2019.

\bibitem{decataldo2009the}
Mark~Andrea de~Cataldo and Luca Migliorini.
\newblock The decomposition theorem, perverse sheaves and the topology of
  algebraic maps.
\newblock {\em Bulletin of the American Mathematical Society}, 46:535--633, 10
  2009.

\bibitem{DeligneI}
Pierre Deligne.
\newblock Théorie de hodge : I.
\newblock {\em Actes, Congrès intern. math.}, 1:425--430, 1970.

\bibitem{DeligneII}
Pierre Deligne.
\newblock Théorie de hodge : Ii.
\newblock {\em Publications Mathématiques de l'IHÉS}, 40:5--57, 1971.

\bibitem{DeligneIII}
Pierre Deligne.
\newblock Théorie de hodge : Iii.
\newblock {\em Publications Mathématiques de l'IHÉS}, 44:5--77, 1974.

\bibitem{mccrory1975cone}
Clint McCrory.
\newblock Cone complexes and pl transversality.
\newblock {\em Transactions of the American Mathematical Society},
  207:269--291, 1975.

\bibitem{goresky1}
Mark Goresky and Robert MacPherson.
\newblock Intersection homology theory.
\newblock {\em Topology}, 19(2):135--162, 1980.

\bibitem{goresky2}
Mark Goresky and Robert MacPherson.
\newblock Intersection homology theory, ii.
\newblock {\em Inventiones Mathematicae}, 72:77--129, 02 1983.

\bibitem{goresky3}
Mark Goresky and Robert MacPherson.
\newblock La dualité de poincaré pour les espaces singuliers.
\newblock {\em C.R. Acad. Sci. Paris}, 284:1549--1551, 1977.

\bibitem{goresky4}
Mark Goresky and Robert MacPherson.
\newblock Simplicial intersection homology.
\newblock {\em Inv. Math.}, 84:432–433, 1986.

\bibitem{goresky5}
Mark Goresky and Robert MacPherson.
\newblock Lefschetz fixed point theorem for intersection homology.
\newblock {\em Comment. Math. Helv.}, 60:366–391, 1985.

\bibitem{goresky6}
Mark Goresky and Robert MacPherson.
\newblock On the topology of complex algebraic maps.
\newblock {\em Proc. of Conference on Algebraic Geometry in La Rabida, Spain,
  Springer Lecture Notes in Mathematics}, 961:119–129, 1981.

\bibitem{goresky7}
Mark Goresky and Robert MacPherson.
\newblock {\em Stratified Morse Theory}, volume~14.
\newblock Springer-Verlag, Berlin New York, 1988.

\bibitem{goresky8}
Mark Goresky and Robert MacPherson.
\newblock Stratified morse theory.
\newblock {\em Proc. of Symp. in Pure Math.}, 40:517--533, 1983.

\bibitem{goresky9}
Mark Goresky and Robert MacPherson.
\newblock Morse theory and intersection homology theory.
\newblock In {\em Analyse et topologie sur les espaces singuliers}, number 101
  in Ast\'erisque, page 135–192. Soci\'et\'e math\'ematique de France, 1983.

\bibitem{Williamson}
G.~Williamson.
\newblock An illustrated guide to perverse sheaves, course given in pisa
  19th-30th of january, 2015.
\newblock
  \url{https://www.maths.usyd.edu.au/u/geordie/perverse_course/lectures.pdf}.

\bibitem{aomoto2011theory}
K.~Aomoto and M.~Kita.
\newblock {\em {Theory of Hypergeometric Functions}}.
\newblock Springer Monographs in Mathematics. Springer Japan, 2011.

\bibitem{Moser1959/60}
Jürgen Moser.
\newblock The order of a singularity in fuchs' theory.
\newblock {\em Mathematische Zeitschrift}, 72:379--398, 1959/60.

\bibitem{Lee:2014ioa}
Roman~N. Lee.
\newblock {Reducing differential equations for multiloop master integrals}.
\newblock {\em JHEP}, 04:108, 2015.

\bibitem{alggeom}
P.~Griffiths and J.~Harris.
\newblock {\em Principles of Algebraic Geometry}.
\newblock John Wiley \& Sons, New York, 1994.

\bibitem{Sogaard:2013fpa}
Mads S\o{}gaard and Yang Zhang.
\newblock {Multivariate Residues and Maximal Unitarity}.
\newblock {\em JHEP}, 12:008, 2013.

\bibitem{grobnerbook}
W.~Adams and P.~Loustaunau.
\newblock {\em An Introduction to Gr{\"o}bner Bases (Graduate Studies in
  Mathematics, Vol. 3)}.
\newblock American Mathematical Society, Providence, Rhode Island, 1994.

\bibitem{macaulay2}
D.~R. Grayson and M.~E. Stillman.
\newblock Macaulay2, a software system for research in algebraic geometry.
\newblock \url{https://http://www.math.uiuc.edu/Macaulay2/}.

\bibitem{saitohigher}
Kyoji Saito.
\newblock The higher residue pairings k f (k) for a family of hypersurface
  singular points.
\newblock 40, 01 1983.

\bibitem{Buchberger1965}
B.~Buchberger.
\newblock Ein {A}lgorithmus zum {A}uffinden der {B}asiselemente des
  {R}estklassenringes nach einem nulldimensionalen {P}olynomideal.
\newblock Dissertation an dem Math. Inst. der Universit\"at von Innsbruck,
  1965.

\end{thebibliography}


\end{document}